\documentclass[acmsmall,screen]{acmart}
\pdfoutput=1 
\newif\ifappendix\appendixtrue
\usepackage{pervasives}

\setcopyright{rightsretained}
\acmPrice{}
\acmDOI{10.1145/3360575}
\acmYear{2019}
\copyrightyear{2019}
\acmJournal{PACMPL}
\acmVolume{3}
\acmNumber{OOPSLA}
\acmArticle{149}
\acmMonth{10}


\begin{document}

\title{Formal Foundations of Serverless Computing}

\author{Abhinav Jangda}
\affiliation{\institution{University of Massachusetts Amherst}\country{United States}}
\author{Donald Pinckney}
\affiliation{\institution{University of Massachusetts Amherst}\country{United States}}
\author{Yuriy Brun}
\affiliation{\institution{University of Massachusetts Amherst}\country{United States}}
\author{Arjun Guha}
\affiliation{\institution{University of Massachusetts Amherst}\country{United States}}

\begin{abstract}
\emph{Serverless computing} (also known as \emph{functions as a service}) is a
new cloud computing abstraction that makes it easier to write robust,
large-scale web services. In serverless computing, programmers write what are
called \emph{serverless functions}, which are programs that respond to external events.
When demand for the serverless function spikes, the platform automatically
allocates additional hardware and manages load-balancing; when
demand falls, the platform silently deallocates idle resources; and when the
platform detects a failure, it transparently retries affected requests.
In
2014, Amazon Web Services introduced the first serverless platform, \emph{AWS
Lambda}, and similar abstractions are now available on all major cloud computing
platforms.

Unfortunately, the serverless computing abstraction exposes several low-level
operational details that make it hard for programmers to write and reason about
their code. This paper sheds light on this problem by presenting \pfname,
an operational semantics of the essence of serverless computing. Despite
being a small (half a page) core calculus, \pfname models all the
low-level details that serverless functions can observe. To show that
\pfname is useful, we present three applications. First, to ease
reasoning about code, we present a simplified \emph{\pnname semantics} of
serverless execution and precisely characterize when the \pnname semantics and
\pfname coincide. Second, we augment \pfname with a key-value store to
allow reasoning about stateful serverless functions.
Third, since a handful of serverless platforms support serverless
function composition, we show how to extend \pfname with a composition language
and show that our implementation can outperform prior work.
\end{abstract}


\maketitle

\section{Introduction}

\emph{Serverless computing}, also known as \emph{functions as a service}, is a new
approach to cloud computing that allows programmers to run event-driven
functions in the cloud without the need to manage resource allocation or
configure the runtime environment. Instead, when a programmer deploys a
\emph{serverless function}, the cloud platform automatically manages
dependencies, compiles code, configures the operating system, and manages
resource allocation.
Unlike virtual machines or containers, which require low-level system and
resource management, serverless computing allows programmers to focus entirely
on application code. If demand for the function suddenly increases, the cloud
platform may transparently load another instance of the function on a new
machine and manage load-balancing without programmer intervention.
Conversely, if demand for the function falls, the platform will transparently
terminate underutilized instances. In fact, the cloud provider may
shutdown all running instances of the function if there is no demand for an
extended period of time. Since the platform completely manages the operating
system and resource allocation in this manner, serverless computing is a
language-level abstraction for programming the cloud.

An economic advantage of serverless functions is that they only incur costs for
the time spent processing events. Therefore, a function that is never invoked
incurs no cost. By contrast, virtual machines incur costs when
they are idle, and they need idle capacity to smoothly handle unexpected
increases in demand. Serverless computing allows cloud platforms to
more easily optimize resource allocation across several customers, which increases hardware
utilization and lowers costs, \eg, by lowering energy consumption.

Amazon Web Services introduced the first serverless computing platform,
\emph{AWS Lambda}, in 2014, and similar abstractions are now available from all major
cloud providers~\cite{istemi:sand, hendrickson:openlambda, google-cloud-functions,
azure-functions, openwhisk, openfaas}.
Serverless computing has seen rapid adoption~\cite{conway:cncf-2017-survey},
and programmers now often use serverless computing to write short, event-driven
computations, such as web services, backends for mobile apps, and aggregators
for IoT devices.

In the research community, there is burgeoning interest in developing
new programming abstractions for serverless computing, including
abstractions for big data processing~\cite{jonas:pywren, fouladi:excamera,
ao:sprocket},
modular programming~\cite{baldini:trilemma}, information flow
control~\cite{alpernas:trapeze}, chatbot design~\cite{baudart:chatbots-protection},
and virtual network functions~\cite{singhvi:nf}.
However, serverless computing has several peculiar
characteristics that prior work has not addressed.

\paragraph{Shortcomings of serverless computing.}

The serverless computing abstraction, despite its many advantages, exposes
several low-level operational details that make it hard for programmers to
write and reason about their code. For example, to reduce latency, serverless
platforms try to reuse the same function instance to process multiple
requests. However, this behavior is not transparent, and it is easy to write
a serverless function that produces incorrect results or leaks confidential
data when reused. A related problem is that serverless platforms abruptly
terminate function instances when they are idle, which can lead to data loss
if the programmer is not careful. To tolerate network and system failures,
serverless platforms automatically re-execute functions on different
machines. However, the responsibility falls on the programmer to ensure that their functions
perform correctly when they are re-executed, which may include concurrent
re-execution when a transient failure occurs. These problems are exacerbated
when an application is composed of several functions.

In summary, programmers face three key challenges when writing serverless
functions:

\begin{enumerate}

\item Reasoning about the correctness of serverless functions is hard, because
of the exposed low-level behavior of the underlying serverless platform, which
affects the behavior of programmer-written functions.

\item Interacting with external services, such as cloud-hosted databases, is also
challenging, because failures and retries can be visible to external services,
which increases the likelihood of data errors.

\item Composing and orchestration of serverless functions makes
reasoning even harder because most serverless platforms do not natively
support serverless function composition.

\end{enumerate}

\paragraph{Our contributions.}

This paper presents a formal foundation for serverless computing that makes
progress towards addressing the above three challenges. Based on our
experience with several major serverless computing platforms (Google Cloud
Functions, Apache OpenWhisk, and AWS Lambda), we have developed a detailed,
operational semantics, called \pfname, which
manifests the essential low-level behaviors of these serverless platforms,
including failures, concurrency, function restarts, and instance reuse. The
details of \pfname matter because they are observable by programs, but
programmers find it hard to write code that correctly addresses all of these
behaviors. By elucidating these behaviors, \pfname can guide the development
of programming tools and extensions to serverless platforms.

We design \pfname to simplify reasoning about serverless programs and to be
extensible to aid programmers using the serverless abstraction in more complex
ways. We evaluate \pfname and demonstrate its utility in three ways:

First, reasoning about serverless programs and the complex, low-level
behavior of serverless platforms is hard. We derive a simplified
\emph{\pnname semantics}
of serverless computing, which elides these complex behaviors and helps
programmers reason about their programs.
We precisely characterize when it is safe for a programmer to use the
\pnname semantics instead of \pfname and prove that if a serverless
function satisfies a simple safety property, then there exists a weak
bisimulation between \pfname and the \pnname semantics. This result helps
programmers confidently abstract away the low-level details of serverless
computing. We provide canonical examples of these safety properties and
examples of ill-behaved serverless functions that are unsafe, which programmers
could easily write by mistake without the \pnname semantics.
Our theorem can serve as the foundation for future work on building
serverless functions that are verifiably safe.

Second, we demonstrate that \pfname can compose with models of external
services. External services matter, because the serverless abstraction is quite
limited by itself. However, the interaction between serverless functions and
external services can be quite complicated, due to the low-level behaviors that
serverless functions exhibit. Specifically, we compose \pfname with a model of
a cloud-hosted key-value store with shared state. Using this extension, we
precisely characterize what it means for a serverless function to be
idempotent, which is necessary to avoid corrupting data in the key-value store.
The recipe that we follow to add a key-value store to \pfname could also be
used to augment \pfname with models of other cloud computing services.

Finally, we demonstrate that \pfname can be extended to model richer
serverless programming abstractions, such as  serverless function orchestration. We
extend \pfname to support a \emph{serverless programming language} (\spl),
which has a suite of I/O primitives, data processing operations, and
composition operators that can run safely and efficiently without
the need for operating system isolation mechanisms. To perform operations that
are beyond the scope of \spl, we allow \spl programs to invoke existing
serverless functions as black boxes. We implement \spl, evaluate its
performance, and find that it can outperform a popular alternative in certain
common cases. Using case studies, we show that \spl is expressive and easy to
extend with new features.

We hope that \pfname will be a foundation for further research on
language-based abstractions for serverless computing. The
case studies we present are detailed examples that show how
to design new abstractions and study existing abstractions for serverless
computing, using \pfname.

The rest of this paper is organized as follows.
\Cref{sec:overview} presents an overview of serverless computing, and
a variety of issues that arise when writing serverless code.
\Cref{sec:Semantics} presents \pfname, our formal semantics of serverless
computing. \Cref{safe-semantics} presents a simplified semantics of
serverless computing and proves exactly when it coincides with \pfname.
\cref{sec:Serverless Functions and Cloud Storage} augments \pfname with
a key-value store.
\Cref{spl} and \Cref{sec:eval} extend \pfname with a language
for serverless orchestration. Finally, \Cref{sec:Related Work} discusses related
work and \cref{sec:Contributions} concludes. Our implementation is
available at \url{plasma.cs.umass.edu/lambda-lambda}.

\section{Overview of Serverless Computing}
\label{sec:overview}

\begin{figure}
\lstset{language=JavaScript}
\begin{lstlisting}
let accounts = new Map();
exports.bank = function(req, res) {
  if (req.body.type === 'deposit') {
    accounts.set(req.body.name, req.body.value);
    res.send(true);
  } else if (req.body.type === 'transfer') {
    let { from, to, amnt } = req.body;
    if (accounts.get(from) >= amnt) {
      accounts.set(to, accounts.get(to) + amnt);
      accounts.set(from, accounts.get(from) - amnt);
      res.send(true);
    } else {
      res.send(false); }}};
\end{lstlisting}
\caption{A serverless function for banking that does not address low-level
details of serverless execution. Therefore, it will exhibit several kinds of
faults. The correct implementation is in \cref{working-bank-example}.}
\label{naive-bank-example}
\end{figure}

\begin{figure}

\lstset{language=JavaScript, escapechar=!}
\begin{lstlisting}
let Datastore = require('@google-cloud/datastore');
exports.bank = function(req, res) {
  let ds = new Datastore({ projectId: 'bank-app' });
  let dst = ds.transaction();
  dst.run(function() {
    let tId = ds.key(['Transaction', req.body.transId]); !\label{line:trans-exist-begin}!
    dst.get(tId, function(err, trans) {
      if (err || trans) {
        dst.rollback(function() { res.send(err || trans); }); !\label{line:trans-exist-end}!
      } else if (req.body.type === 'deposit') {
        let to = ds.key(['Account', req.body.to]);
        dst.get(to, function(err, acct) {
          acct.balance += req.body.amount;
          dst.save({ key: to, data: acct });  !\label{line:trans-commit-1-start}!
          dst.save({ key: tId, data: true });
          dst.commit(function() { res.send(true); }); !\label{line:trans-commit-1}!
        });
      } else if (req.body.type === 'transfer') {
        let amnt = req.body.amount;
        let from = ds.key(['Account', req.body.from]);
        let to = ds.key(['Account', req.body.to]);
        dst.get([from, to], function(err, accts) {
          if (accts[0].balance >= amnt) {
            accts[0].balance -= amnt;
            accts[1].balance += amnt;
            dst.save([{ key: from, data: accts[0] }, !\label{line:trans-commit-2-start}!
              { key: to, data: accts[1] }]);
            dst.save({ key: tId, data: true });
            dst.commit(function() { res.send(true); }); !\label{line:trans-commit-2}!
          } else {
            dst.rollback(function() { res.send(false);
          });}});}}});});};
\end{lstlisting}

\caption{A serverless function for banking that addresses instance termination,
concurrency, and idempotence.}
\label{working-bank-example}
\end{figure}

To motivate the need for a formal foundation of serverless computing, consider
the serverless banking function in \cref{naive-bank-example}.\footnote{The examples in
this paper are in JavaScript\,---\,the language that is most widely supported
by serverless platforms\,---\,and are written for Google Cloud Functions.
However, it is easy to port our examples to other languages and serverless
platforms.} This function processes two types of requests: (1)~a request to
deposit new funds into an account and (2)~a request to transform funds from
one account to another, which will fail if the source account has insufficient
funds.

During deployment, the programmer can specify the types of events
that will \emph{trigger} the function, \eg, messages on a message bus, updates
to a database, or web requests to a \textsc{url}. The function receives two
arguments: (1)~the event as a \textsc{json}-formatted object, and (2)~a callback,
which it must call when event processing is complete. The callback allows the
function to run asynchronously, although our example is presently synchronous.
For brevity, we have elided authentication, authorization, error handling, and
most input validation. However, this function suffers several other problems
because it is not properly designed for serverless execution.

\paragraph{Ephemeral state.}

\lstset{language=JavaScript}
The first problem arises after a few minutes of inactivity: all updates to the
\lstinline{accounts} global variable are lost. This problem occurs because the
serverless platform runs the function in an ephemeral container, and silently
shuts down the container when it is idle. The exact timeout depends on overall
load on the cloud provider's infrastructure, and is not known to the programmer.
Moreover, the function does not receive a notification before a shut down
occurs. Similarly, the platform automatically starts a new container when an
event eventually arrives, but all state is lost. Therefore, the function must
serialize all state updates to a persistent store. In serverless environments,
the local disk is also ephemeral, therefore our example function must use a
network-attached database or storage system.

\paragraph{Implicit parallel execution.}

The second problem arises when the function receives multiple events over a
short period of time. When a function is under high load, the serverless
platform transparently starts new \emph{instances} of the function and manages
load-balancing to reduce latency. Each instance runs in isolation (in a
container), there is no guarantee that two events from the same source will
be processed by the same instance, and instances may process events in
parallel. Therefore, a correct implementation must use transactions to
correctly manage this implicit parallelism.

\paragraph{At-least-once execution.}
\looseness-1
The third problem arises when failures occur in the serverless computing
infrastructure. Serverless platforms are distributed systems that are designed
to re-invoke functions when a failure is detected.\footnote{Even if the serverless
platform does not re-invoke functions itself, there are situations where
the external caller will re-invoke it to workaround a transient failure.
For example, this arises when functions are triggered by \textsc{http} requests.}
However, if the failure is
transient, a single event may be processed to completion multiple times.
Most platforms run functions \emph{at least once} in response to a single
event, which can cause problems when functions have side-effects~\cite{azure-messaging,openwhisk-actions,google-cf-exec,amazon-lambda-invoke}.
In our example function, this would duplicate deposits and transfers. We can fix this
problem in several ways. A common approach is to require each request to have a
unique identifier, maintain a consistent log of processed identifiers in the
database, and ignore requests that are already in the log.

\medskip \Cref{working-bank-example} fixes the three problems mentioned above.
This implementation uses a cloud-hosted key-value store for persistence, uses
transactions to address concurrency, and requires each request to have a unique
identifier to support safe re-invocation. This version of the function is
more than twice as long as the original, and requires the programmer to
have a deep understanding of the serverless execution model. The next section
presents an operational semantics of serverless platforms (\pfname) that
succinctly describes its peculiarities. Using \pfname, the rest of the paper
precisely characterizes the kinds of properties that are needed for serverless
functions, such as our example, to be correct.

\section[Lambda Lambda]{Semantics of Serverless Computing}
\label{sec:Semantics}

\begin{figure}[t]
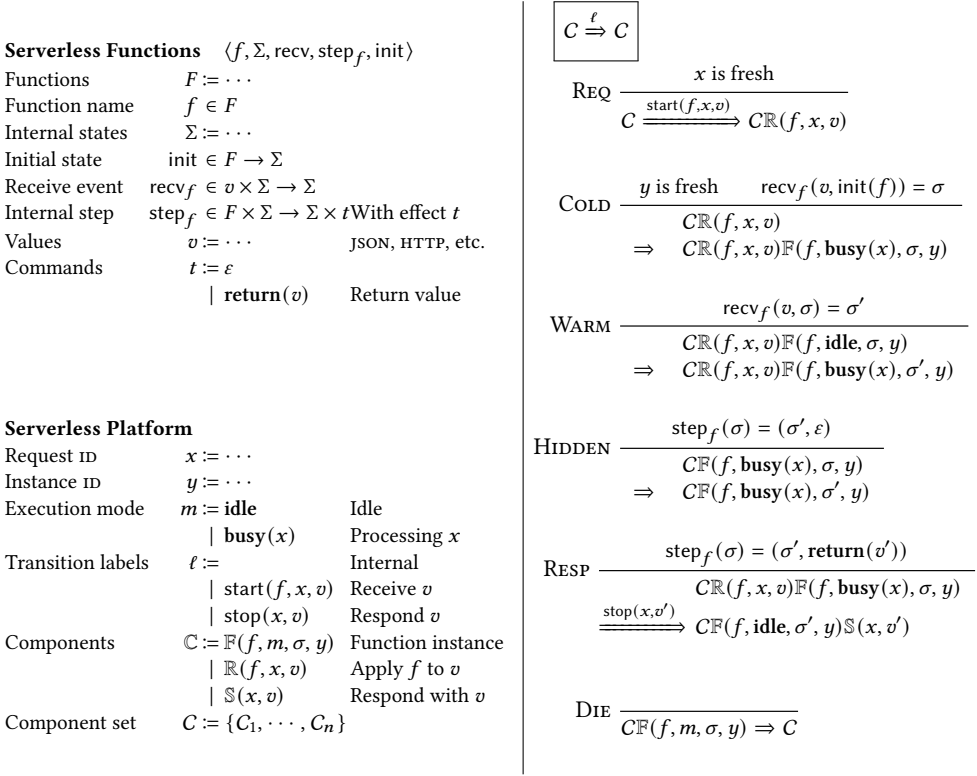

\footnotesize
\begin{subfigure}{0.49\columnwidth}
\(
\begin{array}{@{}l@{}r@{\,}c@{\,}l@{}l}
\multicolumn{5}{@{}l}{\textbf{Serverless Functions}\quad  \langle f, \Sigma, \precvfn, \pstepfn_f, \pinitfn \rangle} \\
\textrm{Functions}        & F        & \pdef & \cdots \\
\textrm{Function name}   	& f        & \in & F \\ 
\textrm{Internal states} 	& \Sigma   & \pdef & \cdots \\
\textrm{Initial state}   	& \pinitfn & \in   & F \rightarrow \Sigma \\
\textrm{Receive event}   	& \precvfn_f & \in   & v \times \Sigma \rightarrow \Sigma \\
\textrm{Internal step}   	& \pstepfn_f & \in   & F \times \Sigma \rightarrow \Sigma \times t & \textrm{With effect $t$}\\
\textrm{Values}          	& v        & \pdef & \cdots                      & \textrm{\textsc{json}, \textsc{http}, etc.} \\
\textrm{Commands}	        & t		   & \pdef & \varepsilon \\
                            &		   & \mid  & \ptreturn{v} & \textrm{Return value} \\[5em]
\multicolumn{5}{@{}l}{\textbf{Serverless Platform}} \\
\textrm{Request \textsc{id}} & x  & \pdef & \cdots \\
\textrm{Instance \textsc{id}}& \pinstid & \pdef & \cdots \\
\textrm{Execution mode}        & m  & \pdef & \pidle                      & \textrm{Idle} \\
                                &    & \mid  & \pbusy(x)                  & \textrm{Processing $x$} \\
\textrm{Transition labels} & \ell & \pdef   &                             & \textrm{Internal} \\
                            &      & \mid    & \kw{start}(f,x,v)            & \textrm{Receive $v$}    \\
                            &      & \mid    & \kw{stop}(x,v)             & \textrm{Respond $v$} \\
\textrm{Components}  & \mathbb{C}      & \pdef & \pinst{f}{m}{\sigma}    & \textrm{Function instance} \\
                     &        & \mid  & \preq{f}{x}{v}           & \textrm{Apply $f$ to $v$} \\
	                 &        & \mid  & \presp{x}{v}           & \textrm{Respond with  $v$} \\

\textrm{Component set} & \pfstatetyp & \pdef & \{ \mathcal{C}_1, \cdots,  \mathcal{C}_n\} \\
\end{array}
\)
\end{subfigure}
\vrule
\quad\,
\begin{subfigure}{0.47\columnwidth}
\(
\begin{array}{@{}lr@{\,}c@{\,}ll}
\multicolumn{5}{@{}l}{\fbox{$\pfstep{\pfstatetyp}{\ell}{\pfstatetyp}$}} \\
\multicolumn{5}{@{\quad\quad\quad}l}{
  \inferrule*[Left=Req]{\textrm{$x$ is fresh}}{\pfstep{\pfstatetyp}{\kw{start}(f,x,v)}{\mathcal{C}\preq{f}{x}{v}}}} \\[1.5em]
\multicolumn{5}{@{\quad\quad\quad}l}{\inferrule*[Left=Cold]
  {\textrm{$\pinstid$ is fresh} \\ \precvfn_f(v, \pinitfn(f)) = \sigma}
  {{\begin{array}{ll}&\pfstatetyp \preq{f}{x}{v}\\
        \pfsteparr{} &\pfstatetyp \preq{f}{x}{v}  \pinst{f}{\pbusy(x)}{\sigma}
\end{array}} }} \\[1.5em]
\multicolumn{5}{@{\quad\quad\quad}l}{\inferrule*[Left=Warm]
  {\precvfn_f(v, \sigma) = \sigma'}
  {{\begin{array}{ll} &\pfstatetyp \preq{f}{x}{v} \pinst{f}{\pidle}{\sigma}\\
         \pfsteparr{} & \pfstatetyp \preq{f}{x}{v} \pinst{f}{\pbusy(x)}{\sigma'}
    \end{array}} }} \\[1.5em]
\multicolumn{5}{@{\quad\quad\quad}l}{\inferrule*[Left=Hidden]
  {\pstepfn_f(\sigma) = (\sigma', \varepsilon)}
  {{\begin{array}{ll}&  \pfstatetyp \pinst{f}{\pbusy(x)}{\sigma} \\
         \pfsteparr{} & \pfstatetyp \pinst{f}{\pbusy(x)}{\sigma'}
    \end{array}} }} \\[1.5em]
\multicolumn{5}{@{\quad\quad}l}{\inferrule*[Left=Resp]
	{\pstepfn_f(\sigma) = (\sigma', \ptreturn{v'})}
  {{\begin{array}{@{}l@{\hspace{3pt}}l}&\pfstatetyp \preq{f}{x}{v} \pinst{f}{\pbusy(x)}{\sigma}\\
    \pfsteparr{\kw{stop}(x,v')}&
    \pfstatetyp \pinst{f}{\pidle}{\sigma'}\presp{x}{v'}
  \end{array}} }} \\[1.5em]
\multicolumn{5}{@{\quad\quad\quad}l}{\inferrule*[Left=Die]
  {\phantom{.}}
  {\pfstep{\pfstatetyp \pinst{f}{m}{\sigma}}{}{\pfstatetyp}}} \\[1.5em]
\end{array}
\)
\end{subfigure}
\caption[]{\pfname: An operational model of serverless platforms.}
\label{function-semantics}
\end{figure}

We now present our operational semantics of serverless platforms (\pfname) that
captures the essential details of the serverless execution model, including
concurrency, failures, execution retries, and function instance reuse.
\Cref{function-semantics} presents \pfname
in full, and is divided into two parts: (1)~a model of
serverless functions, and (2)~a model of a serverless platform
that receives requests (\ie, events) and produces responses by running serverless
functions. The peculiar features of serverless computing are captured by the
latter part.

\paragraph{Serverless functions.}

Serverless platforms allow programmers to write serverless functions in a variety of
source languages, but platforms themselves are source-language agnostic. Most
platforms only require that serverless functions operate asynchronously,
process a single request at a time, and usually consume and produce
\textsc{json} values. These are the features that our platform model includes.
In our model, the operation of a serverless function ($f$) is defined by three
functions:
\begin{enumerate}

  \item A function that produces the initial state of a function (\pinitfn).
   The state of a serverless function ($\sigma$) is abstract to the
   serverless platform, and in practice the state is source language dependent.
   For example, if the serverless function $f$ were written in JavaScript,
   then each state would be the state of the JavaScript VM and
   $\pinitfn(f)$ would be the initial JavaScript heap.

  \item A function that receives a request ($\precvfn_f$) from the platform.

  \item A function that takes an internal step ($\pstepfn_f$) that may produce
  a command for the serverless platform ($t$).
  For now, the only valid command is $\ptreturn{v}$, which
  indicates that the response to the last request is the value $v$.
  \Cref{sec:Serverless Functions and Cloud Storage} extends our model with new
  commands.

\end{enumerate}

\paragraph{Serverless platform.}

\pfname is an operational semantics of a serverless platform that
processes several concurrent requests. \pfname is written in a process-calculus
style, where the state of the platform consists of a collection of running or
idle functions (known as \emph{function instances}), pending requests, and
responses. A new request may arrive at any time for a serverless
function $f$, and each request is given a globally unique identifier $x$ (the
\textsc{Req} rule). However, the platform does not process requests
immediately. Instead, at some later step, the platform may either
\emph{cold-start} a new instance of $f$ (the \textsc{Cold} rule), or may
\emph{warm-start} by processing the request on an existing, idle instance of
$f$ (the \textsc{Warm} rule). The internal steps of a function instance are
unobservable (the \textsc{Hidden} rule). The only observable command that an
instance can produce is to respond to a pending request (the \textsc{Resp}
rule). When an instance responds to a request, \pfname produces a response
object, an observable response event, and  marks the function instance as idle, which allows it to be
reused. Finally, a function instance may die at any time without notification
(the \textsc{Die} rule). These rules are sufficient to capture several
subtleties of serverless execution, as discussed below.

\paragraph{Instance launches are not observable}

\pfname produces an observable event ($\kw{start}(f,x,v)$) when it
receives a request (\textsc{Req}), and not when it starts to process the
request. This is necessary because the platform may start several instances for
a single event $x$, for example, if the platform detects a potential
failure.

\paragraph{State reuse during warm starts}

When a function instance responds to a request (the \textsc{Resp} rule), the
instance becomes idle, but its state is not reinitialized and may be reused to
process another request (the \textsc{Warm} rule). In the following example, the
function instance receives the second request ($x_2$) when its state is
$\sigma_1$, which may not be identical to the initial state of the function.
\[
\begin{array}{lll}
                                & \preq{f}{x_1}{v_1} \pinst{f}{\pbusy(x_1)}{\sigma_0} \\
\pfsteparr{\kw{stop}(x_1,v_1')} & \pinst{f}{\pidle}{\sigma_1}                         & \textrm{By \textsc{Resp}} \\
\pfsteparr{\kw{start}(f,x_2,v_2)} & \preq{f}{x_2}{v_2} \pinst{f}{\pidle}{\sigma_1}      & \textrm{By \textsc{Req}} \\
\pfsteparr{\phantom{XXXXXX}} & \preq{f}{x_2}{v_2} \pinst{f}{\pbusy(x_2)}{\sigma_2} & \textrm{By \textsc{Warm}}
\end{array}
\]

\paragraph{Function instance termination is not observable}

A function instance may terminate at any time. Moreover,
 termination is not an observable event. In practice, there are several reasons
 why termination may occur. (1)~An instance may terminate if there is a
 software or hardware failure on its machine. (2)~The platform may
 deliberately terminate the instance to reclaim idle resources. (3)~The
 platform may deliberately terminate an instance if it takes too long to
 respond to a request. In \pfname, we model all kinds of termination
 with the \textsc{Die} rule.

\paragraph{Function instances may start at any time}

The platform is free to cold-start or warm-start a function instance for any
pending request at any time, even if an existing function instance is processing
the request. Therefore, several function instances may be processing a single
request at once. This occurs in practice when a transient fault makes an
instance temporarily unreachable. However, cold-starts and warm-starts are
not observable events, thus programmers cannot directly observe the number
of instances that are processing a single request. In the example below,
a single request cold-starts two function instances.
\[
\begin{array}{lll}
\pfsteparr{\kw{start}(f,x,v)} & \preq{f}{x}{v}                                                                     & \textrm{By \textsc{Req}} \\
\pfsteparr{\phantom{XXXXX}} & \preq{f}{x}{v} \pinsty{f}{\pbusy(x)}{\pinitfn(f)}{y_1}                                   & \textrm{By \textsc{Cold-Start}} \\
\pfsteparr{\phantom{XXXXX}} & \preq{f}{x}{v} \pinsty{f}{\pbusy(x)}{\pinitfn(f)}{y_1} \pinsty{f}{\pbusy(x)}{\pinitfn(f)}{y_2} & \textrm{By \textsc{Cold-Start}}
\end{array}
\]
This example also shows why the request $\preq{f}{x}{v}$ is not consumed after
the first instance is launched. We need may need the request to launch additional
instances in the future, particular when a failure occurs.

\paragraph{Single response per request}

Although a single request may spawn several function instances, each request
receives one response from a single function instance. Other instances
processing the same request will eventually get stuck because they cannot
respond. However, stuck instances will eventually terminate. In the following
example, two instances start by processing the same request, the first instance
then responds and becomes idle, and finally, the second instance terminates
because it is stuck.
\[
\begin{array}{lll}
                           & \preq{f}{x}{v} \pinsty{f}{\pbusy(x)}{\sigma_1}{y_1} \pinsty{f}{\pbusy(x)}{\sigma_1}{y_2} \\
\pfsteparr{\kw{stop}(x,v')} & \pinsty{f}{\pidle}{\sigma_1'}{y_1} \pinsty{f}{\pbusy(x)}{\sigma_1}{y_2}  & \textrm{By \textsc{Resp}} \\
\pfsteparr{\phantom{XXXXX}} & \pinsty{f}{\pidle}{\sigma_1'}{y_1}  & \textrm{By \textsc{Die}} \\
\end{array}
\]

\paragraph{Summary}

In summary, \pfname succinctly and faithfully models the low-level details of
serverless platforms, and makes manifest the subtleties that make serverless
programming hard. The rest of this paper demonstrates that \pfname is useful in
a variety of ways. The next section shows how to use \pfname to rigorously
define a simpler semantics of serverless programming that is easier for
programmers to understand, \cref{sec:Serverless Functions and Cloud Storage}
shows that \pfname is easy to extend with a model of another cloud service,
and  \Cref{spl} and \Cref{sec:eval} shows how to extend \pfname to model
new serverless programming abstractions.

\begin{figure}
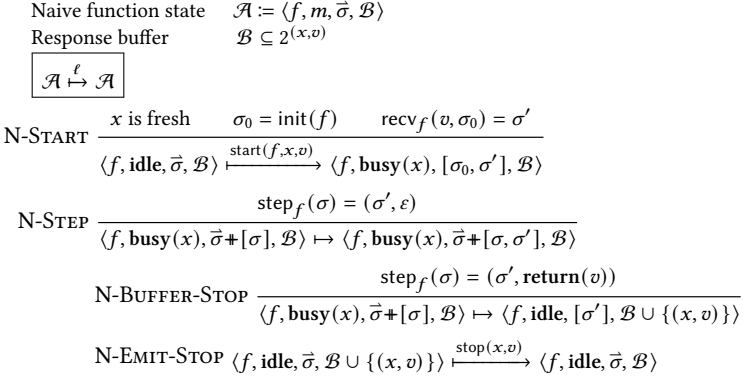

\footnotesize
\(
\begin{array}{@{}lr@{\,}c@{\,}ll}
\textrm{\pnName function state} & \pnstateid & \pdef & \langle f,m,\pnstates,\pnstopbuff\rangle \\
\textrm{Response buffer} & \pnstopbuff & \subseteq & 2^{(x,v)} \\
\multicolumn{5}{@{}l}{\fbox{$\pnstep{\pnstateid}{\ell}{\pnstateid}$}} \\[.5em]
\multicolumn{5}{@{\quad\quad\quad}l}{
    \inferrule*[Left=\pnNameRule-Start]{\textrm{$x$ is fresh} \\ \sigma_0 = \pinitfn (f) \\ \precvfn_f(v, \sigma_0) = \sigma'}
    {\pnstep{\langle f,\pidle,\pnstates,\pnstopbuff\rangle}{\kw{start}(f,x,v)}{\langle f,\pbusy(x),[\sigma_0,\sigma'],\pnstopbuff\rangle}}} \\[.5em]
\multicolumn{5}{@{\quad\quad\quad}l}{
    \inferrule*[Left=\pnNameRule-Step]{\pstepfn_f(\sigma) = (\sigma', \varepsilon)}
    {\pnstep{\langle f,\pbusy(x),\psnoc{\pnstates}{\sigma}, \pnstopbuff \rangle}{}{\langle f,\pbusy(x), \pnstates\pappend[\sigma,\sigma'], \pnstopbuff \rangle}}} \\[.5em]
\multicolumn{5}{@{\quad\quad\quad}l}{
    \inferrule*[left=\pnNameRule-Buffer-Stop]{\pstepfn_f(\sigma) = (\sigma', \ptreturn{v})}
    {\pnstep{\langle f, \pbusy(x), \psnoc{\pnstates}{\sigma}, \pnstopbuff \rangle}{}{\langle f,\pidle, [\sigma'], \pnstopbuff} \cup \{(x,v)\} \rangle}} \\ [.5em]
\multicolumn{5}{@{\quad\quad\quad}l}{
    \inferrule*[left=\pnNameRule-Emit-Stop]{}
    {\pnstep{\langle f, \pidle, \pnstates, \pnstopbuff \cup \{(x,v)\} \rangle }{\kw{stop}(x,v)}{\langle f,\pidle, \pnstates, \pnstopbuff \rangle}}} \\
    [1em]
\end{array}
\)
\caption{A \pnname semantics of serverless functions.}
\label{safe-semantics-figure}
\end{figure}

\section{A Simpler Serverless Semantics}
\label{safe-semantics}

A natural way to make serverless programming easier is to implement a simpler
execution model than \pfname. For example, we could execute functions exactly
once, or eliminate warm starts to avoid reusing state. Unfortunately,
implementing these changes is likely to be expensive (and, in many situations,
beyond our control). Therefore, this section gives programmers a simpler
programming model in the following way. First, we define a simpler \emph{\pnname
semantics} of serverless computing that eliminates most unintuitive behaviors of
\pfname. Second, using a weak bisimulation theorem, we precisely characterize
when the \pnname semantics and \pfname coincide. This theorem addresses the
low-level details of serverless execution once and for all, thus allowing
programmers to reason using the \pnname semantics, even when their code is
running on a full-fledged serverless platform.

In the \pnname serverless semantics (\cref{safe-semantics-figure}), serverless
functions ($f$) are the same as the serverless functions in \pfname. However,
the operational semantics of the \pnname platform is much simpler: the platform
runs a single function $f$ on one request at a time.
At each step of execution, the \pnname semantics either (1)~starts
processing a new event if the platform is idle (\textsc{N-Start}),
(2)~takes an internal step if
the platform is busy (\textsc{N-Step}),
(3)~buffers a completed response (\textsc{N-Buffer-Stop}), or
(4)~responds to a past request (\textsc{N-Emit-Stop}). This buffering behavior is
essential, thus a programmer cannot rely on a platform to
process concurrent messages in-order. However, the naive semantics abstracts away
the details of concurrent execution and warm starts.
The state of a
\pnname platform consists of (1)~the function's name ($f$); (2)~its execution
mode ($m$); (3)~a trace of function states ($\pnstates$), where the last
element of the trace is the current state, and the first element was the
initial state of the function (we write $\pappend$ to append two traces); and (4)~a buffer of responses that
have yet to be returned ($\pnstopbuff$).
The trace is a convenience that helps us relate the \pnname semantics to \pfname,
but has no effect on execution because $\pstepfn_f{}$ only works on the latest state
in the semantics.

\paragraph{\pnName semantics safety.}

Note that the \pnname semantics is an idealized model and is \emph{not correct}
for arbitrary serverless functions. However, we can precisely characterize the
exact conditions when it is safe for a programmer to reason with the \pnname
semantics, even if their code is running on a full-fledged serverless platform
(\ie, using \pfname). We require the programmer to define a \emph{safety
relation} over the state of the serverless function. At a high-level, the
safety relation is an equivalence relation on program states, which ensures
that the (1)~serverless function produces the same observable command (if any)
on equivalent states and that (2)~all final states are equivalent to the initial
state. Intuitively, the latter condition ensures that warm starts and cold
starts are indistinguishable from each other, and the former condition ensures
that interactions between the serverless function and the external world are
identical in equivalent states. The safety relation is formally defined below.

\begin{definition}[Safety Relation]
For a serverless function $\langle f, \Sigma, \precvfn_f, \pstepfn_f, \pinitfn \rangle$,
the relation $\prel \subseteq \Sigma \times \Sigma$ is a \emph{safety relation} if:

\begin{enumerate}

\item $\prel$ is an equivalence relation,

\item for all $(\sigma_1,\sigma_2) \in \prel$ and $v$, $(\precvfn_f(v,\sigma_1), \precvfn_f(v,\sigma_2)) \in \prel$,

\item for all $(\sigma_1,\sigma_2) \in \prel$, if
$(\sigma_1', t_1) = \pstepfn_f(\sigma_1)$ and
$(\sigma_2', t_2) = \pstepfn_f(\sigma_2)$ then
$(\sigma_1', \sigma_2') \in \prel$ and $t_1 = t_2$, and

\item for all $\sigma$, if $\pstepfn_f (\sigma) = (\sigma', \ptreturn{v})$ then
$(\sigma', \pinitfn(f)) \in \prel$.

\end{enumerate}

\label{def:safety-relation}
\end{definition}

\paragraph{Bisimulation relation}

We now define the bisimulation relation, which is a relation between \pnname
states ($\pnstateid$) and \pfname states ($\pfstatetyp$). The bisimulation
relation formally captures several key ideas that are necessary to reason about
serverless execution. (1)~A single \pnname state may be equivalent to multiple
distinct \pfname states. This may occur due to failures and restarts.
(2)~Conversely, a single \pfname state may be equivalent to several \pnname
states. This occurs when a serverless platform is processing several requests.
In fact, we require all \pfname states to be equivalent to all \emph{idle}
\pnname states, which is necessary for \pfname to receive requests at any time.
(3)~The \pfname state may have several function instances evaluating the same
request. (4)~Due to warm starts, the state of a function may not be identical
in the two semantics; however, they will be equivalent (per $\prel$). (5)~Due
to failures, the \pfname semantics can ``fall behind'' the \pnname semantics
during evaluation, but the state of any function instance in \pfname will be
equivalent to some state in the execution history of the \pnname semantics. The
proof of the weak bisimulation theorem accounts for failures, by specifying the
series of \pfname steps needed to then catch up with the \pnname semantics before
an observable event occurs. The bisimulation relation is formally defined below.

\begin{definition}[Bisimulation Relation]
\label{def:bisimulation-relation}
$\pnstate{f}{m}{\pnstates}{\pnstopbuff} \pequiv \pfstatetyp$ is defined as:
\begin{enumerate}

\item
For all $\pinst{f}{\pidle}{\sigma} \in \pfstatetyp$, $(\sigma,\pinitfn(f)) \in \mathcal{R}$;

\item
For all $\pinst{f}{\pbusy(x)}{\sigma_n} \in \pfstatetyp$, there exists
$\preq{f}{x}{v} \in \pfstatetyp$, and a sequence of states $\sigma_0 \cdots \sigma_n$,
such that $(\sigma_0,\pinitfn(f)) \in \mathcal{R}$,
$\precvfn_f(v,\sigma_0) = \sigma_1$, and
$\pstepfn_f(\sigma_i) = (\sigma_{i+1}, \varepsilon)$; and

\item For all $(x,w) \in \pnstopbuff$ there exists $\preq{f}{x}{v}$ $\in \pfstatetyp$,
and  a sequence of states $\sigma_0 \cdots \sigma_n$, such that
$\precvfn_f(v, \pinitfn(f)) = \sigma_0$, 
$\pstepfn_f(\sigma_i) = (\sigma_{i+1}, \epsilon)$, and
$\pstepfn_f(\sigma_{n-1}) = (\sigma_n, \ptreturn{w})$.

\end{enumerate}

\end{definition}

\paragraph{Weak Bisimulation Theorem}

We are now ready to prove a weak bisimulation between the \pnname semantics and
\pfname, conditioned on the serverless functions satisfying the safety relation
defined above. We prove a weak (rather than a strong) bisimulation because
\pfname models serverless execution in more detail. Therefore, a single step in
the \pnname semantics may correspond to several steps in \pfname. The theorem
below states that the \pnname semantics and \pfname are indistinguishable to the
programmer, modulo unobservable steps. The first part of the theorem states
that every step in the \pnname semantics corresponds to some sequence of steps in
\pfname. We can interpret this as the sequence of steps that a serverless
platform needs to execute to faithfully implement the \pnname semantics. On the
other hand, the second part of the theorem states that any arbitrary step in
\pfname---\,including failures, retries, and warm starts\,---\,corresponds to a
(possibly empty) sequence of steps in the \pnname semantics.

An important simplification in the \pnname semantics is that it executes a single
request at a time. Therefore, to relate a \pnname trace to a \pfname
trace, we need to filter out events that are generated by other requests. To do
so, we define $x(\tightoverharp{\ell})$ as the sub-sequence of $\tightoverharp{\ell}$ that only
consists of events labeled $x$. In addition, we write $\pfsteps{}$ and $\pnsteps{}$
for the reflexive-transitive closure of $\Rightarrow$ and $\mapsto$ respectively.
With these definitions, we can state the weak
bisimulation theorem.

\begin{theorem}[Weak Bisimulation]
\label{thm:Weak Bisimulation}
For a serverless function $f$ with a safety relation $\prel$, for all
 $\pnstateid$, $\pfstatetyp$, $\ell$:
 \begin{enumerate}

\item For all $\pnstateid'$, if $\pnstep{\pnstateid}{\ell}{\pnstateid'}$ and
$\pnstateid \pequiv \pfstatetyp$ then there exists $\tightoverharp{\ell_1}$,
$\tightoverharp{\ell_2}$, $\pfstatetyp'$, $\pfstatetyp_i$ and
$\pfstatetyp_{i+1}$ such that
$\pfstatetyp \pfsteps{\tightoverharp{\ell_1}} \pfstep{\pfstatetyp_i}{\ell}{\pfstatetyp_{i+1}} \pfsteps{\tightoverharp{\ell_2}} \pfstatetyp'$,
$x(\tightoverharp{\ell_1}) = \varepsilon$, $x(\tightoverharp{\ell_2}) = \varepsilon$,
and $\pnstateid' \pequiv \pfstatetyp'$

\item For all $\pfstatetyp'$, if $\pfstep{\pfstatetyp}{\ell}{\pfstatetyp'}$ and
$\pnstateid \pequiv \pfstatetyp$ then there exists $\pnstateid'$
such that $\pnstateid' \pequiv \pfstatetyp'$ and
$\pnstateid \pnsteps{\ell} \pnstateid'$.

\end{enumerate}

\end{theorem}

\begin{proof}
See \cref{appendix1}.
\end{proof}

In summary, this theorem allows programmers to justifiably ignore the low-level
details of \pfname, and simply use the \pnname semantics, if their code satisfies
\Cref{def:safety-relation}. There are now several tools that are working toward
verifying these kinds of properties in scripting languages, such as
JavaScript~\cite{javert,park:kjs}, which is the most widely supported
language for writing serverless functions. Our work, which is source
language-neutral, complements this work by establishing the verification
conditions necessary for correct serverless execution. The rest of these
section gives examples that illustrate the kind of reasoning needed to verify
serverless function safety.

\begin{figure}
\begin{subfigure}{0.4\textwidth}
\lstset{language=JavaScript}
\begin{lstlisting}
var cache = new Map();
function auth(req, res) {
  let {user, pass} = req.body;
  if (cache.contains(user, pass)) {
    res.write(true);
  } else if (db.get(user) === pass) {
      cache.insert(user, pass);
      res.write(true);
  } else {
    res.write(false);
  }
}
\end{lstlisting}
\end{subfigure}
\,\vrule\,
\begin{subfigure}{0.56\textwidth}
\footnotesize
\(
\begin{array}{@{}l@{\;}r@{\,}c@{\,}l@{}l}
\textrm{Username} & U & \pdef & \cdots \\
\textrm{Password} & P & \pdef & \cdots \\
\textrm{Cache and Database} &\; C, D & \in & U \rightharpoonup P \\
\textrm{Program state} & \Sigma & \pdef & \textsf{Option}~(U \times P) \times C \times D \\
\end{array}
\)
\(
\begin{array}{@{}l}
\pinitfn(f)  =  (\textsf{None}, \cdot, D) \\
\precvfn_f((u,p), (\textsf{None}, c, D)) =  (\textsf{Some}~(u, p), c, D) \\
\pstepfn (\texttt{Some} (u, p), c, D) = \begin{cases}
  ((\texttt{None}, c[u \mapsto p], D), \ptreturn{\texttt{\lstinline|true|}}) \\
  \hskip 2em \textrm{if }u \notin \mathsf{dom}(C) \wedge D(u) = p \\
  ((\texttt{None}, c, D), \ptreturn{\texttt{\lstinline|false|}}) \\
  \hskip 2em \textrm{if }u \notin \mathsf{dom}(C) \wedge D(u) \ne p \\
  ((\texttt{None}, c, D), \ptreturn{\texttt{\lstinline|true|}}) \\
  \hskip 2em \textrm{if } C(u) = p \\
\end{cases}

\end{array}
\)
\end{subfigure}
\caption{An authentication example that caches the recent authentications to decrease
number of authentication server calls.}
\label{fig-auth-ex}
\end{figure}

\subsection{Examples of Safe and Unsafe Serverless Functions}
\label{naive-examples}

We now give two examples of serverless functions and show that they are safe
and unsafe respectively using only the definition of the safety relation.

\paragraph{In-Memory Cache}
\Cref{fig-auth-ex} is a serverless function that receives a username and
password combination, and returns true if the combination is correct. The
function queries an external database for the password. Since database requests
take time, the function locally caches correct passwords to improve
performance. The cache will be empty on cold starts and may be non-empty on
warm starts. For simplicity, we assume that passwords do not change. (A more
sophisticated example would invalidate the cache after a period of time.)

Ignoring JavaScript-specific details, this program operates in two kinds of
states: (1)~in the initial state, the program is idle and waiting to
receive a request and (2)~while processing a request, the program has a
username ($U$) and password ($P$) in memory. We model the two states with the
type $\textsf{Option} (U \times P)$. In both states, the program has an
in-memory cache ($C$) and access to the database ($D$). Although we assume the
database is read-only, the program may update the cache. Therefore, the complete
type of program state is a product of these three components ($\Sigma$ in
\cref{fig-auth-ex}).

When the program receives a request carrying a username and password, it
records them in program state and leaves the cache unmodified (\precvfn{} in
\cref{fig-auth-ex}). After receiving a request, the JavaScript program performs
a series of internal steps to check for a cached result, query the database (if
needed), and update the cache. For brevity, our model of the program condenses
these operations into a single step (\pstepfn{} in \cref{fig-auth-ex}).

Given this model of the program, we define the safety relation ($\mathcal{R}$) as follows:
\begin{eqnarray*}
((\texttt{Some} (u, p), c, D), (\texttt{Some} (u, p), c', D)) \in \prel &\textrm{ if } c \subseteq D \wedge c' \subseteq D\\
((\texttt{None}, c, D), (\texttt{None}, c', D)) \in \prel &\textrm{ if } c \subseteq D \wedge c' \subseteq D
\end{eqnarray*}
This relation specifies that two program states are equivalent only if they are both
idle states or both processing the same request (same username and password combination).
However, the relation allows the caches ($c$ and $c'$) in either state to be different,
as long as both  are consistent with the database. The latter condition is the key to
ensuring that warm starts are safe.

Finally, we need to prove that the safety relation above satisfies the four criteria of
\cref{def:safety-relation}:

\begin{enumerate}

  \item It is straightforward to show that $\mathcal{R}$ is an equivalence relation.

  \item To show that \precvfn{} maps equivalent states to equivalent states,
  note that \precvfn{} is only defined when the program state does not contain
  a query (\ie, the first component is \textsf{None}).
   Therefore, the two equivalent input states
  may only be of the form $(\mathsf{None}, c, D)$ and $(\mathsf{None}, c', D)$,
  where $c,c' \subseteq D$.   \precvfn{} records the query in
  program state and leaves the cache unmodified, therefore the input states
  are related.

  \item To show that \pstepfn{} maps equivalent states to equivalent states,
  note that \pstepfn{} is only defined when the program state contains
  a query ($\mathsf{Some}~(U, P)$). Moreover, for two states with queries to be
  equivalent, their queries must be identical.
  We have to consider the six combinations of \pstepfn{} and ensure that
  it is never the case that one state produces $\ptreturn{\texttt{true}}$ while
  the other state produces $\ptreturn{\texttt{false}}$. This does not occur because
  the two caches are consistent with the database. We have to also ensure that
  the resulting states are equivalent, which is straightforward because the
  cache updates preserve consistency.

  \item Finally, to show that final states are related to $\pinitfn(f)$, note
  that \pstepfn{} produces a state with \textsf{None} for the query, and all
  these states are related by $\mathcal{R}$, as long as their caches are
  consistent with the database.

\end{enumerate}

Therefore, since $\prel$ is a safety relation, by \cref{thm:Weak Bisimulation},
the function operates the same way using \pfname and the \pnname semantics.

\begin{figure}
\lstset{language=JavaScript}
\begin{lstlisting}
var process = require('process');
exports.main = function (req, res) {
  pid = parseInt(process.pid);
  if (pid > 10000) {
    res.write ({"output": "High process id"});
  } else {
    res.write ({"output": "Low process id"};});}
\end{lstlisting}
\caption{It is not possible to define a safety relation for this function,  because it depends on the process ID.}
\label{fig-process-id}
\end{figure}

\paragraph{Unsafe Serverless Functions}

Not all serverless functions are safe, thus it isn't always possible to define
a safety relation. For example, the function in \Cref{fig-process-id} responds
with the UNIX process ID, which will vary across function instances.
A function can observe other low-level detail of the instance, such as its
IP address or Ethernet address.

\section{Serverless Functions and Cloud Storage}
\label{sec:Serverless Functions and Cloud Storage}

\newcommand{\dbnaivespace}{}

\begin{figure}
\footnotesize
\(
\begin{array}{@{}lr@{\,}c@{\,}ll}
\multicolumn{5}{@{}l}{\textbf{Serverless Functions}} \\
\textrm{Key set}            & k        & \pdef &  & \textrm{strings} \\
\textrm{Key-value map}      & \pkv     & \in   & k \rightharpoonup v \\
\textrm{Commands}	& t		   & \pdef & \cdots \\
&		   & \mid  & \ptbeginTx 	& \textrm{Lock data store} \\
&		   & \mid  & \ptread{k} & \textrm{Read value} \\
&		   & \mid  & \ptwrite{k}{v} & \textrm{Write value} \\
&		   & \mid  & \ptendTx 	& \textrm{Unlock data store} \\
\textrm{Component set} & \pfstatetyp & \pdef & \{ \mathcal{C}_1, \cdots,  \mathcal{C}_n, \pds{\pkv}{L}\} \\
\textrm{Lock state}		& L	& \pdef & \plfree &	\textrm{Data store is free} \\
						&	& \mid	& \plowned{y}{\pkv} & \textrm{Owned by $y$} \\[1.5em]
\multicolumn{5}{@{\quad\quad\quad}l}{\inferrule*[Left=Read]
  	{\pstepfn_f{(\sigma)} = (\sigma', \ptread{k}) \\ \precvfn_f(\pkv'(k), \sigma') = \sigma''}
  	{{\begin{array}{ll}&\mathcal{C} \pinst{f}{\pbusy{(x)}}{\sigma} \pds{\pkv}{\plowned{y}{\pkv'}} \\
         \pfsteparr{}
         &\mathcal{C} \pinst{f}{\pbusy{(x)}}{\sigma''} \pds{\pkv}{\plowned{y}{\pkv'}}
	\end{array}} }} \\[1.5em]
\multicolumn{5}{@{\quad\quad\quad}l}{\inferrule*[Left=Write]
  {\pstepfn_f{(\sigma)} = (\sigma', \ptwrite{k}{v}) \\ \pkv'' = \pkv'[k \mapsto v]}
  {{\begin{array}{ll} &\mathcal{C} \pinst{f}{\pbusy{(x)}}{\sigma} \pds{\pkv}{\plowned{y}{\pkv'}} \\
         \pfsteparr{}
         &\mathcal{C} \pinst{f}{\pbusy{(x)}}{\sigma'} \pds{\pkv}{\plowned{y}{\pkv''}} \end{array}} }} \\[1.5em]
\multicolumn{5}{@{\quad\quad\quad}l}{\inferrule*[Left=BeginTx]
  {\pstepfn_f{(\sigma)} = (\sigma', \ptbeginTx)}
  {
  	{\begin{array}{ll}
  	&\mathcal{C} \pinst{f}{\pbusy{(x)}}{\sigma} \pds{\pkv}{\plfree} \\
         \pfsteparr{}
         &\mathcal{C} \pinst{f}{\pbusy{(x)}}{\sigma'} \pds{\pkv}{\plowned{y}{\pkv}}
    \end{array}}
  }} \\[1.5em]
\multicolumn{5}{@{\quad\quad\quad}l}{\inferrule*[Left=EndTx]
  {\pstepfn_f{(\sigma)} = (\sigma', \ptendTx)}
  {{\begin{array}{ll} &\mathcal{C} \pinst{f}{\pbusy{(x)}}{\sigma} \pds{\pkv}{\plowned{y}{\pkv'}} \\
         \pfsteparr{}
         &\mathcal{C} \pinst{f}{\pbusy{(x)}}{\sigma'} \pds{\pkv'}{\plfree} \end{array}} }} \\[1.5em]
\multicolumn{5}{@{\quad\quad\quad}l}{\inferrule*[Left=DropTx]
  {\forall f \sigma x . \; \pinst{f}{\pbusy{(x)}}{\sigma} \not \in \mathcal{C}}
  {\pfstep{\mathcal{C} \pds{\pkv}{\plowned{y}{\pkv'}}}
         {}
         {\mathcal{C} \pds{\pkv}{\plfree}}}} \\[1.5em]
\end{array}
\)
\caption[Lambda Lambda augmented with a key-value store.]{\pfname augmented with a key-value store.}
\label{shared-state-semantics}
\end{figure}

It is common for serverless functions to use an external database for
persistent storage because their local state is ephemeral. But, serverless
platforms warn programmers that stateful serverless functions must be
\emph{idempotent}~\cite{google-cf-exec, openwhisk-actions, azure-messaging}. In
other words, they should be able to tolerate re-execution. Unfortunately, it is
completely up to programmers to ensure that their code is idempotent, and
platforms do not provide a clear explanation of what idempotence means, given
that serverless functions perform warm-starts, execute concurrently, and may
fail at any time. We now address these problems by adding a key-value store to
both \pfname and the \pnname semantics, and present an extended weak
bisimulation. In particular, the \pnname semantics still processes a single request
at a time, which is a convenient mental model for programmers.

\Cref{shared-state-semantics} augments \pfname with a key-value store that
supports transactions. To the set of components, we add exactly one key-value
store ($\pds{\pkv}{L}$), which has a map from keys to values ($\pkv$) and a
lock ($L$), which is either unlocked ($\plfree$) or contains uncommitted updates
from the function instance that holds the lock ($\plowned{y}{M'}$). An important
detail here is that the lock is held by a function instance and not a request, since there may be
several running instances processing the same request.
We allow serverless functions to produce
four new commands: $\ptbeginTx$ starts a transaction,
$\ptendTx$ commits a transaction, $\ptread{k}$ reads the value associated with
key $k$, and $\ptwrite{k}{v}$ sets the key $k$ to value $v$.
We add four new rules to \pfname that execute these commands in the natural
way: \textsc{BeginTx} blocks until it can acquire a lock,
\textsc{EndTx} commits changes and releases a lock, and for simplicity,
the \textsc{Read} and \textsc{Write} rules require the running instance to
have a lock. Finally, we need a fifth rule (\textsc{DropTx}) that releases a lock
and discards its
uncommitted changes if the function instance that held the lock no longer
exists. This may occur if the function instance dies before committing
its changes.

\begin{figure}
\footnotesize
\(
\begin{array}{@{}lr@{\,}c@{\,}ll}
\multicolumn{5}{@{}l}{\textbf{Serverless Functions}} \\
\textrm{Optional key-value map} & \pnlockid & \pdef & \pnlock{\pkv}{\pnstates} \mid \pncommit{v} \mid \cdot \\
\multicolumn{5}{@{}l}{\fbox{$\pnstep{\pnlockid,\pnstateid}{\ell}{\pnlockid,\pnstateid}$}} \\[.5em]
\multicolumn{5}{@{\dbnaivespace}l}{
    \inferrule*{\pnstep{\pnstateid}{}{\pnstateid'}}
    {\pnstep{\pnlockid,\pnstateid}{}{\pnlockid,\pnstateid'}}

    \inferrule*{\pnstep{\pnstateid}{\kw{start}(f,x,v)}{\pnstateid'}}
    {\pnstep{\pnunlocked,\pnstateid}{\kw{start}(f,x,v)}{\pnunlocked,\pnstateid'}}

    \inferrule*{\pnstep{\pnstateid}{\kw{stop}(x,v)}{\pnstateid'}}
    {\pnstep{\pncommit{v},\pnstateid}{\kw{stop}(x,v)}{\pnunlocked,\pnstateid'}}

} \\[.5em]
\multicolumn{5}{@{\dbnaivespace}l}{
    \inferrule*[left=\pnNameRule-Read]
        {\pstepfn_f{(\sigma)} = (\sigma', \ptread{k}) \\ \precvfn(\pkv(k), \sigma') = \sigma''}
        {{
            \begin{array}{l@{\,}l}\pnstep{&\pnlock{\pkv}{\pnstates'},\langle f,\pbusy(x),\psnoc{\pnstates}{\sigma},\pnstopbuff \rangle \\}
                {}
                {&\pnlock{\pkv}{\pnstates'},\langle f,\pbusy(x),\psnoc{\pnstates}{\sigma,\sigma',\sigma''},\pnstopbuff \rangle}\end{array}
        }}} \\[1.5em]
\multicolumn{5}{@{\dbnaivespace}l}{
    \inferrule*[left=\pnNameRule-Write]
        {\pstepfn_f{(\sigma)} = (\sigma', \ptwrite{k}{v}) \\ \pkv' = \pkv[ k \mapsto v] }
        {{\begin{array}{l@{\,}l}
            \pnstep{&\pnlock{\pkv}{\pnstates'},\langle f,\pbusy(x),\psnoc{\pnstates}{\sigma},\pnstopbuff \rangle \\}
                {}
                {&\pnlock{\pkv'}{\pnstates'},\langle f,\pbusy(x),\psnoc{\pnstates}{\sigma,\sigma'},\pnstopbuff \rangle}
        \end{array}}}} \\[1.5em]
\multicolumn{5}{@{\dbnaivespace}l}{
    \inferrule*[left=\pnNameRule-BeginTx]
        {\pstepfn_f{(\sigma)} = (\sigma', \ptbeginTx) }
        {{\begin{array}{l@{\,}l}\pnstep{&\pnunlocked,\langle f,\pbusy(x),\psnoc{\pnstates}{\sigma},\pnstopbuff \rangle \\}
                {}
                {&\pnlock{\pkv}{\psnoc{\pnstates}{\sigma}},\langle f,\pbusy(x),\psnoc{\pnstates}{\sigma,\sigma'},\pnstopbuff \rangle}
        \end{array}}}} \\[1.5em]
\multicolumn{5}{@{\dbnaivespace}l}{
    \inferrule*[left=\pnNameRule-EndTx]
        {\pstepfn_f{(\sigma)} = (\sigma', \ptendTx) }
        {{\begin{array}{l@{\,}l}
            \pnstep{&\pnlock{\pkv}{\pnstates'},\langle f,\pbusy(x),\psnoc{\pnstates}{\sigma},\pnstopbuff \rangle \\}
                {}
                {&\pncommit{M(x)} ,\langle f,\pbusy(x),\psnoc{\pnstates}{\sigma,\sigma'},\pnstopbuff \rangle}
        \end{array}}}} \\[1.5em]
\multicolumn{5}{@{\dbnaivespace}l}{
    \inferrule*[left=\pnNameRule-Rollback]
        {\phantom{.}}
        {\pnstep{\pnlock{\pkv}{\pnstates'},\langle f,\pbusy(x),\pnstates,\pnstopbuff \rangle}
                {}
                {\cdot,\langle f,\pbusy(x),\pnstates',\pnstopbuff \rangle}}}
\end{array}
\)
\caption{The \pnname semantics with a key-value store.}
\label{safe-state-semantics}
\end{figure}

\paragraph{Idempotence in the \pnname semantics.}

There are several ways to ensure that a serverless
function is idempotent. A common protocol is to save each
output value, keyed by the unique request \textsc{id}, to the key-value store, within a transactional
update. Therefore, if the request is re-tried, the function can lookup and return
the saved output value. We now formally characterize this protocol, and
use it to prove a weak bisimulation theorem
between \pfname and the \pnname semantics, where each is extended with a
key-value store. This will allow programmers to reason about serverless
execution using the \pnname semantics, which processes exactly one request
at a time, without concurrency.

The challenge we face is to extend the bisimulation relation
(\cref{def:bisimulation-relation}) to account for the key-value store. In that
definition, when the \pfname state and the \pnname state are equivalent, it is
possible for all function instances in \pfname to fail. When this occurs,
\pfname ``falls behind'' the \pnname semantics. Nevertheless, we still treat the
states as equivalent, and let the weak bisimulation proof re-invoke function
instances until \pfname catches up with the \pnname semantics. Unfortunately,
this approach does not always work with the key-value
store, since the key-value store may have changed.
 To address this,
we need to ensure that functions that use the key-value store follow an
appropriate protocol to ensure idempotence. A more subtle problem arises when
the \pnname state is within a transaction, and the equivalent
\pfname state takes several steps that
result in a failure, followed by other updates to the key-value store. When this
occurs, the \pnname semantics must rollback to the start of the transaction and
re-execute with the updated key-value store.

\Cref{safe-state-semantics} shows the extended \pnname semantics, which addresses
these
issues. In this semantics, the \pnname key-value store ($\pnlockid$) goes through
three states: (1)~at the start of execution, it is not present; (2)~when a
transaction begins, the semantics selects a new mapping nondeterministically;
and (3)~when the transaction completes, the mapping moves to a committed state,
where it only contains the final result. For simplicity, we assume that reads
and writes only occur within transactions. The semantics also includes an
\textsc{\pnNameRule-Rollback} rule, which allows execution to rollback to the start of
transaction. However, once a transaction is complete
(\textsc{\pnNameRule-EndTx}), a rollback is not possible.

The extended bisimulation relation, shown below, uses the bisimulation
relation from the previous section (\cref{def:bisimulation-relation}).
When the \pnname semantics is within a transaction, the relation requires
some instance in \pfname to be operating in lock-step with the \pnname semantics.
However, other instances in \pfname that are not using the key-value store can
make progress. Therefore, a transaction is not globally atomic in \pfname,
and other requests can be received and processed while some instance is
in a transaction.

\begin{definition}[Extended Bisimulation Relation]
$\pnlockid,\pnstateid \pequiv \pds{\pkv}{L},\pfstatetyp$ is defined as:
\begin{enumerate}

    \item If $\pnstateid \pequiv \pfstatetyp$ then
    $\pnunlocked,\pnstateid \pequiv \pds{\pkv}{\plfree}\pfstatetyp$, or

    \item If $\pnstateid = \pnstate{f}{\pbusy(x)}{\psnoc{\pnstates}{\sigma}}{\pnstopbuff}$,
    $L = \plowned{y}{\pkv'}$, $\pnlockid = \pnlock{\pkv'}{\pnstates}$, and there
    exists $\pinst{f}{\pbusy{(x)}}{\psigmaequiv} \in \pfstatetyp$ such that
    $(\sigma, \psigmaequiv) \in \prel$ then
    $\pnlockid,\pnstateid \pequiv \pds{\pkv}{L},\pfstatetyp$, or

    \item If $\pnstateid = \pnstate{f}{\pbusy(x)}{\pnstates}{\pnstopbuff}$,
    $\pnlockid = \pncommit{v}$, and $\pkv(x)=v$ then
   	$\pnlockid,\pnstateid \pequiv \pds{\pkv}{\plfree}\pfstatetyp$.

\end{enumerate}
\label{def:extended-bisimulation-rel}
\end{definition}

With this definition in place, we can prove an extended weak bisimulation
relation.

\begin{theorem}[Extended Weak Bisimulation]
\label{thm:ext-weak-bisim}
For a serverless function $f$, if $\prel$ is its safety relation and
for all requests $\preq{f}{x}{v}$ the following conditions hold:
\begin{enumerate}
    \item $f$ produces the value stored at key $x$, if it exists,
    \item When $f$ completes a transaction, it stores a value $v'$ at key $x$, and
    \item When $f$ stops, it produces $v'$,
\end{enumerate}
then, for all $\pnlockid$, $\pnstateid$, $\pfstatetyp$, $\ell$ :
\begin{enumerate}

\item For all $\pnlockid',\pnstateid'$, if
$\pnstep{\pnlockid,\pnstateid}{\ell}{\pnlockid',\pnstateid'}$ and
$\pnstateid \pequiv \pfstatetyp$ then there exists
$\tightoverharp{\ell_1}$, $\tightoverharp{\ell_2}$,
$\pfstatetyp'$, $\pfstatetyp_i$ and $\pfstatetyp_{i+1}$ such that
$\pfstatetyp \pfsteps{\tightoverharp{\ell_1}} \pfstep{\pfstatetyp_i}{\ell}{\pfstatetyp_{i+1}} \pfsteps{\tightoverharp{\ell_1}} \pfstatetyp'$,
$x(\tightoverharp{\ell_1}) = \varepsilon$,
$x(\tightoverharp{\ell_2}) = \varepsilon$,
and $\pnstateid' \pequiv \pfstatetyp'$

\item For all $\pfstatetyp'$, if $\pfstep{\pfstatetyp}{\ell}{\pfstatetyp'}$ and
$\pnlockid,\pnstateid \pequiv \pfstatetyp$ then there exists $\pnstateid'$
such that $\pnstateid' \pequiv \pfstatetyp'$ and
$\pnstateid \pnsteps{\ell} \pnstateid'$.

\end{enumerate}

\end{theorem}

\begin{proof}
\ifappendix
By \cref{ext-bisim-fwd,ext-bisim-back} in the appendix.
\else
By Theorems~A.9~and~A.10 in \citet{jangda:lambda-lambda-arxiv}.
\fi
\end{proof}
The theorem statements in
\ifappendix the appendix
\else \citet{jangda:lambda-lambda-arxiv} \fi
formalizes the conditions of the
bisimulation, but the less formal conditions are useful for
programmers, since they are simple requirements that are easy to ensure.
Therefore, this theorem gives the assurance that the reasoning with the
\pnname semantics is adequate, even though the serverless platform operates
using \pfname.

\paragraph{Example}

Consider the banking serverless function (\cref{working-bank-example}), which
uses transactions and an external key-value store. We can show that this
function satisfies the safety relation using the approach presented in
\cref{naive-examples}. We now argue that this function satisfies the three
conditions of \cref{thm:ext-weak-bisim}, which will allow us to reason about
its execution using the \pnname semantics. After the function receives a request,
it extracts the request's ID and checks to see if the database has a value with
that ID (\ref{line:trans-exist-begin} --
\ref{line:trans-exist-end}) and if so immediately returns the saved value. This satisfies first
condition. The function runs a transaction on lines
\ref{line:trans-commit-1-start} -- \ref{line:trans-commit-1} and
\ref{line:trans-commit-2-start} -- \ref{line:trans-commit-2}, and in each case
it stores the value in the database and returns the same value. This satisfies the
second and third conditions.

\section{Serverless Compositions}
\label{spl}

Thus far, we have used \pfname to reason about serverless functions in
isolation, and also extended \pfname to reason about serverless functions that
use a key-value store. In this section, we consider a different kind of extension
to \pfname, which involves a significant change to the serverless computing
platform itself. The change that we make is different from, but inspired by the
work of \citet{baldini:trilemma}. This section first motivates why we would
want to change the serverless platform, and then formalizes the modified
platform. The key takeaway from this section, is that the model of the
modified platform extends \pfname{} as-is, and doesn't involve changing the
definitions and reduction rules in \Cref{function-semantics}. We have also
implemented this modified platform, which we discuss in \cref{sec:eval}.

\begin{figure}
\lstset{language=JavaScript}
\begin{lstlisting}
let request = require('request-promise-native');
exports.postStatus = function(req, res) {
  let {state, sha, url, repo} = req.body;
  request.post({
    url: postStatusToGithub,
    json: { state: state, sha: sha, url: url, repo: repo }})
  .then(function (response, body) {
    if (response.state === "failure") {
      request.post({
        url: postToSlack,
        json: { channel: "<id>", text: "<msg>" }});}});}
\end{lstlisting}
\caption{A serverless function to receive build state from Google Cloud Build,
set status on GitHub using \texttt{postStatusToGitHub}, and report failures to Slack using \texttt{postToSlack}.}
\label{post-build-status}
\end{figure}
\subsection{The Need for Serverless Composition Languages}
\label{sec:motSPL}

\lstset{language=spl_native}

Serverless platforms encourage programmers to decompose large applications into
several little functions (or, ``microservices''). This approach has obvious
benefits: smaller functions are easier to understand, and can be reusable.
Amazon Web Services has a \emph{Serverless Application Repository} that
encourages programmers to reuse and share
microservices.\footnote{https://aws.amazon.com/serverless/serverlessrepo/} For
example, a serverless function to post messages on Slack\,--\,which is
available on the AWS Serverless Application Repository\,--\,could be used to
implement notifications for many different applications.

Many development teams use an application that connects a Slack
channel, a GitHub repository, and a continuous integration (CI) service (e.g.,
TravisCI or Google Cloud Build). The CI service tests every commit to the
GitHub repository. After testing completes, (1)~the CI service invokes the
application with the test results, (2)~the application updates the build status
on GitHub, and (3) the application posts a message to Slack only if testing fails. This application
is a good fit for serverless computing and is easy to write by wiring together
existing serverless functions that post to GitHub and Slack, as sketched in
\cref{post-build-status}. This example is a serverless function that acts as a
coordinator (\lstinline|postStatus|) that invokes two other auxiliary serverless
functions (\lstinline|postStatusToGitHub| and \lstinline|postToSlack|).

However, as \citet{baldini:trilemma} point out, a major problem with this
approach is that the programmer gets ``double-billed'' and has to pay for the time spent
running the coordinator function, which is mostly idle, and for the time spent doing
actual work in the auxiliary functions. An alternative approach is to merge
several functions into a single function. Unfortunately, this approach hinders
code-reuse. In particular, it does not work when source code is unavailable or
when the serverless functions are written in different languages. A third
approach is to write serverless functions that each pass their output as input
to another function, instead of returning to the caller (\ie,
continuation-passing style). However, this approach requires rewriting code.
Moreover, some clients, such as web browsers, cannot produce a continuation
\textsc{url} to receive the final result. The only way to resolve this problem
is to modify the serverless computing platform with function composition
primitives.

\begin{figure}
\newcommand{\splsemspace}{\quad\quad\quad\quad\quad}
\footnotesize
\(
\begin{array}{@{}lr@{\,}c@{\,}ll}
\textrm{Values}          	& v
  & \pdef & \cdots \\
& & \mid  & (v_1, v_2)      & \textrm{Tuples} \\
\textrm{\spl expressions} & \pe
  & \pdef & \pepure{f} & \textrm{Invoke serverless function} \\
& & \mid  & \pefirst{\pe} & \textrm{Run $\pe$ to first part of input} \\
& & \mid  & \peseq{\pe_1}{\pe_2} & \textrm{Sequencing} \\
\textrm{\spl continuations} & \pkontid
  & \pdef & \pkreturn{x} & \textrm{Response to request} \\
& & \mid  & \pkseq{\pe}{\pkontid} & \textrm{In a sequence}  \\
& & \mid  & \pkfirst{v}{\pkontid} & \textrm{In \kw{first}} \\
\textrm{Components}  & C
  & \pdef & \cdots \\
& & \mid  & \prune{\pe}{v}{\pkontid} & \textrm{Running program} \\
& & \mid  & \pwaite{x}{\pkontid}     & \textrm{Waiting program} \\
& & \mid  & \preq{e}{x}{v}           & \textrm{Run program $e$ on $v$} \\
\multicolumn{5}{@{}l}{\fbox{$\pfstep{\pfstatetyp}{\ell}{\pfstatetyp}$}} \\
\multicolumn{5}{@{\splsemspace}l}{
    \inferrule*[Left=P-NewReq]{\textrm{$x$ is fresh}}{\pfstep{\mathcal{C}}{\kw{start}(v)}{\mathcal{C}\preq{e}{x}{v}}}} \\[1em]
\multicolumn{5}{@{\splsemspace}l}{\inferrule*[Left=P-Start]
    {}
    {\pfstep{\mathcal{C} \preq{e}{x}{v}}
            {}
            {\mathcal{C} \preq{e}{x}{v}  \prune{e}{v}{\pkreturn{x}}}}}  \\[1em]
\multicolumn{5}{@{\splsemspace}l}{\inferrule*[Left=P-Respond]
    {}
    {\pfstep{\mathcal{C}\pwaite{v'}{\pkreturn{x}}\preq{e}{x}{v}}
            {\kw{stop}(v')}
            {\mathcal{C} \presp{x}{v'}}}} \\[1em]
\multicolumn{5}{@{\splsemspace}l}{\inferrule*[Left=P-Seq1]
    {}
    {\pfstep{\mathcal{C}\prune{\peseq{\pe_1}{\pe_2}}{v}{\pkontid}}
            {}
            {\mathcal{C}\prune{\pe_1}{v}{\pkseq{\pe_2}{\pkontid}}}}}  \\[1em]
\multicolumn{5}{@{\splsemspace}l}{\inferrule*[Left=P-Seq2]
    {}
    {\pfstep{\mathcal{C}\pwaite{v}{\pkseq{\pe}{\pkontid}}}
            {}
            {\mathcal{C}\prune{\pe}{v}{\pkontid}}}}   \\[1em]
\multicolumn{5}{@{\splsemspace}l}{\inferrule*[Left=P-Invoke1]
    {\textrm{$x'$ is fresh}}
    {\pfstep{\mathcal{C}\prune{\pepure{f}}{v}{\pkontid}}
            {}
            {\mathcal{C}\pwaite{x'}{\pkontid}\preq{f}{x'}{v} }}}   \\[1em]
\multicolumn{5}{@{\splsemspace}l}{\inferrule*[Left=P-Invoke2]
    {}
    {\pfstep{\mathcal{C}\pwaite{x}{\pkontid}\presp{x}{v}}
            {}
            {\mathcal{C}\pwaite{v}{\pkontid}}}}   \\[1em]
\multicolumn{5}{@{\splsemspace}l}{\inferrule*[Left=P-First1]
    {}
    {\pfstep{\mathcal{C}\prune{\pefirst{\pe}}{(v_1, v_2)}{\pkontid}}
            {}
            {\mathcal{C}\prune{\pe}{v_1}{\pkfirst{v_2}{\pkontid}}}}}  \\[1em]
\multicolumn{5}{@{\splsemspace}l}{\inferrule*[Left=P-First2]
    {}
    {\pfstep{\mathcal{C}\pwaite{v_1}{\pkfirst{v_2}{\pkontid}}}
            {}
            {\mathcal{C}\pwaite{(v_1,v_2)}{\pkontid}}}} \\[1em]
\multicolumn{5}{@{\splsemspace}l}{\inferrule*[Left=P-Die]
    {}
    {\pfstep{\mathcal{C} \pwaite{v}{\kappa}}{}{\mathcal{C}}}}
\end{array}
\)
\caption[Extended semantics with \spl.]{Extending \pfname with \spl.}
\label{spl-semantics}
\end{figure}

\subsection{Composing Serverless Functions with Arrows}
\label{sec:designSPL}

We now extend \pfname with a domain specific language for composing
serverless functions, which we call \spl (\emph{serverless programming
language}). Since the
serverless platform is a shared resource and programs are untrusted, \spl
cannot run arbitrary code. However, \spl programs can invoke serverless
functions to perform arbitrary computation when needed. Therefore, invoking a
serverless function is a primitive operation in \spl, which serves as
the wiring between several serverless functions.

\Cref{spl-semantics} extends the \pfname with \spl.
This extension allows requests to run \spl programs ($\preq{\pe}{x}{v}$),
in addition to ordinary requests that name serverless
functions.\footnote{In practice, a
request would name an \spl program instead of carrying the program
itself.}
\spl is based on Hughes' arrows~\cite{hughes:arrows}, thus it supports the
three basic arrow combinators. An \spl
program can (1)~invoke a serverless function ($\kw{invoke}~f$); (2)~run
two subprograms in sequence ($\peseq{\pe_1}{\pe_2}$); or (3)~run a subprogram
on the first component of a tuple, and return the second component
unchanged ($\pefirst{\pe}$). These three operations
are sufficient to describe loop- and branch-free
compositions of serverless functions. It is straightforward to add support
for bounded loops and branches, which we do in our implementation.

To run \spl programs, we introduce a new kind of component
($\mathbb{E}$) that executes programs using an abstract machine that is similar
to a \textsc{ck} machine~\cite{felleisen:cek}. In other words, the evaluation rules
define a small-step semantics with an explicit representation of the
continuation (\pkontid). This design is necessary because programs need to
suspend execution to invoke serverless functions (\textsc{P-Invoke1}) and then
later resume execution (\textsc{P-Invoke2}). Similar to serverless functions,
\spl programs also execute at-least-once. Therefore, a single request may spawn
several programs (\textsc{P-Start}) and a program may die while waiting
for a serverless function to response (\textsc{P-Die}).

\begin{figure}
\footnotesize
\begin{subfigure}[t]{0.45\columnwidth}
  \vspace{0pt}
  \(
\begin{array}{@{}lr@{\,}l@{\,}l}
  \multicolumn{3}{@{}l} {\textrm{\spl expressions}} \\
  \pe & \pdef & \cdots \mid \pjid & \textrm{Run transformation} \\
  \multicolumn{3}{@{}l} {\textrm{\textsc{json} values}} \\
v &::=& n \mid b \mid \mathit{str} \mid \kw{null} \\
  \multicolumn{3}{@{}l} {\textrm{\textsc{json} pattern}} \\
    \pjid & \pdef & v &\textrm{\textsc{json} literal} \\
    \multicolumn{3}{@{}l} {\textrm{\textsc{json} query}} \\
    \pjq & \pdef &                   & \textrm{Empty query} \\
  & \mid  & .[n] \pjq         & \textrm{Array index} \\
  & \mid  & .\mathit{id} \pjq & \textrm{Field lookup} \\
\end{array}
\)
\end{subfigure}
\begin{subfigure}[t]{0.45\columnwidth}
  \vspace{0pt}
\(
\begin{array}{@{}lr@{\,}l@{\,}ll}
  \multicolumn{3}{@{}l}{\textrm{\textsc{json} pattern}} \\
  \pjid
& \pdef & v &\textrm{\textsc{json} literal} \\
 & \mid  & [ \pjid_1, \cdots, \pjid_n ] & \textrm{Array} \\
 & \mid  & \{ \mathit{str}_1 : \pjid_1, \cdots, \mathit{str}_n : \pjid_n \} & \textrm{Object} \\
         &|&\pjid_1~op~\pjid_2 &\textrm{Operators}\\
         &|&\mathtt{if}~(\pjid_1)~\mathtt{then}~\pjid_2~\mathtt{else}~\pjid_3&\textrm{Conditional}\\
         &|&[\mathit{str}_1 \rightarrow \mathit{p}_1]&\textrm{Update field}\\
 & \mid  & \kw{in}~\pjq & \textrm{Input reference} \\
\end{array}
\)
\end{subfigure}
\caption{\textsc{json} transformation language.}
\label{json-dsl}
\end{figure}

\paragraph{A sub-language for \textsc{json} transformations.}

A problem that arises in practice is that input and output values
to serverless functions ($v$) are frequently formatted as \textsc{json}
values, which makes it hard to define the \kw{first} operator in a satisfactory
way. For example, we could define
\kw{first} to operate over two-element \textsc{json} arrays, and then require
programmers to write serverless functions to transform arbitrary \textsc{json}
into this format. However, this approach is cumbersome and resource-intensive.
For even simple transformations, the programmer would have to write and deploy
serverless functions; the serverless platform would need to sandbox the process
using heavyweight \textsc{os} mechanisms; and the platform would have to copy
values to and from the process.

Instead, we augment \spl with a sub-language of \textsc{json} transformations
(\cref{json-dsl}). This language is a superset of \textsc{json}. It has
a distinguished variable (\pjin) that refers to the input \textsc{json} value,
which may be followed by a query to select fragments of the input.
For example, we can use this transformation language to write an \spl program
that receives a two-element array as input and then runs two different
serverless functions on each element:
\lstset{language=spl_native}
\begin{lstlisting}[numbers=none]
first (invoke f) >>> [in[1], in[0]] >>> first (invoke g)
\end{lstlisting}
Without the \textsc{json} transformation language, we would need an
auxiliary serverless function to swap the elements.

\paragraph{A simpler notation for \spl programs.}

\begin{figure}
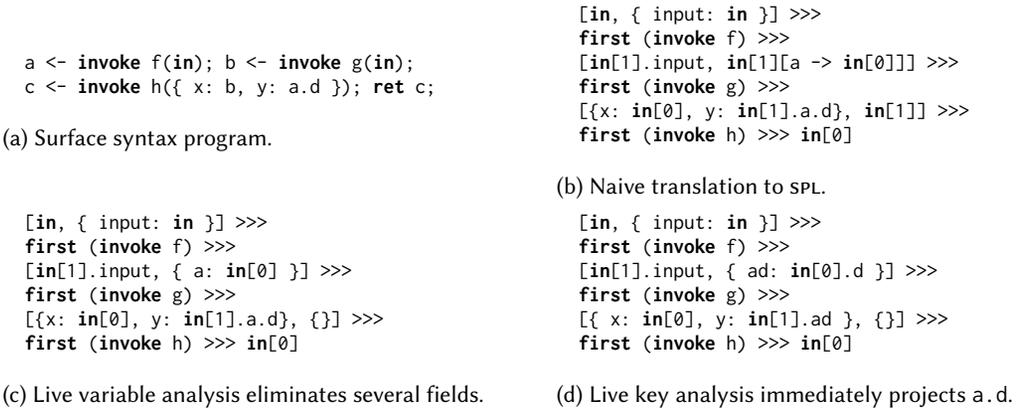

\footnotesize
\captionsetup[subfigure]{justification=justified,singlelinecheck=false}
\begin{subfigure}{0.47\columnwidth}
\lstset{language=spl_surface}
\begin{lstlisting}[numbers=none]
a <- invoke f(in); b <- invoke g(in);
c <- invoke h({ x: b, y: a.d }); ret c;
\end{lstlisting}
\caption{Surface syntax program.}
\label{spl-opt-example}
\end{subfigure}
\hfill
\begin{subfigure}{0.47\columnwidth}
\lstset{language=spl_native}
\begin{lstlisting}[numbers=none]
[in, { input: in }] >>>
first (invoke f) >>>
[in[1].input, in[1][a -> in[0]]] >>>
first (invoke g) >>>
[{x: in[0], y: in[1].a.d}, in[1]] >>>
first (invoke h) >>> in[0]
\end{lstlisting}
\caption{Naive translation to \spl.}
\label{spl-unopt}
\end{subfigure}

\begin{subfigure}{0.47\columnwidth}
\lstset{language=spl_native}
\begin{lstlisting}[numbers=none]
[in, { input: in }] >>>
first (invoke f) >>>
[in[1].input, { a: in[0] }] >>>
first (invoke g) >>>
[{x: in[0], y: in[1].a.d}, {}] >>>
first (invoke h) >>> in[0]
\end{lstlisting}
\caption{Live variable analysis eliminates several fields.}
\label{spl-liveness-opt}
\end{subfigure}
\hfill
\begin{subfigure}{0.47\columnwidth}
\lstset{language=spl_native}
\begin{lstlisting}[numbers=none]
[in, { input: in }] >>>
first (invoke f) >>>
[in[1].input, { ad: in[0].d }] >>>
first (invoke g) >>>
[{ x: in[0], y: in[1].ad }, {}] >>>
first (invoke h) >>> in[0]
\end{lstlisting}
\caption{Live key analysis immediately projects \lstinline{a.d}.}
\label{spl-keys-opt}
\end{subfigure}

\caption{Compiling the surface syntax of \spl.}
\label{fig:compiling spl}
\end{figure}

\spl is designed to be a minimal set of primitives that are
straightforward to implement in a serverless platform. However, \spl programs
are difficult for programmers to comprehend. To address this problem,
we have also developed a surface syntax for \spl that is
based on Paterson's notation for arrows~\cite{paterson:arrow-notation}.
Using the surface syntax, we can rewrite
the previous example as follows:
\lstset{language=spl_surface}
\begin{lstlisting}[numbers=none]
x <- invoke f(in[0]); y <- invoke g(in[1]); ret [y, x];
\end{lstlisting}
This version is far less cryptic than the original.

We describe the surface syntax compiler by example. At a high level, the
compiler produces \spl programs in store-passing style.
For example, \cref{spl-opt-example} shows a surface syntax
program that invokes three serverless functions ($f$, $g$, and $h$). However, the
composition is non-trivial because the input of each function is not simply
the output of the previous function. We compile this program to an
equivalent \spl program that
 uses the \textsc{json} transformation language to save
intermediate values in a dictionary (\Cref{spl-unopt}).
 However, this naive translation carries
unnecessary intermediate state. We address this problem with two optimizations.
First, the compiler performs a live variable analysis, which produces the more
compact program shown in \cref{spl-liveness-opt}. In the original program, the input
reference (\pjin) is not live after $g$, and $c$ is the only live variable after
$h$, thus these are eliminated from the state.
Second, the compiler performs a liveness analysis of the \textsc{json} keys
returned by serverless functions, which produces an even smaller
program (\Cref{spl-keys-opt}). In our example, $f$
returns an object $a$, but the program only uses $a.d$ and discards any other
fields that $a$ may have. There are many situations where the entire object $a$
may be significantly larger than $a.d$, thus extracting it early can shrink
the amount of state a program carries.

\subsection{Implementation}
\label{sec:implSPL}

\paragraph{OpenWhisk implementation.}

Apache OpenWhisk is a mature and widely-deployed serverless platform that is
written in Scala and is the foundation of \textsc{ibm} Cloud Functions. We have
implemented \spl as a 1200 \textsc{loc} patch to OpenWhisk, which includes
the surface syntax compiler and several changes to the OpenWhisk runtime
system. We inherit OpenWhisk's fault tolerance mechanisms (\eg, at-least-once
execution) and reuse OpenWhisk's support for serverless function
sequences~\cite{baldini:trilemma} to implement the \pseqtok{} operator of
\spl.

Our OpenWhisk implementation of \spl has three differences from the language
presented so far. First, it supports bounded loops, which are a programming
convenience. Second, instead of implementing the \kw{first} operator and the
\textsc{json} transformation language as independent expressions, we have a
single operator that performs the same action as \kw{first}, but applies a
\textsc{json} transformation to the input and output, which is how
transformations are most commonly used. Finally, we implement a multi-armed
conditional, which is a straightforward extension to \spl. These operators
allow us to compile the surface syntax to smaller \spl programs, which
moderately improves performance.

\paragraph{Portable implementation.}

We have also built a portable implementation of \spl (1156
\textsc{loc} of Rust) that can invoke serverless functions in public clouds.
(We have tested with Google Cloud Functions.) Whereas the OpenWhisk
implementation allows us to carefully measure load and utilization on our own
hardware test-bed, we cannot perform the same experiments with our standalone
implementation, since public clouds abstract away the hardware used to run
serverless functions. The portable implementation has helped us ensure that the
design of \spl is independent of the design and implementation of OpenWhisk,
and we have used it to explore other kinds of features that a serverless
platform may wish to provide. For example, we have added a \kw{fetch} operator
to \spl that receives the name of a file in cloud storage as input and produces
the file's contents as output. It is common to have serverless functions fetch
private files from cloud storage (after an access control check). The \kw{fetch}
operator can make these kinds of functions faster and consume fewer resources.

\section{Evaluation}
\label{sec:eval}

This section first evaluates the performance of \spl{} using microbenchmarks,
and then highlights the expressivity of \spl{} with three case studies.

\subsection{Comparison to OpenWhisk Conductor}

\begin{figure}

\begin{subfigure}{0.3\textwidth}
\begin{tikzpicture}
\node{\pgfimage[width=.95\columnwidth]{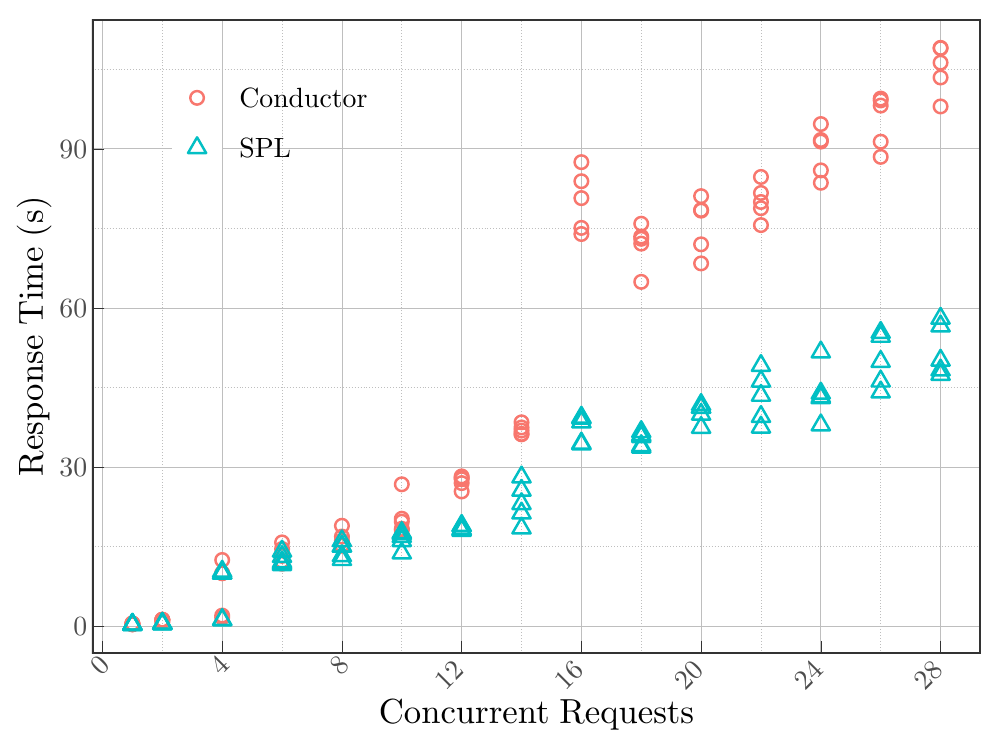}};
\end{tikzpicture}
\caption{Varying concurrent requests.}
\label{fig:invocations-results}
\end{subfigure}
\begin{subfigure}{0.3\textwidth}
\begin{tikzpicture}
\node{\pgfimage[width=.95\columnwidth]{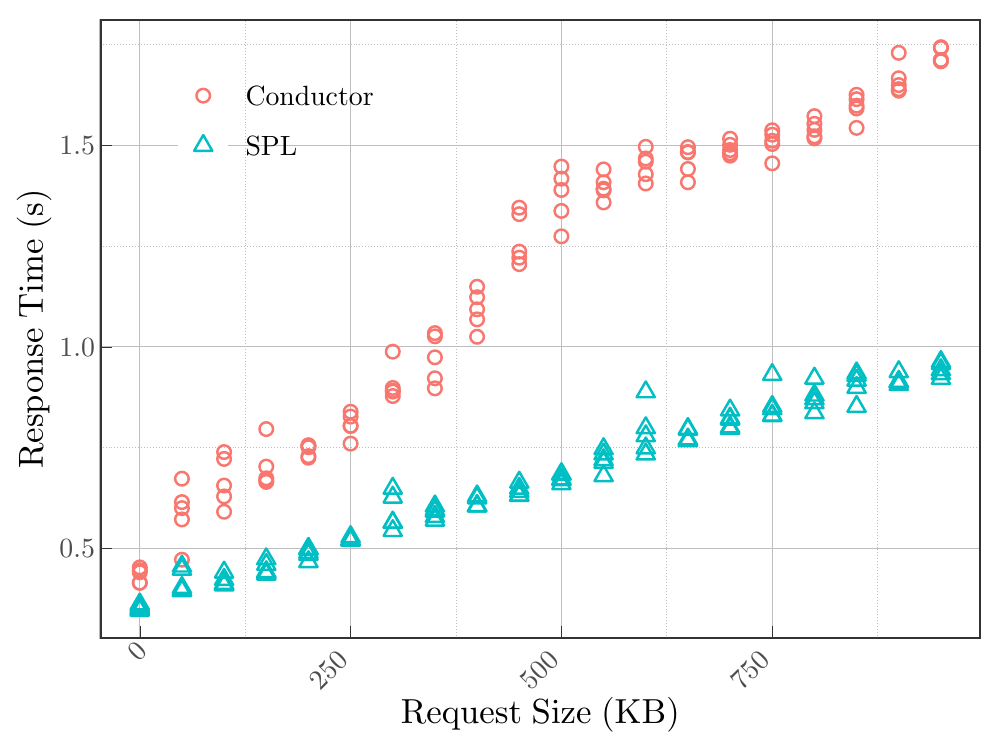}};
\end{tikzpicture}
\caption{Varying request size.}
\label{fig:params-results}
\end{subfigure}
\begin{subfigure}{0.3\textwidth}
\begin{tikzpicture}
\node{\pgfimage[width=.95\columnwidth]{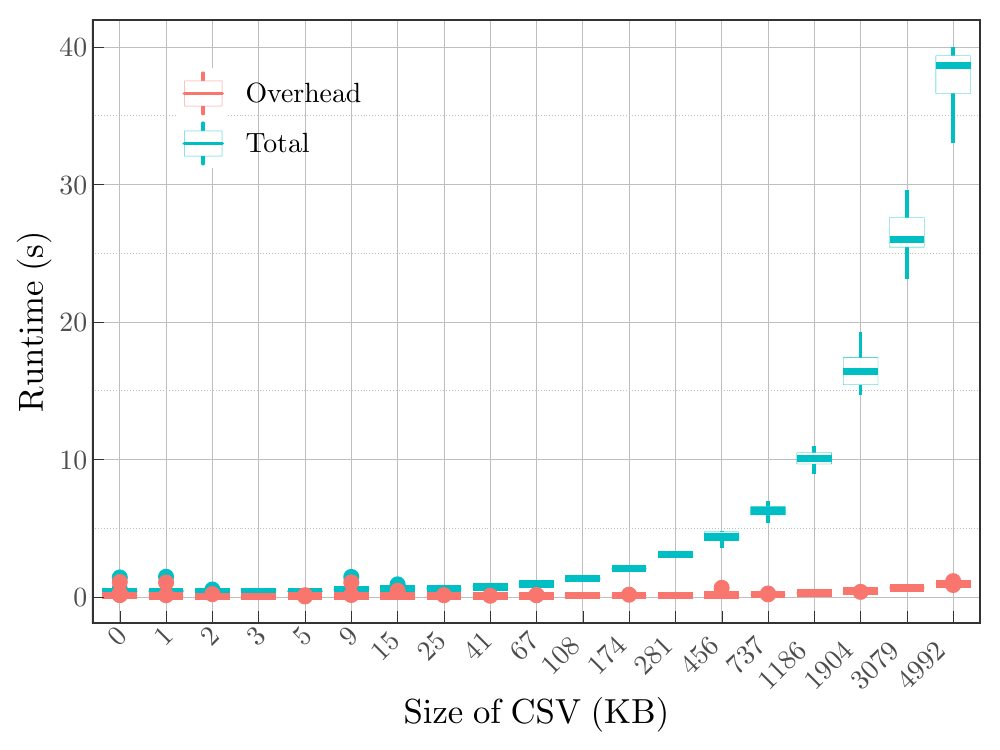}};
\end{tikzpicture}
\caption{Varying CSV size.}
\label{fig:case-study-results}
\end{subfigure}

\caption{\spl benchmarks. \Cref{fig:invocations-results,fig:params-results}
show that \spl is faster than OpenWhisk Conductor when the serverless platform is
processing several concurrent requests, or when the request size increases.
\Cref{fig:case-study-results} shows that most of the execution time is spent
in serverless functions, and not running \spl.}

\end{figure}

The Apache OpenWhisk serverless platform has built-in support for composing
serverless functions using a \emph{Conductor}~\cite{rabbah:composer-better}. A
Conductor is special type of serverless function that, when invoked, may
respond with the name and arguments of an auxiliary serverless function for the
platform to invoke, instead of returning immediately. When the auxiliary
function returns a response, the platform re-invokes the Conductor with the
response value. Therefore, the platform interleaves the execution of the
Conductor function and its auxiliary functions, which allows a Conductor
to implement sequential control flow, similar to \spl. The code for a Conductor
can be written in an arbitrary language (e.g., JavaScript).
The key difference between \spl and Conductor,
is that \spl is designed to run directly on the platform, the
Conductor has to be executed in a container, which consumes additional
resources.

This section compares the performance of \spl to Conductor with a
microbenchmark that stresses the performance of the serverless platform. The
benchmark \spl program runs a sequence of ten serverless functions, and we
translate it to an equivalent program for Conductor. We run two experiments
that (1)~vary the number of concurrent requests and (2)~vary the size of the
requests. In both experiments, we measure end-to-end response time, which is the
metric that is most relevant to programmers. We find that \spl outperforms
Conductor in both experiments, which is expected because its design requires
fewer resources, as explained below.

We conduct our experiments on a six-core Intel Xeon E5-1650 with 64 GB RAM
with Hyper-Threading enabled.

\paragraph{Concurrent invocations.} Our first experiment shows how response
times change when the system is processing several concurrent requests. We run
$N$ concurrent requests of the same program, and then measure five response
times. \Cref{fig:invocations-results} shows that \spl is slightly faster than
Conductor when $N \le 12$, but approximately twice as fast when $N > 12$. We
attribute this to the fact that Conductor interleaves the conductor function
with the ten serverless functions, thus it requires twice as many containers to
run. Moreover, since our \textsc{cpu} has six hyper-threaded cores, Conductor
is overloaded with 12 concurrent requests.

\paragraph{Request size.} Our second experiment shows how response
times depend on the size of the input request. We use the same microbenchmark
as before, and ensure that the platform only processes one request at a time.
We vary the request body size from 0KB to 1MB and measure five response times at each request size.
\Cref{fig:params-results} shows that \spl is almost twice as fast as
Conductor. We again attribute this to the fact that Conductor needs to copy
the request across a sequence of functions that is twice as long as \spl.

\paragraph{Summary.} The \textsc{os}-based isolation techniques that
Conductor uses have a nontrivial cost. \spl, since it uses language-based
isolation, is able to lower the resource utilization of serverless compositions
by up to a factor of two.

\subsection{Case Studies}

The core \spl language, presented in \cref{spl}, is a minimal fragment that
lacks convenient features that are needed for real-world programming. To
identity the additional features that are necessary, we have written several
different kinds of \spl programs, and added new features to \spl when
necessary. Fortunately, these new features easily fit the structure established
by the core language. This section presents some of the programs that we've
built and discusses the new features that they employ. These examples illustrate
that it is easy to grow core \spl into a convenient and expressive programming
language for serverless function composition.

\begin{figure}
\captionsetup[subfigure]{justification=justified,singlelinecheck=false}
\begin{subfigure}{\columnwidth}
\lstset{language=spl_surface}
\begin{lstlisting}
if (in.amount > 100) {
  invoke bank({ type: "deposit", to: "checking",
          amount: in.amount - 100, transId: in.tId1 });
  invoke bank({ type: "deposit", to: "savings",
          amount: 100, transId: in.tId2 });
} else {
  invoke bank({ type: "deposit", to: "checking",
          amount: in.amount, transId: in.tId1 });
} ret;\end{lstlisting}
\caption{Receive funds, deposit \$100 in savings (if feasible),
and deposit the rest in checking.}
\label{spl-bank-example}
\end{subfigure}

\begin{subfigure}{\columnwidth}
\lstset{language=spl_surface}
\begin{lstlisting}
invoke postStatusToGitHub({ state: in.state,
  sha: in.sha, url: in.url, repo: "<repo owner/name>" });
if (in.state == "failure") {
  invoke postToSlack({ channel: "<id>", text: "<msg>" });
} ret;\end{lstlisting}
\caption{Receive build state from Google Cloud Build, set status on GitHub, and
report failures on Slack.}
\label{ci-example}
\end{subfigure}
\begin{subfigure}{\columnwidth}
\lstset{language=spl_surface}
\begin{lstlisting}
data <- get in.url;
json <- invoke csvToJson(data);
out <- invoke plotJson({data:json,x:in.xAxis,y:in.yAxis});
ret out;
\end{lstlisting}
\caption{Receive a \textsc{url} of a \textsc{csv} file and two column names, download the file,
convert it to \textsc{json}, then plot the \textsc{json}.}
\label{spl-csv-ebuildsxample}
\end{subfigure}
\caption{Example \spl{} programs.}
\end{figure}

\paragraph{Conditional bank deposits.}

\Cref{spl-bank-example} uses \spl to write a bank deposit function using the
\lstinline|deposit| function from \Cref{working-bank-example}.
If the received amount is greater than \$100, it is split in two parts
and deposited into the checking and savings accounts, by calling \lstinline|deposit|.
This \spl program does not suffer from ``double billing'' because the serverless
platform suspends the \spl program when it invokes
a serverless function and resumes it when the response is ready. This example
also shows that our implementation supports basic arithmetic, which we add to
the \textsc{json} transformation sub-language in a straightforward way.

\paragraph{Continuous integration.}

\Cref{ci-example} uses \spl to rewrite the Continuous Integration example
(~\Cref{post-build-status}), and is based on an example from
\citet{baldini:trilemma}. The program in \Cref{post-build-status} suffered from
the ``double billing''
problem, since the composite serverless function needs to be active while
waiting for the \lstinline|postStatusToGitHub| and \lstinline|postToSlack| do
the actual work. In contrast, the serverless platform suspends the \spl
program when it invokes a serverless function and resumes it when the response is
ready.  Our \spl program
(\cref{ci-example}) connects GitHub, Slack, and Google Cloud Build (which is a
continuous integration tool, similar to TravisCI).
Cloud Build makes it easy to run tests when a new commit is pushed to GitHub. But, it
is much harder to see the test results in a convenient way. However, Cloud
Build can invoke a serverless function when tests complete, and we use it to
run an \spl program that (1)~uses the GitHub
\textsc{api} to add a test-status icon next to each commit message, and
(2)~uses the Slack \textsc{api} to post a message whenever a test fails.
Instead of writing a monolithic serverless function, we first write two
serverless functions that post to Slack and set GitHub status icons
respectively, and let the \spl{} program invoke them. It is
easy to reuse the GitHub and Slack functions in other applications.

\paragraph{Data visualization.}

Our last example (\cref{spl-csv-ebuildsxample}) receives the \textsc{url} of a
\textsc{csv}-formatted data file with two columns, plots the data, and responds
with the plotted image. This is the kind of task that a power-constrained
mobile application may wish to offload to a serverless platform, especially
when the size of the data is significantly larger than the plot. Our program
invokes two independent serverless functions that are trivial wrappers
around popular JavaScript libraries: \emph{csvjson}, which converts \textsc{csv}
to \textsc{json} and \emph{vega}, which creates plots from \textsc{json}
data. This example uses a new primitive (\kw{get}) that our
implementation supports to download the data file. Downloading is a
common operation that is natural for the platform to provide. Our implementation
simply issues an \textsc{http} \textsc{get} request and suspends the \spl program
until the response is available. However, it is easy to imagine more sophisticated
implementations that support caching, authorization, and other features.
Finally, this example show that our \spl implementation is not limited
to processing \textsc{json}. The \kw{get}
command produces a plain-text \textsc{csv} file, and the \lstinline|plotJson|
invocation produces a \textsc{jpeg} image.

A natural question to ask about this example is whether the decomposition
into three parts introduces excessive communication overhead. We investigate
this by varying the size of the input \textsc{csv}s (ten trials per size),
and measuring the total running time and the overhead, which we define
as the time spent outside serverless functions (\ie, transferring data,
running \kw{get}, and applying \textsc{json} transformations).
\Cref{fig:case-study-results} shows that even as the file size approaches
5 MB, the overhead remains low (less than 3\% for a 5 MB file, and up to 25\% for 1 KB file).

\section{Related Work}
\label{sec:Related Work}

\paragraph{Serverless computing.}

\citet{baldini:trilemma}
introduce the problem of serverless function composition and present a new
primitive for serverless function sequencing. Subsequent work develops
serverless state machines (Conductor) and a \textsc{dsl} (Composer) that makes state
machines easier to write~\cite{rabbah:composer-better}. In \Cref{spl}, we
present an alternative set of composition operators that we formalize as an
extension to \pfname, implement in OpenWhisk, and evaluate their performance.

Trapeze~\cite{alpernas:trapeze} presents dynamic \textsc{ifc} for serverless
computing, and further sandboxes serverless functions to mediate their
interactions with shared storage. Their Coq formalization of
termination-sensitive noninterference does not model some features of
serverless platforms, such as warm starts and failures, that our semantics does
model.

Several projects exploit serverless computing for elastic
parallelization~\cite{jonas:pywren, fouladi:excamera, ao:sprocket, fouladi:gg}.
\Cref{spl} addresses modularity and does not support parallel execution.
However, it would be an interesting challenge to grow the \textsc{dsl} in
\cref{spl} to the point where it can support the aforementioned applications without any
non-serverless computing components. It is worth noting that today's serverless
execution model is not a good fit for all applications. For example,
\citet{singhvi:nf} list several barriers to running network functions on
serverless platforms.

Serverless computing and other container-based platforms suffer several
performance and utilization barriers. There are several ways to address these
problems, including datacenter design~\cite{gan:microservices}, resource
allocation~\cite{bjorkqvist:noms16}, programming
abstractions~\cite{rabbah:composer-better, baldini:trilemma}, edge
computing~\cite{aske:mspc}, and cloud container design~\cite{x-containers:shen}.
\pfname is designed to elucidate subtle semantic
issues (not performance problems) that affect programmers building serverless
applications.

\paragraph{Language-based approaches to microservices.}

SKC~\cite{gabbrielli:SKC} is another formal semantics of serverless computing
that models how serverless functions can interact with each other. Unlike
\pfname, it does not model certain low-level, observable details of serverless
platforms, such as warm starts.

$\lambda_{\textrm{FAIL}}$~\cite{lambda-fail} is a semantics for
horizontally-scaled services with durable storage, which are related to
serverless computing. A key difference
between \pfname and $\lambda_{\textrm{FAIL}}$ is that \pfname models
\emph{warm-starts}, which occur when a serverless platform runs a new request on
an old function instance, without resetting its state. Warm-starts make it hard to
reason about correctness, but this paper presents an approach to do so. Both
\pfname and $\lambda_{\textrm{FAIL}}$ present weak bisimulations between
detailed and \pnname semantics. However, \pfname's \pnname semantics
processes a single request at a time, whereas $\lambda_{\textrm{FAIL}}$'s
idealized semantics has concurrency.
We use \pfname to specify a protocol to ensure that serverless
functions are idempotent and fault tolerant. However, $\lambda_{\textrm{FAIL}}$
also presents a compiler that automatically ensures that these properties
hold for C\# and F\# code. We believe the approach would work for \pfname.
\cref{spl} extends \pfname with new primitives, which we then implement
and evaluate.

Whip~\cite{waye:whip} and \texttt{ucheck}~\cite{panda:hotos2017} are tools that
check properties of microservice-based applications at run-time. These works are
complementary to ours. For example, our paper identifies several important
properties of serverless functions, which could then be checked using Whip or
\texttt{ucheck}.

Orleans~\cite{orleans-msr} is a programming language for developing distributed
systems using virtual actors. Orleans raises the level of abstraction
and provides better locality by automatically
placing hot actors nearby. Orleans is complementary to \pfname, for example, Orleans
can be used to implement \pfname semantics to develop a serverless system.

\paragraph{Cloud orchestration frameworks.}

Ballerina~\cite{weerawarana:ballerina} is a language for managing cloud
environments; Engage~\cite{fischer:engage} is a deployment manager that
supports inter-machine dependencies; Pulumi~\cite{pulumi} is an embedded
\textsc{dsl} for writing programs that configure and run in the cloud;
and CPL~\cite{bracevac:cpl} is a unified language for writing distributed cloud
programs together with their distribution routines. In
contrast, \pfname is a semantics of serverless computing. \cref{spl} uses
\pfname to design and implement a language for composing serverless functions
that runs within a serverless platform.

\paragraph{Verification.}

There is a large body of work on verification, testing, and modular programming
for distributed systems and algorithms (\eg,
\cite{desai:oopsla2018,dragoi:popl2016,guha:coqnet,sergey:popl2018,wilcox:verdi,chajed:sosp2018,bakst:oopsla2017,gomes:oopsla2017,hawblitzel:sosp2015}).
The serverless computation model is more constrained than arbitrary distributed
systems and algorithms. This paper presents a formal semantics of serverless
computing, \pfname, with an emphasis on low-level details that are
observable by programs, and thus hard for programmers to get right.
To demonstrate that \pfname is useful, we present three applications that
employ it and extend it in several ways. This paper does not address verification
for serverless computing, but \pfname could be used as a foundation
for future verification work.

$\lambda_{\textrm{zap}}$~\cite{lambdazap} is a model of computation in the
presence of transient hardware faults (e.g., bit flips). To detect and recover
from such faults, $\lambda_{\textrm{zap}}$ programs replicate computations and
use majority voting. Serverless computing does not address these kinds of data
errors, but does address communication failures by reinvoking functions on new
machines, which we model in \pfname. Unfortunately, function reinvocation
and other low-level behaviors are observable by serverless functions, and its
easy to write a function that goes wrong when reinvoked. Using \pfname, this
paper lays out a methodology to reason about serverless functions while
abstracting away low-level platform behaviors.

\section{Conclusion}
\label{sec:Contributions}

We have presented \pfname, an operational semantics that models the low-level
of serverless platforms that are observable by programmers. We have also
presented three applications of \pfname. (1)~We prove a weak bisimulation to
characterize when programmers can ignore the low-level details of \pfname.
(2)~We extend \pfname with a key-value store to reason about stateful
functions. (3)~We extend \pfname with a language for serverless function
composition, implement it, and evaluate its performance. We hope that these
applications show that \pfname can be a foundation for further research on
language-based approaches to serverless computing.

\section*{Acknowledgements}

This work was partially supported by the National Science Foundation
under grants CNS-1413985, CCF-1453474, and CNS-1513055. We thank
Samuel Baxter, Breanna Devore-McDonald, and Joseph Spitzer for their
work on the \spl implementation.

\bibliography{main}

\ifappendix

\newpage

\appendix
\onecolumn

\usetikzlibrary{cd}
\mathchardef\mhyphen="2D
\newcommand{\pinstprime}[3]{\ensuremath{\mathbb{F}'(#1,#2,#3, \pinstid)}}
\newcommand{\plambdahist}{$\lambda$\awslambda$ +\mathit{hist}$}
\newcommand{\plambdadie}{$\lambda$\awslambda$ +\neg\mathit{die}$}
\newcommand{\plambdabuf}{$\lambda$\awslambda$ +\mathit{buf}$}
\newcommand{\plambdahiststate}{\ensuremath{\mathcal{C}_{\mathit{hist}}}}
\newcommand{\plambdadiestate}{\ensuremath{\mathcal{C}_{\neg\mathit{die}}}}
\newcommand{\plambdabufstate}{\ensuremath{\mathcal{C}_{\mathit{buf}}}}
\newcommand{\plambdaf}{\ensuremath{\lambda\awslambda/f}}
\newcommand{\plambdafstate}{\ensuremath{\mathcal{C}_f}}

\section{Serverless Platform without History}
\label{appendix1}

This section augments the semantics of \pfname to include an execution
history for each function instance. We call this extended semantics
\plambdahist. In this semantics, every function instance
records its 1)~input request when busy and 2)~the history of internal
states, beginning with the state in which it received the request.
The syntax of \plambdahist{} extends \cref{function-semantics} as follows,
and excludes function instances without histories ($\mathbb{F}$):

\[
\begin{array}{@{}lr@{\,}c@{\,}ll}
\textrm{Execution mode with history} & m'  & \pdef & \pidle
  & \textrm{Idle} \\
                                     &    & \mid  & \pbusy(x,v)
  & \textrm{Processing request with ID $x$ and value $v$} \\
\textrm{Components}  & \mathbb{C} & \pdef &  \cdots \\
                     &            &       &\mathbb{F'}(f,m',\pnstates,y)
  &\textrm{Function instance with history}
\end{array}
\]

The semantics of \plambdahist{} extends \cref{function-semantics} to
maintain execution history:

\[
\begin{array}{@{}lr@{\,}c@{\,}ll}
\multicolumn{5}{@{\quad\quad\quad}l}{\inferrule*[Left=Cold]
  {\textrm{$\pinstid$ is fresh} \\ \precvfn_f(v, \pinitfn(f)) = \sigma}
  {\pfstatetyp \preq{f}{x}{v} \pfsteparr{} \pfstatetyp \preq{f}{x}{v}  \mathbb{F'}(f,\pbusy(x,v),[\pinitfn(f),\sigma],y)
 }} \\[1.5em]
\multicolumn{5}{@{\quad\quad\quad}l}{\inferrule*[Left=Warm]
  {\precvfn_f(v, \sigma) = \sigma'}
  {\pfstatetyp \preq{f}{x}{v} \mathbb{F'}(f,\pidle, [\sigma],y)
         \pfsteparr{}  \pfstatetyp \preq{f}{x}{v} \mathbb{F'}(f,\pbusy(x,v), [\sigma,\sigma'],y)
     }} \\[1.5em]
\multicolumn{5}{@{\quad\quad\quad}l}{\inferrule*[Left=Hidden]
  {\pstepfn_f(\sigma) = (\sigma', \varepsilon)}
  {\pfstatetyp \mathbb{F'}(f,\pbusy(x,v), \psnoc{\pnstates}{\sigma},y) 
         \pfsteparr{} \pfstatetyp\mathbb{F'}(f,\pbusy(x,v), \psnoc{\pnstates}{\sigma,\sigma'},y) 
     }} \\[1.5em]
\multicolumn{5}{@{\quad\quad}l}{\inferrule*[Left=Resp]
	{\pstepfn_f(\sigma) = (\sigma', \ptreturn{v'})}
  {\pfstatetyp \preq{f}{x}{v} \mathbb{F'}(f,\pbusy(x,v), \psnoc{\pnstates}{\sigma},y) 
    \pfsteparr{\kw{stop}(x,v')}
    \pfstatetyp 
    \mathbb{F'}(f,\pidle, [\sigma'],y) 
     \presp{x}{v'}
   }} \\[1.5em]
\multicolumn{5}{@{\quad\quad\quad}l}{\inferrule*[Left=Die]
  {\phantom{.}}
  {\pfstep{\pfstatetyp \mathbb{F'}(f,m', \pnstates,y) }{}{\pfstatetyp}}} \\[1.5em]
\end{array}
\]

Note that the \textsc{Resp} rule above and the \textsc{N-Buffer-Stop} rule
in the naive semantics ensure that all steps in the history are unobservable.

It is straightforward to establish a strong bisimulation between \pfname and
\plambdahist, since the added state does not affect execution. To do so, we
first define the bisimulation relation below. The relation requires execution histories
to be a valid history (\cref{lamlam-hist-bisim-rel}~\cref{real-hist}),
with states that are equivalent to the initial state of the function
(\cref{init-related-1,init-related-2}).

\begin{definition}[Bisimulation Relation]
\label{lamlam-hist-bisim-rel}
We define $\mathcal{C} \approx \mathcal{C}_\mathit{hist}$ as:
\begin{enumerate}
    
\item Both configurations have an identical set of requests ($\mathbb{R}$) and
responses ($\mathbb{S}$), and

\item For every function instance ID ($y$), either:

\begin{enumerate}
    
\item There is no instance with ID $y$ in either configuration,

\item Both instances $y$ are busy, i.e.:

\begin{enumerate}
    
\item $\mathbb{F}(f,\pbusy(x), \sigma_n,y)\mathbb{R}(x,v) \in \mathcal{C}$,

\item $\mathbb{F'}(f,\pbusy(x,v), [\sigma_0,\cdots,\sigma_n],y) \in \mathcal{C}_\mathit{hist}$,

\item\label{real-hist}

$\pstepfn_f(\sigma_0) = (\sigma_1, t_1), \cdots \pstepfn_f(\sigma_{n-1}) = (\sigma_n, t_n)$, 
where $t_i \ne \ptreturn{v}$ and

\item\label{init-related-1} $(\sigma_0,\pinitfn(f)) \in \mathcal{R}$.

\end{enumerate}

\item Both instances $y$ are idle, i.e.:

\begin{enumerate}
    
\item $\mathbb{F}(f,\pidle, \sigma,y)\mathbb{R}(x,v) \in \mathcal{C}$,

\item $\mathbb{F'}(f,\pidle, [\sigma],y) \in \mathcal{C}_\mathit{hist}$,

\item\label{init-related-2} $(\sigma, \pinitfn(f)) \in \mathcal{R}$.

\end{enumerate}

\end{enumerate}

\end{enumerate}

\end{definition}

The strong bisimulation theorem is straightforward.

\begin{theorem}[Bisimulation]
  \label{lam-bisim-lam-hist} If 
$\mathcal{C} \approx \mathcal{C}_\mathit{hist}$
then:

\begin{enumerate}
    
\item If $\mathcal{C}\pfsteparr{\ell}\mathcal{C'}$ then
$\mathcal{C}_\mathit{hist}\pfsteparr{\ell}\mathcal{C'}_\mathit{hist}$ and
$\mathcal{C}' \approx \mathcal{C}_\mathit{hist}'$.

\item If $\mathcal{C}_\mathit{hist}\pfsteparr{\ell}\mathcal{C'}_\mathit{hist}$
then $\mathcal{C}\pfsteparr{\ell}\mathcal{C'}$ and $\mathcal{C}' \approx
\mathcal{C}_\mathit{hist}'$.

\end{enumerate}

\end{theorem}

\begin{proof}
By case analysis on the definition of $\pfsteparr{\ell}$.
\end{proof}

\subsection{Eliminating Instance Termination}

We define \plambdadie{} as \plambdahist{} without the \textsc{Die} rule. We prove
that evaluation in \plambdadie{} is weakly bisimilar to evaluation in
\plambdahist{}. To do so, we define the bisimulation relation below that allows
\plambdadie{} to have extra function instances. Intuitively, these correspond
to instances that would have died. The latter two conditions are invariants
that ensure that the sequence of states are valid execution histories that
begin with a state that is equivalent to the initial state of the function.

\begin{definition}[Bisimulation Relation]
\label{hist-rel-die}
 We define $\mathcal{C}_\mathit{hist} \approx \plambdadiestate$ as:

\begin{enumerate}

\item Both configurations have an identical set of requests ($\mathbb{R}$) and
responses ($\mathbb{S}$),

\item The function instances of $\plambdadiestate$ are
a superset of the function instances of $\plambdahiststate$,

\item For all idle function instances in either configuration
($\mathbb{F'}(f,\pidle, [\sigma],y)$), we have $(\sigma,
\pinitfn(f)) \in \mathcal{R}$, and

\item For all busy function instances in either configuration
($\mathbb{F'}(f,\pbusy(x,v), [\sigma_0,\cdots,\sigma_n],y)$):

\begin{enumerate}

\item $\pstepfn_f(\sigma_0) = (\sigma_1, \varepsilon), \cdots \pstepfn_f(\sigma_{n-1}) = (\sigma_n, \varepsilon)$, and

\item $(\sigma_0,\pinitfn(f)) \in \mathcal{R}$.

\end{enumerate}

\end{enumerate}

\end{definition}

With this definition, we can state the weak simulation theorem:

\begin{theorem}[Weak Bisimulation]
\label{hist-bisim-die} If $\mathcal{C}_\mathit{hist} \approx
\mathcal{C}_{\mathit{hist}\backslash\mathit{die}}$ then:

\begin{enumerate}
    
\item If $\mathcal{C}_\mathit{hist}\pfsteparr{\ell}\mathcal{C'}_\mathit{hist}$
then $\mathcal{C}_\mathit{\neg\mathit{die}} \pfsteps{\ell}
\mathcal{C'}_\mathit{\neg\mathit{die}}$ and
$\mathcal{C}'_\mathit{hist} \approx
\mathcal{C}_\mathit{\neg\mathit{die}}'$; and

\item If $\mathcal{C}_\mathit{\neg die} \pfsteparr{\ell}
\mathcal{C'}_\mathit{\neg die}$ then
$\mathcal{C}_\mathit{hist}\pfsteps{\ell}\mathcal{C'}_\mathit{hist}$ and
$\mathcal{C}'_\mathit{hist} \approx
\mathcal{C}_\mathit{\neg\mathit{die}}'$.

\end{enumerate}

\end{theorem}

\begin{proof}

We first prove case~(1) of the theorem. There are two cases to consider:

\begin{enumerate}

\item Suppose $\mathcal{C}_\mathit{hist}\pfsteparr{\ell}\mathcal{C'}_\mathit{hist}$
employs the \textsc{Die} rule. Thus $\ell = \varepsilon$. In this case, 
let $\mathcal{C}_\mathit{\neg\mathit{die}}' =
\mathcal{C}_\mathit{\neg\mathit{die}}$. Moreover,
the bisimulation relation holds after the step because the function instances 
of $\mathcal{C}_\mathit{\neg\mathit{die}}$ are a strict superset
of the function instances of $\mathcal{C}'_\mathit{hist}$.

\item Suppose
$\mathcal{C}_\mathit{hist}\pfsteparr{\ell}\mathcal{C'}_\mathit{hist}$ does not
employ the \textsc{Die} rule. Since the function instances of
$\mathcal{C}_\mathit{\neg\mathit{die}}$ are a superset of the
function instances of $\mathcal{C}_\mathit{hist}$, we can take the same step on
the equivalent function ($\mathcal{C}_\mathit{\neg die}
\pfsteparr{\ell}\mathcal{C'}_\mathit{\neg die}$). The bisimulation
relation holds after the step since both $\mathcal{C'}_\mathit{hist}$ and
$\mathcal{C'}_\mathit{\neg die}$ transition the same function
instance to the same state.

\end{enumerate}

We now prove case~(2) of the theorem. There are three cases to consider:

\begin{enumerate}

\item Suppose $\mathcal{C}_\mathit{\neg die} \pfsteparr{\ell}
\mathcal{C'}_\mathit{\neg die}$ involves a function instance that
also exists in $\mathcal{C}_\mathit{hist}$. In these cases, we can take the
same step $\mathcal{C}_{\mathit{hist}} \pfsteps{\ell}
\mathcal{C'}_{\mathit{hist}}$. The bisimulation relation holds after the step
since both $\mathcal{C'}_\mathit{hist}$ and
$\mathcal{C'}_\mathit{\neg die}$ transition the same function
instance to the same state.

\item Suppose $\mathcal{C}_\mathit{\neg die} \pfsteparr{}
\mathcal{C'}_\mathit{\neg die}$, an unobservable step,
involves a function instance that
\emph{does not} exist in $\mathcal{C}_\mathit{hist}$. Since the step
is unobservable, we do not take a step, and $\mathcal{C}_\mathit{hist}' = 
\mathcal{C}_\mathit{hist}$.

\item Suppose $\mathcal{C}_\mathit{\neg die} \pfsteparr{\ell}
\mathcal{C'}_\mathit{\neg die}$ involves a function instance that
\emph{does not} exist in $\mathcal{C}_\mathit{hist}$ and $\ell$ is observable.
Let us call this function
instance $\mathbb{F'}(f,m',[\sigma_0,\cdots,\sigma_{n-1}],y)$ \emph{before}
the observable step is taken.

The bisimulation relation ensures that the sequence of states in the function
instance are a valid history of unobservable steps.
Moreover, the step is
$\pstepfn_f(\sigma_{n-1}) = (\sigma,t)$ is observable. We can construct an equivalent
sequence of states:

\begin{tikzcd}
\sigma_0 \arrow[r]\arrow[d,dash,dashed,"\in\mathcal{R}"] 
  & \sigma_1 \arrow[r,dashed]\arrow[d,dash,dashed,"\in\mathcal{R}"] 
  & \sigma_{n-1} \arrow[r,"t"]\arrow[d,dash,dashed,"\in\mathcal{R}"] 
  & \sigma_{n} \arrow[d,dash,dashed,"\in\mathcal{R}"] \\
\sigma_0' \arrow[r] 
  & \sigma_1' \arrow[r,dashed]
  & \sigma_{n-1}' \arrow[r,"t"]
  & \sigma_{n}'
\end{tikzcd}

In the diagram above:
\begin{itemize}

\item $\sigma'_0 = \pinitfn(f)$, thus $\sigma_0 \approx \sigma_0'$ by the
latter cases of the bisimulation relation.
    
\item Let $\pstepfn(\sigma_i') = (\sigma_{i+1}',t)$. Since $\sigma_i \approx
\sigma_i'$, by the safety relation, we have $\pstepfn_f(\sigma_i) =
(\sigma_{i+1}, t)$, $\pstepfn_f(\sigma_i') = (\sigma_{i+1}', t)$, and
$\sigma_{i+1} \approx \sigma'_{i+1}$

\end{itemize}

Thus, we can create a series of steps in $\mathcal{C}_\mathit{hist}$ that
begins with a \textsc{Cold} step, followed by a sequence of \textsc{Hidden} steps,
and ends with an observable step. Use the state after the final step to construct
$\mathcal{C}_\mathit{hist}'$.

\end{enumerate}

\end{proof}

\subsection{Output Buffering and Single Instance per Function}
\label{buffering-semantics}

\plambdabuf{} is a semantics for a serverless platform that has exactly one
instance per function, and a set of buffered responses that have yet to
be emitted as observations.

\[
\begin{array}{@{}lr@{\,}c@{\,}ll}
\textrm{Components}  & \mathbb{C} & \pdef &  \cdots \\
                     &            & \mid  &  \mathbb{B}(x,v)
  & \textrm{Buffered response}
\end{array}
\]

The semantics of \plambdabuf{} assumes that in the initial configuration all
functions $f$ have exactly one idle instance with state initialized to
$\pinitfn(f)$. Thus, we remove the \textsc{Cold} rule, but keep the
\textsc{Warm} rule. Moreover, we add the two following rules:

\[
\begin{array}{@{}lr@{\,}c@{\,}ll}
\multicolumn{5}{@{\quad\quad}l}{\inferrule*[Left=RespBuf]
  {\pstepfn_f(\sigma) = (\sigma', \ptreturn{v'}) \\
  \mathbb{B}(x,v') \notin \pfstatetyp}
  {\pfstatetyp \preq{f}{x}{v} 
   \mathbb{F'}(f,\pbusy(x,v), \psnoc{\pnstates}{\sigma},y) 
    \pfsteparr{}
    \pfstatetyp 
    \preq{f}{x}{v} 
    \mathbb{F'}(f,\pidle, [\sigma'],y) 
    \mathbb{B}(x,v')
   }} \\[1.5em]
\multicolumn{5}{@{\quad\quad}l}{\inferrule*[Left=BufEmit]
  {\phantom{.}}
  {\pfstatetyp \preq{f}{x}{v} \mathbb{B}(x,v') 
   \pfsteparr{\kw{stop}(x,v')}
   \pfstatetyp \presp{x}{v'}
  }} 
\end{array}
\]

Note that we \emph{do not} remove the \textsc{Resp} rule. Thus,
an instance may non-deterministically produce a response without buffering.

\begin{definition}[Bisimulation Relation] \label{die-rel-buf} We define
$\plambdadiestate \approx \plambdabufstate$ as:

\begin{enumerate}

\item Both configurations have an identical set of requests ($\mathbb{R}$) and
responses ($\mathbb{S}$);

\item\label{buf-bisim-rel-subset} For all
$\mathbb{F'}(f,\pbusy(x,v), \pnstates,y) \in \plambdabufstate$
there exists $\mathbb{F'}(f,\pbusy(x,v), \pnstates',y) \in
\plambdadiestate$ such that $(\plast{\pnstates},\plast{\pnstates'}) \in
\mathcal{R}$;

\item If $\mathbb{F'}(f,\pbusy(x,v), \pnstates,y) \in \plambdabufstate$
then for all $v'$,  $\mathbb{B}(x,v') \notin \plambdabufstate$
(every request is processed exactly once);

\item For all $\mathbb{B}(x,v) \in \plambdabufstate$, there exists a $v'$ such
that $\preq{f}{x}{v'} \in \plambdabufstate$ (all buffered responses have a
pending request);

\item\label{bisim-diebuf-idle-init} For all idle function instances in either
configuration ($\mathbb{F'}(f,\pidle, [\sigma],y)$), we have
$(\sigma, \pinitfn(f)) \in \mathcal{R}$; and

\item For all busy function instances in either configuration
($\mathbb{F'}(f,\pbusy(x,v), [\sigma_0,\cdots,\sigma_n],y)$):

\begin{enumerate}

\item $\pstepfn_f(\sigma_0) = (\sigma_1, \varepsilon), \cdots \pstepfn_f(\sigma_{n-1}) = (\sigma_n, \varepsilon)$, and

\item $(\sigma_0,\pinitfn(f)) \in \mathcal{R}$.

\end{enumerate}

\end{enumerate}

\end{definition}

\begin{theorem}[Weak Bisimulation]\label{die-bisim-buf} If $\plambdadiestate\approx
\plambdabufstate$ then:
    
\begin{enumerate}

\item If $\plambdadiestate\pfsteparr{\ell}\plambdadiestate'$ then
$\plambdabufstate\pfsteps{\ell}\plambdabufstate'$ and
$\plambdadiestate'\approx\plambdabufstate'$; and

\item If $\plambdabufstate\pfsteparr{\ell}\plambdabufstate'$ then
$\plambdadiestate\pfsteps{\ell}\plambdadiestate'$ then
$\plambdadiestate'\approx\plambdabufstate'$.

\end{enumerate}

\end{theorem}

\begin{proof}

We first prove Case~(1) of the theorem. There are four cases to consider:

\begin{enumerate}

\item Suppose a busy instance $\mathbb{F'}(f,m', \pnstates,y) \in \plambdadiestate$
takes the step $\plambdadiestate\pfsteparr{\ell}\plambdadiestate'$, and there
exists a $\mathbb{F'}(f,m', \pnstates',y) \in \plambdabufstate$. By the
bisimulation relation, $(\plast{\pnstates},\plast{\pnstates'}) \in \mathcal{R}$.
In these cases, the latter instance can take a related step
$\plambdabufstate\pfsteps{\ell}\plambdabufstate'$.

\item Suppose an idle instance 
$\mathbb{F'}(f,\pidle, \pnstates',y) \in \plambdadiestate$ causes the step
$\plambdadiestate\pfsteparr{\ell}\plambdadiestate'$, which can only
be due to the \textsc{Cold} rule. Moreover, suppose
$\mathbb{F'}(f,\pidle, \pnstates,y) \in \plambdabufstate$.
Case~\ref{bisim-diebuf-idle-init} of the bisimulation relation
ensures that $(\plast{\pnstates},\pinitfn(f)),(\plast{\pnstates'},\pinitfn(f)) \in\mathcal{R}$, thus 
$\mathbb{F'}(f,\pidle, \pnstates,y)$ steps to a related state by the safety relation.

\item Suppose the instance
$\mathbb{F'}(f,\pbusy(x,v), \pnstates,y) \in \plambdadiestate$ makes an
unobservable step,
$\plambdadiestate\pfsteparr{}\plambdadiestate'$. However,
$\mathbb{F'}(f,\pbusy(x',v'), \pnstates',y') \in \plambdabufstate$ where $x' \ne x$
(i.e., \plambdabufstate{} is processing a different request).
In these cases, let $\plambdabufstate' = \plambdabufstate$.

\item Suppose the instance
$\mathbb{F'}(f,\pbusy(x,v), \pnstates,y) \in \plambdadiestate$ makes an
observable step:
\[ \plambdadiestate\pfsteparr{\kw{stop}(x,v)}\plambdadiestate' \] 
and $\mathcal{C}\mathbb{F'}(f,\pbusy(x',v'),\pnstates',y') =\plambdabufstate$
(i.e., processing a different request).
In these cases, we can construct the sequence of states shown below:

\begin{tikzcd}
\plambdabufstate = \mathcal{C}\mathbb{F'}(f,\pbusy(x',v'),\pnstates',y')
  \arrow[d,dashed] \\
\mathcal{C}\mathbb{B}(x',v'')\mathbb{F'}(f,\pidle,[\sigma'],y')
\arrow[d] \arrow[r,dash,dashed,"\in\mathcal{R}"] 
& \pinitfn(f) \\
\mathcal{C}\mathbb{B}(x',v'')\mathbb{F'}(f,\pbusy(x,v),[\sigma',\sigma''],y')
\arrow[d,dashed] \\
\mathcal{C}\mathbb{B}(x',v'')\mathbb{F'}(f,\pbusy(x,v),\pnstates'',y')
\arrow[d,"{\textsf{stop}(x,v)}"] \arrow[r,dash,dashed,"\in\mathcal{R}"]  &
\plast{\pnstates} \\
\plambdabufstate' = \mathcal{C}\mathbb{B}(x',v'')\mathbb{F'}(f,\pidle,[\sigma'''],y') 
\arrow[r,dash,dashed,"\approx"]
& \plambdadiestate'
\end{tikzcd}

In the diagram above, the first sequence of steps terminates with 
\textsc{RespBuf}, which produces an idle instance with
state $\sigma'$, and $(\sigma',\pinitfn(f)) \in \mathcal{R}$. The
next step starts processing request $(x,v$). The next sequence of
steps terminates with state $\plast{\pnstates''}$, where
$(\plast{\pnstates''},\plast{\pnstates}) \in \mathcal{R}$ by definition
of $\mathcal{R}$ and the bisimulation relation. Thus, the final step
produces the same observation.

\end{enumerate}

We now prove case~(2) of the theorem. If a busy instance takes the step
$\plambdabufstate\pfsteparr{\ell}\plambdabufstate'$, then
case~\ref{buf-bisim-rel-subset} of the bisimulation relation ensures that a
related instance exists in $\plambdadiestate$ that can take the same step to a
related state.

\end{proof}
    
\subsection{Single Request Simulation}

We now modify the previous semantics to 1)~run exactly one function instance $f$
and 2)~only receive requests when it is idle.  Thus we delete the \textsc{Req}
rule and add the following rule:
\[
\inferrule*[Left=ReqIdle]{\textrm{$x$ is fresh}}{\pfstep{\pfstatetyp\mathbb{F'}(f,\pidle,\pnstates,y)}
{\kw{start}(f,x,v)}{\mathcal{C}\mathbb{F'}(f,\pidle,\pnstates,y)\preq{f}{x}{v}}}
\]
Let $\plambdafstate$ be a configuration of this modified semantics.

We define $x(\tightoverharp{\ell})$ as the sub-sequence of $\ell$ that only has
observations labelled $x$.

\begin{definition}[Bisimulation Relation]
We define $\plambdabufstate \approx \plambdafstate$ as:
\begin{enumerate}
  
  \item $\plambdafstate \subseteq \plambdabufstate$, and
  
  \item Any component that is in $\plambdabufstate$ and not in
  $\plambdafstate$, must involve another function $f' \ne f$.
\end{enumerate}
\end{definition}

\begin{theorem}[Weak Bisimulation]
\label{buf-bisim-f}
If $\plambdabufstate \approx \plambdafstate$ then:

\begin{enumerate}

\item If $\plambdafstate\pfsteparr{\ell}\plambdafstate'$ then
$\plambdabufstate \pfsteps{\ell} \plambdabufstate'$ and
$\plambdabufstate' \pequiv \plambdafstate'$.

\item If $\pfstep{\plambdabufstate}{\ell}{\plambdabufstate'}$ then
$\plambdafstate \pfsteps{\ell'} \plambdafstate'$, where $\ell' = x(\ell)$ and
$\plambdabufstate' \pequiv \plambdafstate'$.

\end{enumerate}

\end{theorem}

\begin{proof}

Straightforward case analysis of $\plambdafstate\pfsteparr{\ell}\plambdafstate'$
and $\pfstep{\plambdabufstate}{\ell}{\plambdabufstate'}$. In case~(1), note
that the \textsc{Req} applies in all cases that \textsc{ReqIdle} applies.

\end{proof}

\subsection{Abstract Machine Simulation}

\begin{definition}[Bisimulation Relation]
\label{bisim-rel-f-abs}
We define $\pnstateid \approx \plambdafstate$ in the obvious way.
Let $\pnstateid = \langle f,m,\pnstates,\mathcal{B} \rangle$.
\begin{enumerate}
  \item Every buffered response  $\mathbb{B}(x,v) \in \plambdafstate$ has
  a corresponding buffered response $(x,v) \in \mathcal{B}$,
  \item $\plambdafstate$ has a single function instance, with $\pnstates$,
  and mode $m' \approx m$.

\end{enumerate}
\end{definition}

\begin{theorem}[Bisimulation]
If $\pnstateid \approx \plambdafstate$:
\begin{enumerate}

\item If $\pnstep{\pnstateid}{\ell}{\pnstateid'}$ then
$\pfstep{\plambdafstate'}{\ell}{\plambdafstate'}$ and $\pnstateid' \approx
\plambdafstate'$.

\item If $\pfstep{\plambdafstate'}{\ell}{\plambdafstate'}$ then
$\pnstep{\pnstateid}{\ell}{\pnstateid'}$ and $\pnstateid' \approx
\plambdafstate'$.

\end{enumerate}

\end{theorem}

\begin{proof}
Straightforward, since there are only syntactic differences between the
two semantics.
\end{proof}

\subsection{Final Bisimulation Relation}

We now define a bisimulation relation between $\pfstatetyp$ and $\pnstateid$,
without reference to any of the intermediate relations or semantics.

\begin{definition}[Bisimulation Relation]
\label{final-bisim-relation}

$\pnstate{f}{m}{\pnstates}{\pnstopbuff} \pequiv \pfstatetyp$ is defined as:
\begin{enumerate}

\item\label{final-bisim-idle-rel}
For all $\pinst{f}{\pidle}{\sigma} \in \pfstatetyp$, $(\sigma,\pinitfn(f)) \in \mathcal{R}$;

\emph{All idle instances are in a state related to the initial state. This
should be an invariant at every intermediate relation.}

\item\label{final-bisim-busy-seq}
For all $\pinst{f}{\pbusy(x)}{\sigma_n} \in \pfstatetyp$, there exists
$\preq{f}{x}{v} \in \pfstatetyp$, and a sequence of states $\sigma_0 \cdots \sigma_n$,
such that $(\sigma_0,\pinitfn(f)) \in \mathcal{R}$,
$\precvfn_f(v,\sigma_0) = \sigma_1$, and
$\pstepfn_f(\sigma_i) = (\sigma_{i+1}, \varepsilon)$

\emph{All busy instances are in a state that is reachable from a state related
to the initial state, with a sequence of unobservable steps. This should be an
invariant at every intermediate relation.}

\item For all $(x,w) \in \pnstopbuff$ there exists $\preq{f}{x}{v}$ $\in \pfstatetyp$,
and  a sequence of states $\sigma_0 \cdots \sigma_n$, such that
$\precvfn_f(v, \pinitfn(f)) = \sigma_0$, 
$\pstepfn_f(\sigma_i) = (\sigma_{i+1}, \epsilon)$, and
$\pstepfn_f(\sigma_{n-1}) = (\sigma_n, \ptreturn{w})$

\emph{For every buffered response in $\pnstateid$, there is a corresponding
request in $\pfstatetyp$. Moreover, there exists a sequence of states, recorded
in the history, that derives this response from the request.}

\end{enumerate}

\end{definition}

\begin{theorem}

The final bisimulation relation is contained in the composition of the sequence
of bisimulation relations defined above.

\end{theorem}

\begin{proof}
Consider an element of
 $\pnstate{f}{m}{\pnstates}{\pnstopbuff} \pequiv \pfstatetyp$.

\begin{enumerate}

  \item We construct \plambdahiststate{}, where $\pfstatetyp \approx \plambdahiststate$. Most of the
  criteria of this relation are straightforward and follow directly from the definition of
  $\pfstatetyp$. However, we need to construct the histories that \plambdahist{} introduces.  
  Each busy function instance in $\plambdahiststate{}$ has a history, which is the sequence of
  states in \cref{final-bisim-busy-seq} of \cref{final-bisim-relation}. The current state in the
  history of every idle instance in $\plambdahiststate{}$ is related to its initial state, by
  \cref{final-bisim-idle-rel} of \cref{final-bisim-relation}.

  \item We construct \plambdadiestate{}, where $\plambdahiststate \approx \plambdadiestate$.
  Let $\plambdadiestate = \plambdahiststate$, and the criteria for the relation are straightforward.

  \item We construct \plambdabufstate{}, where $\plambdadiestate \approx \plambdabufstate$.
  Let $\plambdabufstate \supseteq \{ \mathbb{B}(x,v) \mid (x,v) \in \mathcal{B} \}$.
  Moreover, $\plambdabufstate$ contains exactly once instance per function,
  constructed as follows:
  
  \begin{itemize}
    
    \item A busy instance may be in $\plambdabufstate$, only if it is not processing a request that
    has already been buffered. More precisely, if $\mathbb{F'}(f,\pbusy(x),\pnstates,y) \in
    \plambdabufstate$, then $\mathbb{F'}(f,\pbusy(x),\pnstates,y) \in \plambdadiestate$ and
    $x \not\in \mathit{dom}(\mathcal{B})$.

    \item Any function $f$ may have an idle instance in $\plambdabufstate$ with initial-equivalent
    state. More precisely, if $\mathbb{F'}(f,\pidle,\pnstates,y) \in \plambdadiestate$ then
    $(\plast{\pnstates}, \pinitfn(f)) \in \mathcal{R}$.

  \end{itemize}

   \item We construct $\plambdafstate$, where $\plambdabufstate \approx \plambdafstate$. Let
   $\plambdafstate \subseteq \plambdabufstate$ be the largest subset of $\plambdabufstate$ where
   all components involve the function $f$.

   \item Finally, we construct $\pnstateid$ from $\plambdafstate$ in a
   straightforward manner, since they only have syntactic differences.

\end{enumerate}

\end{proof}

\section{Severless with a Key Value Store}

The following notation will help simplify the proofs. First, note that the parameter $f$ of $\pstepfn_f$ is clear from context as we consider each function individually, so we write $\pstepfn$ instead. In addition, let $\sigma_0 = \pinitfn (f)$ for concision. Finally, we define the following relation:

\begin{definition}[Equivalent in History]
For all $\sigma$, $\pnstates$, and $\prel$, $\sigma \in_\prel \pnstates$ holds if there exists $\sigma' \in \pnstates$
and $(\sigma,\sigma') \in \prel$.
\end{definition}

We now show some preliminary lemmas.

\begin{lemma}
\label{lemma:forward-induction}
For all $f$, $\pnstates$, $\pnstates'$, $\psigmaequiv$, $\mathcal{C}$, $x$, $v$, and $\prel$ if:

\begin{enumerate}

  \item[H1.] $\langle f,\pbusy(x),\pnstates,\pnstopbuff\rangle \pnsteps{} \langle f,\pbusy(x),\pnstates\pappend\pnstates',\pnstopbuff\rangle$, and

  \item[H2.] $\langle f,\pbusy(x),\pnstates,\pnstopbuff\rangle \pequiv C \preq{f}{x}{v} \pinst{f}{\pbusy(x)}{\psigmaequiv} \land (\psigmaequiv, \plast{\pnstates}) \in \prel$

\end{enumerate}

then there exists a $\psigmaequiv'$ such that:

\begin{enumerate}

  \item[G1.] $C \preq{f}{x}{v} \pinst{f}{\pbusy(x)}{\psigmaequiv} \pfsteps{}
   C \preq{f}{x}{v} \pinst{f}{\pbusy(x)}{\psigmaequiv'}$, and

  \item[G2.] $\pnstate{f}{\pbusy(x)}{\pnstates\pappend\pnstates'}{\pnstopbuff} \pequiv C \preq{f}{x}{v}
         \pinst{f}{\pbusy(x)}{\psigmaequiv'} \land (\psigmaequiv', \plast{\pnstates\pappend\pnstates'}) \in \prel$

\end{enumerate}

\end{lemma}

\begin{proof}

By induction on $\pnstates'$.

\noindent
\textbf{Base case.} When $\pnstates' = \pempty$, we choose $\psigmaequiv' = \psigmaequiv$, and then both sides of G1 are equal and G2
is H2.

\noindent
\textbf{Inductive step} Let $\pnstates' = \psnoc{\pnstates''}{\sigma}$. The inductive hypotheses are:

\begin{enumerate}

  \item[H3.] $C \preq{f}{x}{v} \pinst{f}{\pbusy(x)}{\psigmaequiv} \pfsteps{}
   C \preq{f}{x}{v} \pinst{f}{\pbusy(x)}{\psigmaequiv''}$, and

  \item[H4.] $\pnstate{f}{\pbusy(x)}{\pnstates\pappend\pnstates''}{\pnstopbuff} \pequiv C \preq{f}{x}{v}
         \pinst{f}{\pbusy(x)}{\psigmaequiv''} \land (\psigmaequiv'', \plast{\pnstates\pappend\pnstates''}) \in \prel$

\end{enumerate}
Substitute  $\pnstates'$ with $\psnoc{\pnstates''}{\sigma}$ in the goals to get:
\begin{enumerate}

  \item[G1'.] $C \preq{f}{x}{v} \pinst{f}{\pbusy(x)}{\psigmaequiv} \pfsteps{}
   C \preq{f}{x}{v} \pinst{f(x)}{\pbusy}{\psigmaequiv'}$, and

  \item[G2'.] $\pnstate{f}{\pbusy(x)}{\pnstates\pappend \psnoc{\pnstates''}{\sigma}}{\pnstopbuff} \pequiv C \preq{f}{x}{v}
         \pinst{f}{\pbusy}{\psigmaequiv'} \land (\psigmaequiv', \plast{\pnstates\pappend \psnoc{\pnstates''}{\sigma}}) \in \prel$

\end{enumerate}
Substitute  $\pnstates'$ with $\psnoc{\pnstates''}{\sigma}$ in $H1$ to get:
\begin{enumerate}

\item[H1'.] $\langle f,\pbusy(x),\pnstates,\pnstopbuff\rangle \pnsteps{} \langle f,\pbusy(x),\pnstates\pappend\psnoc{\pnstates''}{\sigma},\pnstopbuff\rangle$
\end{enumerate}
By definition of $~~\pnsteps{}$, the last step of $\mathit{H1}'$ is:
\begin{enumerate}

\item [H5.] $\pnstep{\langle f,\pbusy(x),\pnstates\pappend\pnstates'',\pnstopbuff\rangle}{}
            {\langle f,\pbusy(x),\pnstates\pappend\psnoc{\pnstates''}{\sigma},\pnstopbuff\rangle}$
\end{enumerate}
By definition of $\pnstep{}{}{}$, $\mathit{H5}$ must be an application of the
\textsc{Step} rule. Therefore, by hypothesis of \textsc{Step} and $\mathit{H5}$:
\begin{enumerate}

\item [H6.] $\pstepfn(\plast{\pnstates\pappend\pnstates''}) = (\sigma, \varepsilon)$

\end{enumerate}
Let $(\psigmaequiv', t) = \pstepfn{(\psigmaequiv'')}$. By $\mathit{H4}$, $\mathit{H6}$ and \cref{def:safety-relation}:
\begin{enumerate}
    \item [H7.] $(\psigmaequiv', \sigma) \in \prel \land t = \varepsilon$
\end{enumerate}
By \textsc{Hidden} and $\mathit{H7}$:
\begin{enumerate}

\item [H8.] $\pfstep{C \preq{f}{x}{v} \pinst{f}{\pbusy(x)}{\psigmaequiv''}}{}
 C \preq{f}{x}{v} \pinst{f}{\pbusy(x)}{\psigmaequiv'}$

\end{enumerate}
Finally, $\mathit{H3}$ and $\mathit{H8}$ establish goal $G1'$.

To establish goal $\mathit{G2'}$, we use the second case in the definition of $\pequiv$,
thus need to show:
\begin{enumerate}

\item[G3.] For all $\psigmaequiv_i$ such that $\pinst{f}{\pbusy(x)}{\psigmaequiv_i} \in
C\preq{f}{x}{v}\pinst{f}{\pbusy}{\psigmaequiv'}$
we have $\psigmaequiv_i \in_\prel \pnstates\pappend \psnoc{\pnstates''}{\sigma}$.
\end{enumerate}
$\mathit{G3}$ is trivial for $\pinst{f}{\pbusy}{\psigmaequiv'}$, since $(\psigmaequiv', \sigma) \in \prel$ by $\mathit{H7}$ and holds for $C$ by $\mathit{H2}$, since
the history on the left-hand side of $\mathit{H2}$ is a prefix of the history on
the left-hand side of $\mathit{G2'}$. By $\mathit{G3}$ and $\mathit{H7}$ the goal $\mathit{G2'}$ is established.

\end{proof}

\begin{lemma}
\label{lemma:three}
For all $x$, $l$, $f$, $m$, $\pnstates$, $\mathcal{C}$, and $\sigma'$, if:
\begin{enumerate}
    \item [H1.] $\forall \pinst{f}{\pbusy(x)}{\sigma'}.$ $\sigma' \in_\prel \pnstates$ 
\end{enumerate}
then
\begin{enumerate}
    \item [G1.] $\mathcal{C}\pinst{f}{\pbusy(x)}{\sigma'} \pfsteps{} \mathcal{C}\pinst{f}{\pbusy(x)}{last(\pnstates)}$ with $l$ = $\epsilon$
\end{enumerate}
\end{lemma}
\begin{proof}
By Induction on $(\sigma', \pnstates_{n-k})$ s.t. $1 \le k \le n$.
\paragraph{Base Case:} When $(\sigma', last(\pnstates)) \in \prel$ then it is trivial.
\paragraph{Induction Step:}  When $(\sigma', \pnstates_{n-k}) \in \prel$ s.t. $1 \le k < n$
then by \textsc{Hidden} $\pfstep{\mathcal{C}\pinst{f}{\pbusy(x)}{\sigma'}}{\epsilon}{\mathcal{C}\pinst{f}{\pbusy(x)}{\sigma^{n-k-1}}}$ s.t. $\kw{step}(\sigma^{n-k}, v) = (\sigma^{n-k-1}, \epsilon)$ and $(\sigma^{n-k-1}, \pnstates_{n-k-1}) \in \prel$\\
Again by \textsc{Hidden} $\pfstep{\mathcal{C}\pinst{f}{\pbusy(x)}{\sigma^{n-k-2}}}{\epsilon}{\mathcal{C}\pinst{f}{\pbusy(x)}{\sigma^{n-k-2}}}$ s.t. $\kw{step}(\sigma^{n-k-1}, v) = (\sigma^{n-k-2}, \epsilon)$ and $(\sigma^{n-k-2}, \pnstates_{n-k-2}) \in \prel$\\
Similarly, by repeated application of \textsc{Hidden}
$\pfstep{\mathcal{C}\pinst{f}{\pbusy(x)}{\sigma^{n-1}}}{\epsilon}{\mathcal{C}\pinst{f}{\pbusy(x)}{\sigma^{n}}}$ s.t. $\kw{step}(\sigma^{n-1}, v) = (\sigma^{n}, \epsilon)$ and $(\sigma^{n}, \pnstates_{n}) \in \prel$\\
This follows Goal $\mathit{G1}$
\end{proof}

\begin{definition}
\label{def:idempotence}
	Let $X$ be the set of all request IDs. A serverless function $\langle f, \sigma_0, \Sigma, \precvfn, \pstepfn \rangle$ with request ID $x \in X$ is idempotent if the following conditions hold:

\begin{enumerate}

	\item Let $\sigma_1 = \precvfn(v, \sigma_0)$.
	\item Let $(\sigma_2, t_2) = \pstepfn(\sigma_1)$, and it must be that $t_2 = \ptbeginTx$.
	\item Let $(\sigma_3, t_3) = \pstepfn(\sigma_2)$, and it must be that $t_3 = \ptread{x})$.
	\item Let $\sigma_4 = \precvfn(\pnil, \sigma_3)$.
	\item Let $\sigma_4' = \precvfn(w, \sigma_3)$, for $w \neq \pnil$.
	\item Let $(\sigma_f, t_f) = \pstepfn(\sigma_4')$, and it must be that $t_f = \ptendTx$.
	\item $\exists n \geq 4$ such that:
		\begin{enumerate}
			\item Let $(\sigma_{n+1}, t_{n+1}) = \pstepfn(\sigma_n)$ and $t_{n+1} = \ptwrite{x}{w}$.
			\item $4 \leq m < n \implies \pstepfn(\sigma_m) = (\sigma_{m+1}, t_{m+1}) \land t_{m+1} \in \{ \varepsilon, \ptread{k}, \ptwrite{k}{v'} \mid k \not \in X \}$
		\end{enumerate}
	\item $\pstepfn(\sigma_{n+1}) = (\sigma_f, \ptendTx)$
	\item Let $(\sigma_{ff}, t_{ff}) = \pstepfn(\sigma_f)$, and it must be that $t_{ff} = \ptreturn{w}$.

\end{enumerate}
\end{definition}


\begin{lemma} \label{lemma:forward-induction-store}
For all $\widehat{\mathcal{C}}_0$, $\pkv$, $f$, $x$, $\pnstates$, $\sigma$ if:
\begin{enumerate}
    \item [H1.] $\mathcal{C} = \pds{\pkv}{\plfree} \widehat{\mathcal{C}}_0$
    \item [H2.] $\pnunlocked, \pnstate{f}{\pbusy(x)}{\psnoc{\pnstates}{\sigma}}{\pnstopbuff} \pequiv \mathcal{C}$
    \item [H3.] $\exists \pinsty{f}{\pbusy(x)}{\psigmaequivold}{y'} \in \widehat{\mathcal{C}}_0$
\end{enumerate}

Then there exists $\widehat{\mathcal{C}}_1$ such that

\begin{enumerate}
    \item [G1.] $\mathcal{C} \pfsteps{\tightoverharp{\ell}} \pds{\pkv}{\plfree} \pinsty{f}{\pbusy(x)}{\psigmaequiv'}{y} \widehat{\mathcal{C}}_1 \land x(\tightoverharp{\ell}) = \varepsilon$
	\item [G2.] $\pnunlocked, \pnstate{f}{\pbusy(x)}{\psnoc{\pnstates}{\sigma}}{\pnstopbuff} \pequiv \pds{\pkv}{\plfree} \pinsty{f}{\pbusy(x)}{\psigmaequiv'}{y} \widehat{\mathcal{C}}_1$
	\item [G3.] $(\sigma, \psigmaequiv') \in \prel$
\end{enumerate}

\end{lemma}

\begin{proof}
By H1, H2 and \cref{def:bisimulation-relation}:
\begin{enumerate}
	\item [H4.] For all $\pinsty{f}{\pbusy(x)}{\sigma''}{y} \in \widehat{\mathcal{C}}_0$, $\sigma'' \in_\prel \psnoc{\pnstates}{\sigma}$
\end{enumerate}

Let $\psnoc{\pnstates}{\sigma} = [\sigma_0, \sigma_1, \cdots, \sigma_n], \sigma_n = \sigma$. Choose $i$ and $\psigmaequiv$ such that $i$ is the greatest index in $\psnoc{\pnstates}{\sigma}$ such that there exists $\pinsty{f}{\pbusy(x)}{\psigmaequiv}{y} \in \widehat{\mathcal{C}}_0$ with $(\sigma_i, \psigmaequiv) \in \prel$. Let $\widehat{\mathcal{C}}_0' = \widehat{\mathcal{C}}_0 \setminus \pinsty{f}{\pbusy(x)}{\psigmaequiv}{y}$, and let $r = n - i$. Note that $0 \leq r \leq n$. We now perform induction on $r$:

\paragraph{Base Case: $r = 0$:}
By hypothesis $r = 0$, so $i = n$ and $(\sigma, \psigmaequiv) \in \prel$. Then, G1--G3 are established trivially.

\paragraph{Inductive Case: $0 < r \leq n$:}
\begin{enumerate}
	\item [H5.] $(\psigmaequiv, \sigma_i) \in \prel$
	\item [H6.] $0 \leq i < n$
\end{enumerate}
For notation, let $\pnstateid_j = \pnstate{f}{\pbusy(x)}{[\sigma_0, \sigma_1, \cdots, \sigma_j]}{\pnstopbuff}$. Since the only possible \pnname lock state transitions prior to a $\kw{stop}(x, v)$ are $\pnstep{\pnunlocked}{}{\pnunlocked}$, $\pnstep{\pnunlocked}{}{\pnlock{M}{\pnstates'}}$ and $\pnstep{\pnlock{M}{\pnstates'}}{}{\pncommit{v}}$ the \pnname steps prior to $\pnunlocked, \pnstate{f}{\pbusy(x)}{\psnoc{\pnstates}{\sigma}}{\pnstopbuff}$ must be of the form:
\[
	\pnunlocked, \pnstateid_0 \pnsteparrd{}{\kw{start}(x, v)} \pnunlocked, \pnstateid_1 \pnsteparrd{}{\varepsilon} \cdots \pnsteparrd{}{\varepsilon} \pnunlocked, \pnstateid_i \pnsteparrd{}{\varepsilon} \pnunlocked, \pnstateid_{i+1} \pnsteparrd{}{\varepsilon} \cdots \pnsteparrd{}{\varepsilon} \pnunlocked, \pnstateid_n = \pnunlocked, \pnstate{f}{\pbusy(x)}{\psnoc{\pnstates}{\sigma}}{\pnstopbuff}
\]
In particular, we know that:
\begin{enumerate}
	\item [H7.] $\pnstep{\pnunlocked, \pnstateid_i}{}{\pnunlocked, \pnstateid_{i+1}}$
\end{enumerate}
Since $\pnstateid_i = \pnstate{f}{\pbusy(x)}{[\sigma_0, \cdots, \sigma_i]}{\pnstopbuff}$ is in the busy state, and the first step above is a $\kw{start}(x, v)$, $i \neq 0$. By case analysis on \pnname step rules, this must be a \textsc{\pnNameRule-Step} rule invocation. Thus:
\begin{enumerate}
	\item [H9.] $\pstepfn{(\sigma_i)} = (\sigma_{i+1}, \varepsilon)$
\end{enumerate}
Let $(\psigmaequiv', t) = \pstepfn{(\psigmaequiv)}$. By H5, H9 and \cref{def:safety-relation}:
\begin{enumerate}
	\item [H10.] $(\psigmaequiv', \sigma_{i+1}) \in \prel$
	\item [H11.] $t = \varepsilon$
\end{enumerate}
By H10, H11 and \textsc{Hidden}:
\begin{enumerate}
	\item [H12.] $\pfstep{
			\pinsty{f}{\pbusy(x)}{\psigmaequiv}{y} \pds{\pkv}{\plfree} \widehat{\mathcal{C}}_0'
		}{}{
			\pinsty{f}{\pbusy(x)}{\psigmaequiv'}{y} \pds{\pkv}{\plfree} \widehat{\mathcal{C}}_0'
		}$
\end{enumerate}
By H4 and H10:
\begin{enumerate}
	\item [H13.] $\pnunlocked, \pnstateid_{i+1} \pequiv \pinsty{f}{\pbusy(x)}{\psigmaequiv'}{y} \pds{\pkv}{\plfree} \widehat{\mathcal{C}}_0'$
\end{enumerate}
By H10 we now have $r' = r - 1$, so by the inductive hypothesis we have that there exists $\widehat{\mathcal{C}}_1$ such that:
\begin{enumerate}
	\item [H14.] $\pinsty{f}{\pbusy(x)}{\psigmaequiv'}{y} \pds{\pkv}{\plfree} \widehat{\mathcal{C}}_0' \pfsteps{} \pinsty{f}{\pbusy(x)}{\psigmaequiv''}{y} \pds{\pkv}{\plfree} \widehat{\mathcal{C}}_1$
	\item [H15.] $\pnunlocked, \pnstate{f}{\pbusy(x)}{\psnoc{\pnstates}{\sigma}}{\pnstopbuff} \pequiv \pinsty{f}{\pbusy(x)}{\psigmaequiv''}{y} \pds{\pkv}{\plfree} \widehat{\mathcal{C}}_1$
	\item [H16.] $(\sigma, \psigmaequiv'') \in \prel$
\end{enumerate}
Then, G1--G3 are established by H14, H15 and H16.
\end{proof}


\begin{lemma}[Classification of \pnname semantics traces]\label{lemma:forward-classification}
Suppose:
\begin{enumerate}
	\item [H1.] $\langle f, \sigma_0, \Sigma, \precvfn, \pstepfn \rangle$ satisfies \cref{def:idempotence}.
\end{enumerate}
For notation, define the following, using $n$ chosen in \cref{def:idempotence}:
\begin{align*}
	\pnstateid_0 &= \pnstate{f}{\pidle}{\sigma_0}{\pnstopbuff} \\
	\pnstateid_i &= \pnstate{f}{\pbusy(x)}{[\sigma_0, \sigma_1, \cdots, \sigma_i]}{\pnstopbuff}, \textrm{ for } i \leq n+1 \\
	\pnstateid_{n+2} &= \pnstate{f}{\pbusy(x)}{[\sigma_0, \sigma_1, \cdots, \sigma_{n+1}, \sigma_f]}{\pnstopbuff} \\
	\pnstateid_{n+3} &= \pnstate{f}{\pidle}{[\sigma_{ff}]}{\pnstopbuff \cup \{ (x, w) \}} \\
	\pnstateid_{n+4} &= \pnstate{f}{\pidle}{[\sigma_{ff}]}{\pnstopbuff} \\
	\pnstateid_4' &= \pnstate{f}{\pbusy(x)}{[\sigma_0, \sigma_1, \cdots, \sigma_3, \sigma_4']}{\pnstopbuff} \\
	\pnstateid_5' &= \pnstate{f}{\pbusy(x)}{[\sigma_0, \sigma_1, \cdots, \sigma_3, \sigma_4', \sigma_f]}{\pnstopbuff} \\
	\pnstateid_6' &= \pnstate{f}{\pidle}{[\sigma_{ff}]}{\pnstopbuff \cup \{ (x, w) \}} \\
	\pnstateid_7' &= \pnstate{f}{\pidle}{[\sigma_{ff}]}{\pnstopbuff} \\
	L &= \pnlock{M}{[\sigma_0, \sigma_1]} \\
	L' &= \pnlock{M'}{[\sigma_0, \sigma_1]} \\
	L'' &= \pnlock{M'[x \mapsto w]}{[\sigma_0, \sigma_1]} \\
	C &= \pncommit{w}
\end{align*}

Then there are exactly two classes of program traces in the \pnname semantics. Either the \pnname program trace is exactly:
\begin{align*}
	\pnunlocked, \pnstateid_0 &\pnsteparrd{Rule \textsc{Start}}{\kw{start}(x, v)}
	\pnunlocked, \pnstateid_1 \pnsteparrd{Rule \textsc{BeginTx}}{}
	L, \pnstateid_2 \pnsteparrd{Rule \textsc{Read}}{}
	L, \pnstateid_4' \pnsteparrd{Rule \textsc{EndTx}}{}
	C, \pnstateid_5' \\ &\pnsteparrd{Rule \textsc{\pnNameRule-Buffer-Stop}}{}
	C, \pnstateid_6' \pnsteparrd{Rule \textsc{\pnNameRule-Emit-Stop}}{\kw{stop}(x, w)}
	C, \pnstateid_7'
\end{align*}
or otherwise it has the following form:
\begin{align*}
	\pnunlocked, \pnstateid_0 &\pnsteparrd{Rule \textsc{Start}}{\kw{start}(x, v)}
	\pnunlocked, \pnstateid_1 \pnsteparrd{Rule \textsc{BeginTx}}{}
	L, \pnstateid_2 \pnsteparrd{Rule \textsc{Read}}{}
	L, \pnstateid_4 \pnstepsarrd{Rules \textsc{Read}, \textsc{Write}, or \textsc{Step}}{} \;
	L', \pnstateid_n \\ &\pnsteparrd{Rule \textsc{Write}}{}
	L'', \pnstateid_{n+1} \pnsteparrd{Rule \textsc{EndTx}}{}
	C, \pnstateid_{n+2} \pnsteparrd{Rule \textsc{\pnNameRule-Buffer-Stop}}{}
	C, \pnstateid_{n+3} \pnsteparrd{Rule \textsc{\pnNameRule-Emit-Stop}}{\kw{stop}(x, w)}
	C, \pnstateid_{n+4}
\end{align*}

\end{lemma}

\begin{proof}
There are two cases:

\paragraph{Case $\mathbf{M(x) = \pnil}$:}

By \cref{def:idempotence}, step 1 and the \textsc{Start} rule:
\begin{enumerate}
	\item [H2.] $\pnunlocked, \pnstateid_0 \pnsteparrd{Rule \textsc{Start}}{\kw{start}(x, v)} \pnunlocked, \pnstateid_1$
\end{enumerate}

By \cref{def:idempotence}, step 2 and the \textsc{BeginTx} rule:
\begin{enumerate}
	\item [H3.] $\pnunlocked, \pnstateid_1 \pnsteparrd{Rule \textsc{BeginTx}}{} L, \pnstateid_2$
\end{enumerate}

By \cref{def:idempotence}, steps 3 and 4, the case hypothesis, and the \textsc{Read} rule:
\begin{enumerate}
	\item [H4.] $L, \pnstateid_2 \pnsteparrd{Rule \textsc{Read}}{} L, \pnstateid_4$
\end{enumerate}

By \cref{def:idempotence}, steps 7a and 7b, and the \textsc{Step}, \textsc{Read} and \textsc{Write} rules:
\begin{enumerate}
	\item [H5.] $L, \pnstateid_4 \pnstepsarrd{Rules \textsc{Read}, \textsc{Write}, or \textsc{Step}}{} \;
	L', \pnstateid_n \pnsteparrd{Rule \textsc{Write}}{} L'', \pnstateid_{n+1}$
\end{enumerate}

By \cref{def:idempotence}, step 8, and the \textsc{EndTx} rule:
\begin{enumerate}
	\item [H6.] $L'', \pnstateid_{n+1} \pnsteparrd{Rule \textsc{EndTx}}{}
	C, \pnstateid_{n+2}$
\end{enumerate}

By \cref{def:idempotence}, step 9, and the \textsc{\pnNameRule-Buffer-Stop} and \textsc{\pnNameRule-Emit-Stop} rules:
\begin{enumerate}
	\item [H7.] $C, \pnstateid_{n+2} \pnsteparrd{Rule \textsc{\pnNameRule-Buffer-Stop}}{}
	C, \pnstateid_{n+3} \pnsteparrd{Rule \textsc{\pnNameRule-Emit-Stop}}{\kw{stop}(x, w)} C, \pnstateid_{n+4}$
\end{enumerate}

By H2, H3, H4, H5, H6 and H7, the second \pnname program trace form is established.

Case complete.

\paragraph{Case $\mathbf{M(x) \neq \pnil}$:}

By \cref{def:idempotence}, step 1 and the \textsc{Start} rule:
\begin{enumerate}
	\item [H2.] $\pnunlocked, \pnstateid_0 \pnsteparrd{Rule \textsc{Start}}{\kw{start}(x, v)} \pnunlocked, \pnstateid_1$
\end{enumerate}

By \cref{def:idempotence}, step 2 and the \textsc{BeginTx} rule:
\begin{enumerate}
	\item [H3.] $\pnunlocked, \pnstateid_1 \pnsteparrd{Rule \textsc{BeginTx}}{} L, \pnstateid_2$
\end{enumerate}

By \cref{def:idempotence}, steps 3 and 5, the case hypothesis, and the \textsc{Read} rule:
\begin{enumerate}
	\item [H4.] $L, \pnstateid_2 \pnsteparrd{Rule \textsc{Read}}{} L, \pnstateid_4'$
\end{enumerate}

By \cref{def:idempotence}, step 6, and the \textsc{EndTx} rule:
\begin{enumerate}
	\item [H5.] $L, \pnstateid_4' \pnsteparrd{Rule \textsc{EndTx}}{}
	C, \pnstateid_5'$
\end{enumerate}

By \cref{def:idempotence}, step 9, and the \textsc{\pnNameRule-Buffer-Stop} and \textsc{\pnNameRule-Emit-Stop} rules:
\begin{enumerate}
	\item [H6.] $C, \pnstateid_{5}' \pnsteparrd{Rule \textsc{\pnNameRule-Buffer-Stop}}{}
	C, \pnstateid_{6}' \pnsteparrd{Rule \textsc{\pnNameRule-Emit-Stop}}{\kw{stop}(x, w)} C, \pnstateid_{7}'$
\end{enumerate}

By H2, H3, H4, H5, and H6, the first \pnname program trace form is established.

Case complete.
\end{proof}


\begin{theorem}[Forward direction of weak bisimulation with the key-value store]
\label{ext-bisim-fwd}
For all $\pnlockid$, $\pnlockid'$, $\pnstateid$, $\pnstateid'$, $\mathcal{C}$, $\ell$, $\Sigma$, $\precvfn$, $\pstepfn$ if:
\begin{enumerate}
    \item [H1.] $\pnstep{\pnlockid,\pnstateid}{\ell}{\pnlockid',\pnstateid'}$
    \item [H2.] $\pnlockid,\pnstateid \pequiv \mathcal{C}$
    \item [H3.] $\langle f, \sigma_0, \Sigma, \precvfn, \pstepfn,\pnstopbuff \rangle$ satisfies \cref{def:idempotence}
\end{enumerate}

Then there exists $\tightoverharp{\ell_1}$, $\tightoverharp{\ell_2}$, $\mathcal{C}_i$, $\mathcal{C}_{i+1}$ and $\mathcal{C}'$ such that

\begin{enumerate}
    \item [G1.] $\pnlockid',\pnstateid' \pequiv \mathcal{C'}$
    \item [G2.] $\mathcal{C} \pfsteps{\tightoverharp{\ell_1}} \mathcal{C}_i \land x(\tightoverharp{\ell_1}) = \varepsilon$
    \item [G3.] $\pfstep{\mathcal{C}_i}{\ell}{\mathcal{C}_{i+1}}$
    \item [G4.] $\mathcal{C}_{i+1} \pfsteps{\tightoverharp{\ell_2}} \mathcal{C}' \land x(\tightoverharp{\ell_2}) = \varepsilon$
\end{enumerate}

\end{theorem}

\begin{proof}
By case analysis of the relation
$\pnstep{\pnlockid,\pnstateid}{\ell}{\pnlockid',\pnstateid'}$.

\paragraph{Case \textsc{Start}:}
By definition and by \cref{lemma:forward-classification}:
\begin{enumerate}
	\item [H1'.] $\pnlockid, \pnstateid = \pnstep{
			\pnunlocked , \pnstate{f}{\pidle}{\sigma_0}{\pnstopbuff}
		}{\kw{start}(x, v)}{
			\pnunlocked, \pnstate{f}{\pbusy(x)}{[\sigma_0, \sigma_1]}{\pnstopbuff}
		} = \pnlockid', \pnstateid'$
\end{enumerate}

Choose $x$ to be fresh. By \textsc{Req}:
\begin{enumerate}
	\item [H4.] $\pfstep{\mathcal{C}}{\kw{start}(x, v)}{\preq{f}{x}{v} \mathcal{C}}$
\end{enumerate}

Choose $\mathcal{C}' = \preq{f}{x}{v} \mathcal{C}$, and G1 is established. By H4, G2--G4 are established.

Case complete.

\paragraph{Case {\sc Step}:}
By definition:
\begin{enumerate}
	\item [H1'.] $\pnlockid, \pnstateid = \pnstep{\pnlockid, \pnstate{f}{\pbusy(x)}{\psnoc{\pnstates}{\sigma}}{\pnstopbuff}}{}{\pnlockid, \pnstate{f}{\pbusy(x)}{\psnoc{\pnstates}{\sigma, \sigma'}}{\pnstopbuff}} = \pnlockid', \pnstateid'$
\end{enumerate}

By hypothesis of \textsc{Step} and $H1'$:
\begin{enumerate}
	\item [H4.] $\pstepfn(\sigma) = (\sigma', \varepsilon)$
\end{enumerate}

By \cref{lemma:forward-classification}:
\begin{enumerate}
	\item [H5.] There exists an $m$ with $4 \leq m < n$ such that $\pnlockid, \pnstateid = \pnlock{M''}{[\sigma_0, \sigma_1]}, \pnstateid_m$ and $\pnlockid', \pnstateid' = \pnlock{M''}{[\sigma_0, \sigma_1]}, \pnstateid_{m+1}$. Thus, $\sigma = \sigma_m, \sigma' = \sigma_{m+1}$.
\end{enumerate}

By H5:
\begin{enumerate}
	\item [H6.] $\pnlockid, \pnstateid = \pnstep{
			\pnlock{M''}{[\sigma_0, \sigma_1]},
			\pnstate{f}{\pbusy(x)}{[\sigma_0, \cdots, \sigma_m]}
		}{}{ \newline
			\pnlock{M''}{[\sigma_0, \sigma_1]},
			\pnstate{f}{\pbusy(x)}{[\sigma_0, \cdots, \sigma_{m+1}]}
		} = \pnlockid', \pnstateid'$
\end{enumerate}

By H2 and H6:
\begin{enumerate}
	\item [H7.] $C = \pds{M_0}{\plowned{y}{M''}} \pinsty{f}{\psigmaequiv_m}{\pbusy(x)}{y} \widehat{\mathcal{C}}_0$
	\item [H8.] $(\sigma_m, \psigmaequiv_m) \in \prel$
\end{enumerate}

By H4, H5, H8 and \cref{def:safety-relation}:
\begin{enumerate}
	\item [H9.] $\pstepfn(\psigmaequiv_m) = (\psigmaequiv_{m+1}, \varepsilon)$
	\item [H10.] $(\sigma_{m+1}, \psigmaequiv_{m+1}) \in \prel$
\end{enumerate}

By H9 and \textsc{Hidden}:
\begin{enumerate}
	\item [H11.] $\pfstep{
			\pds{M_0}{\plowned{y}{M''}} \pinsty{f}{\psigmaequiv_m}{\pbusy(x)}{y} \widehat{\mathcal{C}}_0
		}{}{
			\pds{M_0}{\plowned{y}{M''}} \pinsty{f}{\psigmaequiv_{m+1}}{\pbusy(x)}{y} \widehat{\mathcal{C}}_0
		}$
\end{enumerate}

Let $\mathcal{C}' = \pds{M_0}{\plowned{y}{M''}} \pinsty{f}{\psigmaequiv_{m+1}}{\pbusy(x)}{y} \widehat{\mathcal{C}}_0$. By H10, G1 is satisfied. By H11, G2--G4 are satisfied.

Case complete.

\paragraph{Case \textsc{\pnNameRule-Buffer-Stop}:}
By definition:
\begin{enumerate}
	\item [H1'.] $\pnlockid, \pnstateid = \pnstep{\pnlockid, \pnstate{f}{\pbusy(x)}{\psnoc{\pnstates}{\sigma} }{\pnstopbuff}}{}{\pnlockid, \pnstate{f}{\pidle}{[\sigma_N']}{\pnstopbuff \cup \{(x, w)\}}} = \pnlockid', \pnstateid'$
\end{enumerate}

By $H2$, $\preq{f}{x}{v} \in \pfstatetyp$, so then $\pnlockid', \pnstateid' \pequiv \pfstatetyp$. Thus goals G1--G4 are established by a zero step of $\pfstatetyp$.

%
%

\paragraph{Case \textsc{\pnNameRule-Emit-Stop}:}
By definition:
\begin{enumerate}
	\item [H1'.] $\pnlockid, \pnstateid = \pnstep{\pnlockid, \pnstate{f}{\pidle}{\pnstates}{\pnstopbuff \cup \{(x, w)\}}}{\kw{stop}(x, w)}{\pnlockid, \pnstate{f}{\pidle}{\pnstates}{\pnstopbuff}} = \pnlockid', \pnstateid'$
\end{enumerate}

Case complete.


By \cref{lemma:forward-classification} there are two cases:

\textbf{Sub-Case 1:} $\pncommit{w}, \pnstateid_6' \pnsteparrd{Rule \textsc{\pnNameRule-Emit-Stop}}{\kw{stop}(x, w)}
	\pncommit{w}, \pnstateid_7'$.

\textbf{Sub-Case 2:} $\pncommit{w}, \pnstateid_{n+3} \pnsteparrd{Rule \textsc{\pnNameRule-Emit-Stop}}{\kw{stop}(x, w)}
	\pncommit{w}, \pnstateid_{n+4}$.

Note that both sub-cases end in the same state as given by \cref{lemma:forward-classification}, $\pnstateid_{7}' = \pnstateid_{n+4}'$. In both sub-cases we have:
\begin{enumerate}
	\item [H2'.] $\mathcal{C} = \pds{\pkv}{\plfree} \preq{f}{x}{v} \widehat{\mathcal{C}}_0$
	\item [H5.] $w \neq \pnil$
	\item [H6.] $M(x) = w$
\end{enumerate}

By H3:
\begin{enumerate}
	\item [H7.] $\precvfn(v, \sigma_0) = \sigma_1$
\end{enumerate}

Choose $y$ to be fresh, and by H7 and \textsc{Cold}:
\begin{enumerate}
	\item [H8.] $\pfstep{\pds{\pkv}{\plfree} \preq{f}{x}{v} \widehat{\mathcal{C}}_0}{}{\pds{\pkv}{\plfree} \preq{f}{x}{v} \pinsty{f}{\sigma_1}{\pbusy(x)}{y} \widehat{\mathcal{C}}_0}$
\end{enumerate}

By H3:
\begin{enumerate}
	\item [H9.] $\pstepfn(\sigma_1) = (\sigma_2, \ptbeginTx)$
\end{enumerate}

By H9 and \textsc{BeginTx}:
\begin{enumerate}
	\item [H10.] $\pfstep{
		\pds{\pkv}{\plfree} \preq{f}{x}{v} \pinsty{f}{\sigma_1}{\pbusy(x)}{y} \widehat{\mathcal{C}}_0
	}{}{
		\pds{\pkv}{\plowned{y}{M}} \preq{f}{x}{v} \pinsty{f}{\sigma_2}{\pbusy(x)}{y} \widehat{\mathcal{C}}_0
	}$
\end{enumerate}

By H3:
\begin{enumerate}
	\item [H11.] $\pstepfn(\sigma_2) = (\sigma_3, \ptread{x} )$
	\item [H12.] $\precvfn(w, \sigma_3) = \sigma_4'$, for $w \neq \pnil$
\end{enumerate}

By H5, H6, H11, H12 and \textsc{Read}:
\begin{enumerate}
	\item [H13.] $\pfstep{
		\pds{\pkv}{\plowned{y}{M}} \preq{f}{x}{v} \pinsty{f}{\sigma_2}{\pbusy(x)}{y} \widehat{\mathcal{C}}_0
	}{}{ \newline
		\pds{\pkv}{\plowned{y}{M}} \preq{f}{x}{v} \pinsty{f}{\sigma_4'}{\pbusy(x)}{y} \widehat{\mathcal{C}}_0
	}$
\end{enumerate}

By H3:
\begin{enumerate}
	\item [H14.] $\pstepfn(\sigma_4') = (\sigma_f, \ptendTx )$
\end{enumerate}

By H14 and \textsc{EndTx}:
\begin{enumerate}
	\item [H15.] $\pfstep{
		\pds{\pkv}{\plowned{y}{M}} \preq{f}{x}{v} \pinsty{f}{\sigma_4'}{\pbusy(x)}{y} \widehat{\mathcal{C}}_0
	}{}{
		\pds{\pkv}{\plfree} \preq{f}{x}{v} \pinsty{f}{\sigma_f}{\pbusy(x)}{y} \widehat{\mathcal{C}}_0
	}$
\end{enumerate}

By H3:
\begin{enumerate}
	\item [H16.] $\pstepfn(\sigma_f) = (\sigma_{ff}, \ptreturn{w} )$
\end{enumerate}

By H16 and \textsc{Resp}:
\begin{enumerate}
	\item [H17.] $\pfstep{
		\pds{\pkv}{\plfree} \preq{f}{x}{v} \pinsty{f}{\sigma_f}{\pbusy(x)}{y} \widehat{\mathcal{C}}_0
	}{\kw{stop}(x, w)}{
		\pds{\pkv}{\plfree} \preq{f}{x}{v} \pinsty{f}{\sigma_{ff}}{\pidle}{y}\presp{x}{w} \widehat{\mathcal{C}}_0
	}$
\end{enumerate}

Let $\mathcal{C}_{i+1} = \mathcal{C}' = \pds{\pkv}{\plfree} \preq{f}{x}{v} \pinsty{f}{\sigma_{ff}}{\pidle}{y}\presp{x}{w} \widehat{\mathcal{C}}_0$. By H6, G1 is satisfied.

Let $\mathcal{C}_i = \pds{\pkv}{\plfree} \preq{f}{x}{v} \pinsty{f}{\sigma_f}{\pbusy(x)}{y} \widehat{\mathcal{C}}_0$. By H8, H10, H13, H15 and H17, G2--G4 are satisfied.

Case complete.

\paragraph{Case \textsc{Read}:}
By definition:
\begin{enumerate}
    \item [H1'.]
    $\pnstep{\pnlock{\pkv}{\pnstates'}, \pnstate{f}{\pbusy(x)}{\psnoc{\pnstates}{\sigma}}{\pnstopbuff}}
        {}
        {\pnlock{\pkv}{\pnstates'}, \pnstate{f}{\pbusy(x)}{\psnoc{\pnstates}{\sigma,\sigma',\sigma''}}{\pnstopbuff}}$
    \item [H2'.] $\mathcal{C} = \pds{\pkv'}{\plowned{y}{\pkv}} \pinsty{f}{\psigmaequiv}{\pbusy{(x)}}{y} \widehat{\mathcal{C}}_0$ and
    $\langle f,\pbusy(x),\psnoc{\pnstates}{\sigma} \rangle \pequiv \pinsty{f}{\psigmaequiv}{\pbusy{(x)}}{y} \widehat{\mathcal{C}}_0$ and $(\sigma, \psigmaequiv) \in \prel$.
\end{enumerate}

By hypothesis of \textsc{Read} and $H1'$:
\begin{enumerate}
	\item [H4.] $\pstepfn{(\sigma)} = (\sigma', \ptread{k})$
	\item [H5.] $\precvfn{(\pkv(k), \sigma')} = \sigma''$
\end{enumerate}

By H4, $H2'$ and \cref{def:safety-relation}:
\begin{enumerate}
	\item [H6.] $\pstepfn{(\psigmaequiv)} = (\psigmaequiv', \ptread{k}) \land (\sigma', \psigmaequiv') \in \prel$.
\end{enumerate}

By H5, H6 and \cref{def:safety-relation}:
\begin{enumerate}
	\item [H7.] $\precvfn{(\pkv(k), \psigmaequiv')} = \psigmaequiv'' \land (\sigma'', \psigmaequiv'') \in \prel$
\end{enumerate}

By the \textsc{Read} rule, H6 and H7:
\begin{enumerate}
	\item [H8.] $\pfstep{
		\pds{\pkv'}{\plowned{y}{\pkv}} \pinsty{f}{\psigmaequiv}{\pbusy{(x)}}{y} \widehat{\mathcal{C}}_0}{}{
		\pds{\pkv'}{\plowned{y}{\pkv}} \pinsty{f}{\psigmaequiv''}{\pbusy{(x)}}{y} \widehat{\mathcal{C}}_0
		}$
\end{enumerate}

Let $\mathcal{C}' = \pds{\pkv'}{\plowned{y}{\pkv}} \pinsty{f}{\psigmaequiv''}{\pbusy{(x)}}{y} \widehat{\mathcal{C}}_0$. Goals G1--G4 follow by H8.

Case complete.

\paragraph{Case \textsc{Write}:}
By definition:
\begin{enumerate}
    \item [H1'.]
    $\pnstep{\pnlock{\pkv}{\pnstates'}, \pnstate{f}{\pbusy(x)}{\psnoc{\pnstates}{\sigma}}{\pnstopbuff}}
        {}
        {\pnlock{\pkv'}{\pnstates'}, \pnstate{f}{\pbusy(x)}{\psnoc{\pnstates}{\sigma,\sigma'}}{\pnstopbuff}}$
    \item [H2'.] $\mathcal{C} = \pds{\pkv''}{\plowned{y}{\pkv}} \pinsty{f}{\psigmaequiv}{\pbusy{(x)}}{y} \widehat{\mathcal{C}}_0$ and
    $\langle f,\pbusy(x),\psnoc{\pnstates}{\sigma},\pnstopbuff \rangle \pequiv \pinsty{f}{\psigmaequiv}{\pbusy{(x)}}{y} \widehat{\mathcal{C}}_0$ and $(\sigma, \psigmaequiv) \in \prel$.
\end{enumerate}

By hypothesis of \textsc{Write} and $H1'$:
\begin{enumerate}
	\item [H4.] $\pstepfn{(\sigma)} = (\sigma', \ptwrite{k}{v})$
	\item [H5.] $\pkv' = \pkv [ k \mapsto v ]$
\end{enumerate}

By H4, $H2'$ and \cref{def:safety-relation}:
\begin{enumerate}
	\item [H6.] $\pstepfn{(\psigmaequiv)} = (\psigmaequiv', \ptwrite{k}{v}) \land (\sigma', \psigmaequiv') \in \prel$.
\end{enumerate}

By the \textsc{Write} rule, H5 and H6:
\begin{enumerate}
	\item [H7.] $\pfstep{
		\pds{\pkv''}{\plowned{y}{\pkv}} \pinsty{f}{\psigmaequiv}{\pbusy{(x)}}{y} \widehat{\mathcal{C}}_0}{}{
		\pds{\pkv''}{\plowned{y}{\pkv'}} \pinsty{f}{\psigmaequiv'}{\pbusy{(x)}}{y} \widehat{\mathcal{C}}_0
		}$
\end{enumerate}

Let $\mathcal{C}' = \pds{\pkv''}{\plowned{y}{\pkv'}} \pinsty{f}{\psigmaequiv'}{\pbusy{(x)}}{y} \widehat{\mathcal{C}}_0$. Goals G1--G4 follow by H7.

Case complete.

\paragraph{Case \textsc{BeginTx}:}
By definition:
\begin{enumerate}
    \item [H1'.]
    $\pnstep{\pnunlocked,\langle f,\pbusy(x),\psnoc{\pnstates}{\sigma},\pnstopbuff \rangle}
        {}
        {\pnlock{\pkv}{\psnoc{\pnstates}{\sigma}},\langle f,\pbusy(x),\psnoc{\pnstates}{\sigma,\sigma'},\pnstopbuff \rangle}$
    \item[H2'.] $\mathcal{C} = \pds{M}{\plfree}\widehat{\mathcal{C}}_0$ and
    $\langle f,\pbusy(x),\psnoc{\pnstates}{\sigma},\pnstopbuff \rangle \pequiv \widehat{\mathcal{C}}_0$.
\end{enumerate}

There are two sub-cases:

\textbf{Sub-Case 1:} $\exists \pinsty{f}{\pbusy(x)}{\sigma''}{y} \in \widehat{\mathcal{C}}_0$ for some $\sigma''$ and $y$.

By \cref{lemma:forward-induction-store}:
\begin{enumerate}
    \item[H4A.] $\pds{M}{\plfree} \widehat{\mathcal{C}}_0 \pfsteps{\varepsilon} \pds{M}{\plfree}\pinsty{f}{\pbusy(x)}{\psigmaequiv}{y}\widehat{C}_1$.
    \item[H4B.] $\pnunlocked, \pnstate{f}{\pbusy(x)}{\psnoc{\pnstates}{\sigma}}{\pnstopbuff} \pequiv \mathcal{C}$.
    \item[H4C.] $(\sigma, \psigmaequiv) \in \prel$.
\end{enumerate}

\textbf{Sub-Case 2:} $\neg \exists \pinsty{f}{\pbusy(x)}{\sigma''}{y} \in \widehat{\mathcal{C}}_0$ for some $\sigma''$ and $y$. By \textsc{Cold} and \cref{lemma:forward-induction-store}:
\begin{enumerate}
    \item[H4A.] $\pds{M}{\plfree} \widehat{\mathcal{C}}_0 \pfsteparr{} \pds{M}{\plfree} \pinsty{f}{\pbusy(x)}{\sigma''}{y} \widehat{\mathcal{C}}_0' \pfsteps{\varepsilon} \pds{M}{\plfree}\pinsty{f}{\pbusy(x)}{\psigmaequiv}{y}\widehat{C}_1$.
    \item[H4B.] $\pnunlocked, \pnstate{f}{\pbusy(x)}{\psnoc{\pnstates}{\sigma}}{\pnstopbuff} \pequiv \mathcal{C}$.
    \item[H4C.] $(\sigma, \psigmaequiv) \in \prel$.
\end{enumerate}

By hypothesis of \textsc{BeginTx} and $\mathit{H1}'$:

\begin{enumerate}
    \item[H5.]$\pstepfn{(\sigma)} = (\sigma', \ptbeginTx)$
\end{enumerate}

By \textsc{BeginTx}, $\mathit{H4C}$ (either sub-case), $\mathit{H5}$ and \cref{def:safety-relation}:

\begin{enumerate}
  \item[H6.]
$\pfstep{\pds{\pkv}{\plfree}\pinsty{f}{\pbusy(x)}{\psigmaequiv}{y} \widehat{C}_1}{}
{\pds{\pkv}{\plowned{y}{\pkv}}\pinsty{f}{\pbusy(x)}{\psigmaequiv'}{y} \widehat{C}_1} \land (\sigma', \psigmaequiv') \in \prel$
\end{enumerate}

Let $\mathcal{C}' = \pds{\pkv}{\plowned{y}{\pkv}}\widehat{C}_1\pinsty{f}{\pbusy(x)}{\psigmaequiv'}{y}$.
Therefore, goals G2--G4 follow by H4A (either sub-case) and H6. Goal G1 follows from the
right-hand side of H1' and $\mathcal{C}'$.

Case complete.

\paragraph{Case \textsc{EndTx}:}
By definition:
\begin{enumerate}
    \item [H1'.]
    $\pnstep{\pnlock{\pkv}{\pnstates'},\langle f,\pbusy(x),\psnoc{\pnstates}{\sigma} \rangle}
             {}
             {\pncommit{\pkv(x)} ,\langle f,\pbusy(x),\psnoc{\pnstates}{\sigma,\sigma'} \rangle}$
    \item [H2'.] $\mathcal{C} = \pds{\pkv'}{\plowned{y}{\pkv}} \pinsty{f}{\psigmaequiv}{\pbusy{(x)}}{y} \widehat{\mathcal{C}}_0$ and
    $\langle f,\pbusy(x),\psnoc{\pnstates}{\sigma} \rangle \pequiv \pinsty{f}{\psigmaequiv}{\pbusy{(x)}}{y} \widehat{\mathcal{C}}_0$ and $(\sigma, \psigmaequiv) \in \prel$.
\end{enumerate}

By hypothesis of \textsc{EndTx} and $\mathit{H1}'$:
\begin{enumerate}
	\item [H4.] $\pstepfn{(\sigma)} = (\sigma', \ptendTx)$
\end{enumerate}

By $H2'$, $H4$ and \cref{def:safety-relation}:
\begin{enumerate}
	\item [H5.] $\pstepfn{(\psigmaequiv)} = (\psigmaequiv', \ptendTx)$
	\item [H6.] $(\sigma', \psigmaequiv') \in \prel$
\end{enumerate}

By $H5$ and the \textsc{EndTx} rule:
\begin{enumerate}
	\item [H7.] $\pds{\pkv'}{\plowned{y}{\pkv}} \pinst{f}{\psigmaequiv}{\pbusy{(x)}} \widehat{\mathcal{C}}_0 \pfsteparr{}
         \pds{\pkv}{\plfree} \pinst{f}{\psigmaequiv'}{\pbusy{(x)}}  \widehat{\mathcal{C}}_0$
\end{enumerate}

Let $\mathcal{C}' = \pds{\pkv}{\plfree} \pinst{f}{\psigmaequiv'}{\pbusy{(x)}}  \widehat{\mathcal{C}}_0$. Therefore, goals G2--G4 follow by H7. Goal G1 follows by H6 and because $\pkv(x) = \pkv(x)$.

Case complete.

\paragraph{Case \textsc{Rollback}:}
By definition:
\begin{enumerate}
	\item [H1'.] $\pnstep{\pnlock{\pkv}{\pnstates'},\langle f,\pbusy(x), \psnoc{\pnstates}{\sigma} \rangle}
                {}
                {\cdot,\langle f,\pbusy(x),\pnstates' \rangle}$
    \item [H2'.] $\mathcal{C} = \pds{\pkv'}{\plowned{y}{\pkv}} \pinsty{f}{\psigmaequiv}{\pbusy{(x)}}{y} \widehat{\mathcal{C}}_0$ and
    $\langle f,\pbusy(x),\psnoc{\pnstates}{\sigma} \rangle \pequiv \pinsty{f}{\psigmaequiv}{\pbusy{(x)}}{y} \widehat{\mathcal{C}}_0$ and $(\sigma, \psigmaequiv) \in \prel$.
\end{enumerate}

Since each $y$ is globally unique:
\begin{enumerate}
	\item [H4.] There does not exist $\pinsty{f}{\sigma''}{\pbusy{(x')}}{y} \in \widehat{\mathcal{C}}_0$, for any $\sigma'', x' \neq x$.
\end{enumerate}

By the \textsc{Die} rule:
\begin{enumerate}
	\item [H5.] $\pfstep{\pds{\pkv'}{\plowned{y}{\pkv}} \pinsty{f}{\psigmaequiv}{\pbusy{(x)}}{y} \widehat{\mathcal{C}}_0}	{}	{\pds{\pkv'}{\plowned{y}{\pkv}}\widehat{\mathcal{C}}_0}$
\end{enumerate}

By $H5$ and the \textsc{DropTx} rule:
\begin{enumerate}
	\item [H6.] $\pfstep{\pds{\pkv'}{\plowned{y}{\pkv}}\widehat{\mathcal{C}}_0}	{}	{\pds{\pkv'}{\plfree}\widehat{\mathcal{C}}_0}$
\end{enumerate}

Let $\mathcal{C}' = \pds{\pkv'}{\plfree}\widehat{\mathcal{C}}_0$. Goals G2--G4 are established by H5 and H6. Since $\langle f,\pbusy(x),\psnoc{\pnstates}{\sigma} \rangle \pequiv \pinsty{f}{\psigmaequiv}{\pbusy{(x)}}{y} \widehat{\mathcal{C}}_0$, there exist no other function instances in $\widehat{\mathcal{C}}_0$ with histories after the \textsc{BeginTx} and $\pnstates'$ contains history only up until the \textsc{BeginTx}, we have that $\langle f,\pbusy(x), \pnstates' \rangle \pequiv \widehat{\mathcal{C}}_0$. Thus, we conclude G1, that $\pds{\pkv'}{\plfree}\widehat{\mathcal{C}}_0 \pequiv \pnunlocked, \langle f,\pbusy(x),\pnstates' \rangle$.

Case complete.

\end{proof}

\begin{theorem}[Backward direction of weak bisimulation with key-value store]
\label{ext-bisim-back}
For all $\pnlockid$, $\pnstateid$, $\mathcal{C}$, $\mathcal{C}'$, $\ell$,
$\precvfn$, $\pstepfn$, if:
\begin{enumerate}
  \item [H1.] $\pfstep{\pds{\pkv_1}{L_1}\mathcal{C}}{\ell}{\pds{\pkv_2}{L_2}\mathcal{C}'}$
  \item [H2.] $\pnlockid$, $\pnstateid \pequiv \mathcal{C}$
  \item [H3.] $\langle f, \sigma_0, \Sigma, \precvfn, \pstepfn \rangle$ satisfies \cref{def:idempotence}
\end{enumerate}

Then there exists $\pnlockid'$ and $\pnstateid'$ such that
\begin{enumerate}
  \item [G1.]  $\pnlockid'$, $\pnstateid' \pequiv \mathcal{C}$
  \item [G2.] $\pnstep{\pnlockid, \pnstateid}{\ell}{\pnlockid', \pnstateid'}$
\end{enumerate}

\end{theorem}

\begin{proof} By case analysis for the relation $\pfstep{\pds{\pkv_1}{L_1}\mathcal{C}}{\ell}{\pds{\pkv_2}{L_2}\mathcal{C}'}$
\paragraph{Case \textsc{NewReq}:} By definition:
\begin{enumerate}
  \item [H4.] $\pkv_1 = \pkv_2 = M$
  \item [H5.] $L_1 = L_2 = .$
  \item [H6.] $\mathcal{C} = \mathcal{C}_0$
  \item [H7.] $\mathcal{C}'= \mathcal{C}_0 \preq{f}{x}{v} \pds{\pkv}{\plfree}$
  \item [H8.] $\kw{start}(x) = v$
  \item [H9.] $x$ is fresh
\end{enumerate}
By $\mathit{H9}$, $\preq{f}{x}{v} \not\in C$. By definition of $\pequiv$ and
$\mathit{H2}$,
\begin{enumerate}
  \item [H10.] $\exists \pnstates \pnstateid = \langle f, \pidle, \pnstates \rangle$
\end{enumerate}
Let $\kw{recv} (v, \sigma_0) = \sigma'$. By \textsc{START}, $\mathit{H9}$ and $\mathit{H4}$:
\begin{enumerate}
  \item [H11.]  $\pnstep{\pnunlocked, \pnstate{f}{\pidle}{\pnstates}{\pnstopbuff}}{\kw{start}(x, v)}{\pnunlocked, \pnstate{f}{\pbusy(x)}{\left[\sigma_0, \sigma' \right]}{\pnstopbuff}}$
\end{enumerate}
By \cref{def:extended-bisimulation-rel},
\begin{enumerate}
  \item [H12.] $\pnstate{f}{\pbusy(x)}{\left[\sigma_0, \sigma' \right]}{\pnstopbuff} \pequiv \mathcal{C'}$
\end{enumerate}
Let $\pnstateid' = \pnstate{f}{\pbusy(x)}{\left[\sigma_0, \sigma' \right]}{\pnstopbuff}$
and $\pnlockid' = \pnunlocked$, thus $\mathit{H12}$ and $\mathit{H11}$ establish goal
$\mathit{G1}$ and $\mathit{G2}$.

\paragraph{Case \textsc{Cold}:} By definition:
\begin{enumerate}
  \item [H4.] $\pkv_1 = \pkv_2 = M$
  \item [H5.] $L_1 = L_2 = .$
  \item [H6.] $\mathcal{C} = \mathcal{C}_0 \preq{f}{x}{v}$
  \item [H7.] $\mathcal{C}'= \mathcal{C}_0 \preq{f}{x}{v}\pinst{f}{\pbusy(x)}{\sigma'}$
  \item [H8.] $\kw{recv}(\sigma_0, v) = \sigma'$
\end{enumerate}
Since, $\pnstate{f}{m}{\pnstates}{\pnstopbuff} \pequiv \mathcal{C}_0\preq{f}{x}{v}$,
\begin{enumerate}
  \item [H9.] $m$ = $\pbusy(x)$
  \item [H10.] $\forall \pinst{f}{\pbusy(x)}{\sigma'}.\sigma' \in_{\prel} \pnstates$
\end{enumerate}
Let $\pnstateid' = \pnstate{f}{\pbusy(x)}{\pnstates}{\pnstopbuff}$, hence,
\begin{enumerate}
    \item [H11.] $\pnstep{\pnstate{f}{\pbusy(x)}{\pnstates}{\pnstopbuff}}{0}{\pnstate{f}{\pbusy(x)}{\pnstates}{\pnstopbuff}}$
\end{enumerate}
Since, $\sigma_0 \in_\prel \pnstates$, there exists $\sigma_0' \in \pnstates$, such that
\begin{enumerate}
  \item [H12.] $(\sigma_0', \sigma_0) \in \prel$.
\end{enumerate}
By $\mathit{H12}$,
\begin{enumerate}
  \item [H13.] $\kw{recv}(\sigma_0', v) = \sigma$, s.t. $\sigma \in \pnstates$ and $(\sigma, \sigma') \in \prel$
\end{enumerate}
By $\mathit{H13}$ and \cref{def:extended-bisimulation-rel}:
\begin{enumerate}
    \item [H14.] $\pnstate{f}{\pbusy(x)}{\pnstates}{\pnstopbuff} \pequiv {\mathcal{C}_0\preq{f}{x}{v}\pinst{f}{\pbusy(x)}{\sigma}}$
\end{enumerate}
Finally, $\mathit{H14}$ and $\mathit{H11}$ establish goal $\mathit{G1}$ and $\mathit{G2}$.

\paragraph{Case \textsc{Warm}:} By definition
\begin{enumerate}
  \item [H4.] $\pkv_1 = \pkv_2 = M$
  \item [H5.] $L_1 = L_2 = .$
  \item [H6.] $\mathcal{C} = \mathcal{C}_0 \preq{f}{x}{v} \pinst{f}{\pidle}{\sigma}$
  \item [H7.] $\mathcal{C}'= \mathcal{C}_0 \preq{f}{x}{v}\pinst{f}{\pbusy(x)}{\sigma'}$
  \item [H8.] $\kw{recv}(\sigma, v) = \sigma'$
\end{enumerate}
Since, $\pnstate{f}{m}{\pnstates}{\pnstopbuff} \pequiv \mathcal{C}_0\preq{f}{x}{v}\pinst{f}{\pidle}{\sigma}$,
\begin{enumerate}
  \item [H8.] $m$ = $\pbusy(x)$
  \item [H9.] $\forall \pinst{f}{\pbusy(x)}{\sigma''}.\sigma'' \in_{\prel} \pnstates$
\end{enumerate}
Let $\pnstateid' = \pnstate{f}{\pbusy(x)}{\pnstates}{\pnstopbuff}$, hence,
\begin{enumerate}
    \item [H10.] $\pnstep{\pnstate{f}{\pbusy(x)}{\pnstates}{\pnstopbuff}}{0}{\pnstate{f}{\pbusy(x)}{\pnstates}{\pnstopbuff}}$
\end{enumerate}
$\mathbb{F}$ can transition from $\pbusy$ to $\pidle$ only after
a $\textsc{RESPONSE}$ rule, and a $\textsc{WARM}$ rule can be applied only when
$\mathbb{F}$ is in $\pidle$ state. Hence, $(\sigma, \sigma_0) \in \prel$.\\
Since $\sigma \in_\prel \pnstates$, it holds that $\sigma_0 \in \pnstates$ and $(\sigma_0, \sigma) \in \prel$. Then:
\begin{enumerate}
  \item [H13.] $\kw{recv}(\sigma_0, v) = \sigma''$, s.t. $\sigma'' \in \pnstates$ and $(\sigma'', \sigma') \in \prel$
\end{enumerate}
By $\mathit{H13}$ and \cref{def:extended-bisimulation-rel}:
\begin{enumerate}
    \item [H14.] $\pnstate{f}{\pbusy(x)}{\pnstates}{\pnstopbuff} \pequiv {\mathcal{C}_0\preq{f}{x}{v}\pinst{f}{\pbusy(x)}{\sigma'}}$
\end{enumerate}
Finally, $\mathit{H10}$ and $\mathit{H14}$ establish goal $\mathit{G1}$ and $\mathit{G2}$.
\paragraph{Case \textsc{Hidden}:} By definition
\begin{enumerate}
  \item [H4.] $\pkv_1 = \pkv_2 = M$
  \item [H5.] $L_1 = L_2 = .$
  \item [H6.] $\mathcal{C} = \mathcal{C}_0 \preq{f}{x}{v} \pinst{f}{\pidle}{\sigma}$
  \item [H7.] $\mathcal{C}'= \mathcal{C}_0 \preq{f}{x}{v}\pinst{f}{\pbusy(x)}{\sigma'}$
\end{enumerate}
\begin{enumerate}
    \item [H8.] $\pnstate{f}{\pbusy(x)}{
    \pnstates}{\pnstopbuff} \pequiv \mathcal{C}_0\pinst{f}{\pbusy(x)}{\sigma}$ such that $\forall \pinst{f}{\pbusy(x)}{\sigma} \in C$. $\sigma \in_\prel \pnstates$
\end{enumerate}
By \cref{lemma:three}
\begin{enumerate}
    \item [H9.] $\pfstep{\mathcal{C}_0\pinst{f}{\pbusy(x)}{\sigma}}{}{\mathcal{C}_0\pinst{f}{\pbusy(x)}{\sigma''}}$ such that $(\sigma'',last(\pnstates)) \in \prel$
\end{enumerate}
By $\mathit{H1}$ and \textsc{Hidden}
\begin{enumerate}
    \item [H10.] $\pfstep{\mathcal{C}\pinst{f}{\pbusy(x)}{\sigma''}}{} {\mathcal{C}\pinst{f}{\pbusy(x)}{\sigma'}}$
\end{enumerate}
By \textsc{\pnNameRule-Step}
\begin{enumerate}
    \item [H11.] $\pnstep{\pnlockid, \pnstate{f}{\pbusy(x)}{\pnstates}{\pnstopbuff}}{}{\pnlockid, \pnstate{f}{\pbusy(x)}{\pnstates + [\sigma']}{\pnstopbuff}}$
\end{enumerate}
Let $\pnlockid' = \pnlockid$ and $\pnstateid' = \pnstate{f}{\pbusy(x)}{\pnstates + [\sigma']}{\pnstopbuff}$.
Finally, $\mathit{0}$ and $\mathit{H11}$ establish goal $\mathit{G1}$ and $\mathit{G2}$.
\paragraph{Case \textsc{Respond}:} By definition
\begin{enumerate}
  \item [H4.] $\pkv_1 = \pkv_2 = M$
  \item [H5.] $L_2 = \plfree$
  \item [H6.] $\mathcal{C} = \mathcal{C}_0 \preq{f}{x}{v} \pinst{f}{\pbusy(x)}{\sigma_C}$
  \item [H7.] $\mathcal{C}'= \mathcal{C}_0 \pinst{f}{\pbusy(x)}{\sigma_C'}\presp{x}{v}$
  \item [H8.] $\kw{step(\sigma_C)} = (\sigma_C', \ptreturn{v})$
\end{enumerate}

By \cref{def:extended-bisimulation-rel}:
\begin{enumerate}
	\item [H9.] $\pnlockid = \pncommit{M(x)}$
	\item [H10.] $\pnstateid = \pnstate{f}{\pbusy(x)}{\pnstates}{\pnstopbuff}$
\end{enumerate}

Let $\sigma_N = \plast{\pnstates}$. By H3 satisfying \cref{def:idempotence}, $\pstepfn_f(\sigma_C)$ can produce $\ptreturn{v}$
only when $L_1 = \plfree$ and $(\sigma_N, \sigma_C) \in \prel$.

%

By \textsc{\pnNameRule-Buffer-Stop}:
\begin{enumerate}
  \item [H13.] $\pnstep{\pncommit{M(x)}, \pnstate{f}{\pbusy(x)}{\psnoc{\pnstates}{\sigma_N}}{\pnstopbuff}}{}{\pnunlocked, \pnstate{f}{\pidle}{[\sigma_N']}{\pnstopbuff \cup \{(x, v))\}}}$ with $\pstepfn_f(\sigma_N) = (\sigma_N', \ptreturn{v})$
\end{enumerate}

By \textsc{\pnNameRule-Emit-Stop}:
\begin{enumerate}
  \item [H14.] $\pnstep{\pncommit{M(x)}, \pnstate{f}{\pidle}{[\sigma_N']}{\pnstopbuff}\cup \{(x, v))\}}{\kw{stop(x, v)}}{\pnunlocked, \pnstate{f}{\pidle}{[\sigma_N']}{\pnstopbuff}}$
\end{enumerate}
Let $\pnstateid' = \pnstate{f}{\pidle}{[\sigma_N']}{\pnstopbuff}$. By \cref{def:extended-bisimulation-rel} and $\mathit{H14}$ we get $\mathit{G1}$ and $\mathit{G2}$.
\paragraph{Case \textsc{Die}:} By definition
\begin{enumerate}
  \item [H4.] $\mathcal{C} = \mathcal{C}_0 \pinst{f}{m}{\sigma}$
  \item [H5.] $\mathcal{C'} = \mathcal{C}_0$
\end{enumerate}
Sub-Case $m = \pidle$\\
  By definition
  \begin{enumerate}
    \item [H6.] $L_1 = L_2 = \plfree$
    \item [H7.] $\pnlockid' = \pnunlocked$
  \end{enumerate}
  Let $\pnstateid' = \pnstate{f}{\pidle}{\sigma}{\pnstopbuff}$.
  By \cref{def:extended-bisimulation-rel}
  and $\mathit{H7}$ we get $\mathit{G1}$ and $\mathit{G2}$.\\
Sub-Case Complete\\
Sub-Case $m = \pbusy(x)$: By \cref{def:extended-bisimulation-rel}
  \begin{enumerate}
    \item [H9.] $L_1 = \plowned{y}{M}$
    \item [H10.] $\pnlockid = \pnlock{M}{\pnstates}$
    \item [H11.] $\pnlockid, \pnstate{f}{\pbusy(x)}{\psnoc{\pnstates}{\sigma}}{\pnstopbuff} \pequiv \pds{M}{\plowned{y}{M}} \mathcal{C}$
  \end{enumerate}
  Let $\mathcal{C}_0 = \mathcal{C}_1 \preq{f}{x}{v}$.
  Let $\pnstateid' = \pnstate{f}{\pbusy(x)}{\psnoc{\pnstates}{\sigma}}{\pnstopbuff}$ and
  $\pnlockid' = \pnlock{M}{\pnstates}$. Hence, by definition of $\pequiv$
  we get $\mathit{G1}$ and $\mathit{G2}$.
Sub-Case Complete
\paragraph{Case \textsc{Read}:} By definition:
\begin{enumerate}
  \item [H4.] $\pkv_1 = \pkv_2 = M$
  \item [H5.] $L_1 = L_2 = \plowned{y}{M'}$
  \item [H6.] $\mathcal{C} = \mathcal{C}_0 \pinst{f}{\sigma}{\pbusy{(x)}} \pds{\pkv}{\plowned{y}{\pkv'}}$
  \item [H7.] $\mathcal{C}'= \mathcal{C}_0 \pinst{f}{\sigma''}{\pbusy{(x)}} \pds{\pkv}{\plowned{y}{\pkv'}}$
  \item [H8.] $\pstepfn(\sigma) = (\sigma', \ptread{k})$
  \item [H9.] $\precvfn(\pkv'(k), \sigma') = \sigma''$
\end{enumerate}
By \textsc{\pnNameRule-Read}:
\begin{enumerate}
  \item [H8.] $\pnlockid = \pnlock{\pkv}{\pnstates'}$
  \item [H9.] $\pnstateid = \langle f,\pbusy(x),\psnoc{\pnstates}{\sigma} ,\pnstopbuff\rangle$
  \item [H10.] $\pnlockid' = \pnlock{\pkv}{\pnstates'}$
  \item [H11.] $\pnstateid' = \langle f,\pbusy(x),\psnoc{\pnstates}{\sigma, \sigma', \sigma''}, \pnstopbuff\rangle$
\end{enumerate}
By \cref{def:extended-bisimulation-rel}, $\mathit{H10}$, and $\mathit{H11}$, we have,
$\mathit{G1}$ and $\mathit{G2}$.

\paragraph{Case \textsc{Write}:} By definition:
\begin{enumerate}
  \item [H4.] $\pkv_1 = M$
  \item [H5.] $\pkv_2 = M''$
  \item [H6.] $L_1 = \plowned{y}{M'}$
  \item [H7.] $L_2 = \plowned{y}{M''}$
  \item [H8.] $\mathcal{C} = \mathcal{C}_0 \pinst{f}{\sigma}{\pbusy{(x)}} \pds{\pkv}{\plowned{y}{\pkv'}}$
  \item [H9.] $\mathcal{C}'= \mathcal{C}_0 \pinst{f}{\sigma'}{\pbusy{(x)}} \pds{\pkv}{\plowned{y}{\pkv''}}$
  \item [H10.] $\pstepfn(\sigma) = (\sigma', \ptwrite{k}{v})$
  \item [H11.] $\pkv'' = \pkv'[k \mapsto v]$
\end{enumerate}
By \textsc{\pnNameRule-Write}:
\begin{enumerate}
  \item [H8.] $\pnlockid = \pnlock{\pkv'}{\pnstates'}$
  \item [H9.] $\pnstateid = \langle f,\pbusy(x),\psnoc{\pnstates}{\sigma},\pnstopbuff \rangle$
  \item [H10.] $\pnlockid' = \pnlock{\pkv''}{\pnstates'}$
  \item [H11.] $\pnstateid' = \langle f,\pbusy(x),\psnoc{\pnstates}{\sigma, \sigma'},\pnstopbuff \rangle$
\end{enumerate}
By \cref{def:extended-bisimulation-rel}, $\mathit{H10}$, and $\mathit{H11}$, we have,
$\mathit{G1}$ and $\mathit{G2}$.

\paragraph{Case \textsc{BeginTx}:} By definition:
\begin{enumerate}
  \item [H4.] $\pkv_1 = M$
  \item [H5.] $\pkv_2 = M$
  \item [H6.] $L_1 = \plfree$
  \item [H7.] $L_2 = \plowned{y}{M}$
  \item [H8.] $\mathcal{C} = \mathcal{C}_0 \pinst{f}{\sigma}{\pbusy{(x)}} \pds{\pkv}{\plfree}$
  \item [H9.] $\mathcal{C}'= \mathcal{C}_0 \pinst{f}{\sigma'}{\pbusy{(x)}} \pds{\pkv}{\plowned{y}{\pkv}}$
  \item [H10.] $\pstepfn(\sigma) = (\sigma', \ptbeginTx)$
\end{enumerate}
By \textsc{\pnNameRule-BeginTx}:
\begin{enumerate}
  \item [H11.] $\pnlockid = \pnunlocked$
  \item [H12.] $\pnstateid = \langle f,\pbusy(x),\psnoc{\pnstates}{\sigma},\pnstopbuff \rangle$
  \item [H13.] $\pnlockid' = \pnlock{\pkv}{\psnoc{\pnstates}{\sigma}}$
  \item [H14.] $\pnstateid' = \langle f,\pbusy(x),\psnoc{\pnstates}{\sigma, \sigma'},\pnstopbuff \rangle$
\end{enumerate}
By \cref{def:extended-bisimulation-rel}, $\mathit{H13}$, and $\mathit{H14}$, we have,
$\mathit{G1}$ and $\mathit{G2}$.

\paragraph{Case \textsc{EndTx}:} By definition:
\begin{enumerate}
  \item [H4.] $\pkv_1 = M$
  \item [H5.] $\pkv_2 = M'$
  \item [H6.] $L_1 = \plowned{y}{M'}$
  \item [H7.] $L_2 = \plfree$
  \item [H8.] $\mathcal{C} = \mathcal{C}_0 \pinst{f}{\sigma}{\pbusy{(x)}} \pds{\pkv}{\plfree}$
  \item [H9.] $\mathcal{C}'= \mathcal{C}_0 \pinst{f}{\sigma'}{\pbusy{(x)}} \pds{\pkv}{\plowned{y}{\pkv}}$
  \item [H10.] $\pstepfn(\sigma) = (\sigma', \ptendTx)$
\end{enumerate}
By \textsc{\pnNameRule-EndTx}:
\begin{enumerate}
  \item [H11.] $\pnlockid = \pnlock{\pkv'}{\pnstates'}$
  \item [H12.] $\pnstateid = \langle f,\pbusy(x),\psnoc{\pnstates}{\sigma},\pnstopbuff \rangle$
  \item [H13.] $\pnlockid' = \pncommit{\pkv (x)}$
  \item [H14.] $\pnstateid' = \langle f,\pbusy(x),\psnoc{\pnstates}{\sigma, \sigma'},\pnstopbuff \rangle$
\end{enumerate}
By \cref{def:extended-bisimulation-rel} $\mathit{H13}$, and $\mathit{H14}$, we have,
$\mathit{G1}$ and $\mathit{G2}$.

\paragraph{Case \textsc{DropTx}:} By definition:
\begin{enumerate}
  \item [H4.] $\pkv_1 = M'$
  \item [H5.] $\pkv_2 = M$
  \item [H6.] $L_1 = \plowned{y}{M'}$
  \item [H7.] $L_2 = \plfree$
  \item [H8.] $\mathcal{C} = \mathcal{C}_0 \pinst{f}{\sigma}{\pbusy{(x)}} \pds{\pkv}{\plfree}$
  \item [H9.] $\mathcal{C}'= \mathcal{C}_0 \pinst{f}{\sigma'}{\pbusy{(x)}} \pds{\pkv}{\plowned{y}{\pkv}}$
\end{enumerate}
By \textsc{\pnNameRule-Rollback}:
\begin{enumerate}
  \item [H10.] $\pnlockid = \pnlock{\pkv'}{\pnstates'}$
  \item [H11.] $\pnstateid = \langle f,\pbusy(x),\pnstates,\pnstopbuff \rangle$
  \item [H12.] $\pnlockid' = \pnunlocked$
  \item [H13.] $\pnstateid' = \langle f,\pbusy(x),\pnstates',\pnstopbuff \rangle$
\end{enumerate}
By \cref{def:extended-bisimulation-rel}, $\mathit{H12}$, and $\mathit{H13}$, we have,
$\mathit{G1}$ and $\mathit{G2}$.
\end{proof}

\fi

\end{document}